\documentclass[11pt,reqno]{amsart}
\usepackage[top=2cm, bottom=2cm, left=2.5cm, right=2.5cm]{geometry}
\usepackage[pagebackref]{hyperref}

\hypersetup{
  colorlinks=true,
  linkcolor=blue!60!black,
  citecolor=blue!60!black,
  urlcolor=blue!60!black
}

\usepackage{booktabs}
\usepackage{siunitx}  

\usepackage{graphicx}
\usepackage{microtype}
\usepackage{amsmath,amssymb,amsthm,bm,amsfonts,mathrsfs,bbm}
\usepackage{aliascnt}
\usepackage{xspace}
\usepackage{pgf,tikz}
\usepackage{comment}
\usepackage{xcolor}
\usepackage{multirow}
\usepackage{array}
\usepackage{makecell}      
\usepackage{amssymb}       
\usepackage{ragged2e}      
\usepackage{array}
\usepackage{tikz-cd}
\usepackage{makecell}

\newcommand{\yes}{\textcolor{green!55!black}{\checkmark}}
\newcommand{\no}{\textcolor{red!70!black}{\ensuremath{\times}}}
\newcommand{\nodet}{\textcolor{black!45}{---}}

\usepackage{bigstrut}
\usepackage{braket}
\usepackage{color}
\usepackage[numbers]{natbib}
\usepackage{tabularx}
\usepackage[table]{xcolor}
\usepackage{multirow}
\usepackage{mathtools}
\usepackage{float}
\usepackage{xcolor,colortbl}
\usepackage{physics}
\usepackage{amsmath}
\usepackage{color}
\usepackage[capitalize]{cleveref}
\usepackage[justification=justified, format=plain]{subcaption}
\usepackage[justification=raggedright]{caption}

\newcommand{\be}{\begin{equation}}
\newcommand{\ee}{\end{equation}}
\newcommand{\ba}{\begin{eqnarray}}
\newcommand{\ea}{\end{eqnarray}}
\newcommand{\proj}[1]{\ket{#1}\bra{#1}}

\newcommand{\etal}{{\it{et al.}}\xspace}

\newcommand{\N}{\mathbb{N}}

\newcommand{\R}{\mathbb{R}}
\newcommand{\C}{\mathbb{C}}

\def\d{{\rm d}}

\def\>{\rangle}
\def\<{\langle}

\newcommand{\Rn}{\operatorname{\mathsf{R}}}

\newtheorem{theorem}{Theorem}[section]
\newaliascnt{proposition}{theorem}
\newtheorem{proposition}[proposition]{Proposition}
\aliascntresetthe{proposition}
\newaliascnt{lemma}{theorem}
\newtheorem{lemma}[lemma]{Lemma}
\aliascntresetthe{lemma}
\newaliascnt{corollary}{theorem}
\newtheorem{corollary}[corollary]{Corollary}
\aliascntresetthe{corollary}
\newaliascnt{definition}{theorem}
\newtheorem{definition}[definition]{Definition}
\aliascntresetthe{definition}
\newaliascnt{remark}{theorem}
\newtheorem{remark}[remark]{Remark}
\aliascntresetthe{remark}
\newaliascnt{example}{theorem}
\newtheorem{example}[example]{Example}
\aliascntresetthe{example}
\newaliascnt{conjecture}{theorem}

\aliascntresetthe{conjecture}
\newaliascnt{question}{theorem}

\aliascntresetthe{question}

\DeclareMathOperator{\Sep}{Sep}
\DeclareMathOperator{\PPT}{PPT}
\DeclareMathOperator{\cone}{cone}
\DeclareMathOperator{\conv}{conv}
\DeclareMathOperator{\spanop}{span}
\DeclareMathOperator{\diag}{diag}

\newcommand\ineqall{\ensuremath{\mathsf{Ineq}^{\rm all}}}
\newcommand\momall{\ensuremath{\mathsf{Mom}^{\rm all}}}
\newcommand\ineqpsd{\ensuremath{\mathsf{Ineq}^{\rm psd}}}
\newcommand\mompsd{\ensuremath{\mathsf{Mom}^{\rm psd}}}
\newcommand\ineqsep{\ensuremath{\mathsf{Ineq}^{\rm sep}}}
\newcommand\momsep{\ensuremath{\mathsf{Mom}^{\rm sep}}}

\newcommand{\M}{\mathcal M}
\newcommand{\Msa}{\mathcal M^{\mathrm{sa}}}
\newcommand{\Mpsd}{\mathcal M^+}

\newsavebox{\myboxbuffer}
\newenvironment{mybox}
  {\begin{lrbox}{\myboxbuffer}\begin{minipage}{0.88\linewidth}\vspace{2pt}}
  {\vspace{2pt}\end{minipage}\end{lrbox}\begin{center}\begin{tikzpicture}\node[fill=gray!15,draw=none,rounded corners=3pt,inner sep=6pt]{\usebox{\myboxbuffer}};\end{tikzpicture}\end{center}}

\begin{document}

\title{Optimal entanglement criteria from trace invariants}

\author{Bivas Mallick}
\email{bivasqic@gmail.com}
\address{S. N. Bose National Centre for Basic Sciences, Block JD, Sector III, Salt Lake, Kolkata 700 106, India}
\address{Statistics and Mathematics Unit, Indian Statistical Institute, 560059 Bangalore,
India}

\author{Aabhas Gulati}
\email{aabhasgulati7@gmail.com}
\address{Institut de Mathématiques, Université de Toulouse, UPS, France}

\author{Ion Nechita}
\email{nechita@irsamc.ups-tlse.fr}
\address{Laboratoire de Physique Th\'eorique, Universit\'e de Toulouse, CNRS, UPS, France}

\maketitle

\begin{abstract}
Randomized measurement protocols give experimental access to low-degree polynomial invariants of a quantum state rather than to the state itself, which raises the question of what the best entanglement criterion is that can be built from finitely many such invariants. In this work we answer this question through convex duality. 

For each degree we introduce the cone of local unitary trace inequalities that are valid for separable bipartite states in all local dimensions, together with its one-matrix counterparts for Hermitian and positive semidefinite matrices. Each cone is the polar dual of a moment cone; we also determine the minimal combinatorial parametrization of the separable cone. The framework organizes the known moment criteria: the PPT criterion using the first three moments of the partial transposition is the image under an embedding of a Hankel determinant from the positive semidefinite cone to the bipartite cone, which we show to be an extremal ray, and we determine the first nontrivial separable cone completely, finding exactly four extremal rays. 

The framework allows us to analyze the realignment criterion and its centered (or enhanced) variant. We prove that the even singular value moments of the (centered) realignment matrix are trace invariants, associated with explicit fixed-point-free involutions. Extracting the optimal trace-norm bound from $m$ such moments thus becomes a dimension-free truncated moment problem, which we solve in closed form for $m = 2$ and relax to a semidefinite program of size $O(m)$ for arbitrary $m$. The dual solutions are polynomial minorants of $\sqrt{x}$ on $[0,1]$, and two explicit families, binomial and $L^1$-weighted, yield hierarchies of separability inequalities. Finally, we compare the criteria against one another and on PPT entangled states.
\end{abstract}

\tableofcontents

\section{Introduction}

Entanglement is the feature that most sharply separates quantum from
classical correlations, and it is the resource underlying quantum
teleportation \cite{bennett1993teleporting}, quantum cryptography
\cite{Ekert1991crypto}, and measurement-based quantum computation
\cite{Raussendorf2001computer,Raussendorf2003computer}.  Deciding whether a
bipartite density matrix $\rho$ on $\mathbb C^{d_A}\otimes\mathbb C^{d_B}$ is
\emph{separable}, that is, whether it admits a decomposition
$$
\rho=\sum_i p_i\,\sigma_i^A\otimes\sigma_i^B,
\qquad p_i\geq0,\quad \sigma_i^A,\sigma_i^B\succeq0,
$$
is the central computational problem of entanglement theory.  The problem is
NP-hard in the local dimension \cite{gurvits2003classical,gharibian2010strong}, and the practice
of the field is therefore built on \emph{separability criteria}: efficiently
checkable necessary conditions for separability, whose violation certifies
entanglement.  We refer the reader to the surveys
\cite{horodecki2009quantum,guhne2009entanglement} for a panorama.

Two criteria have been central to the subject.  The \emph{positive partial transpose}
(PPT) criterion of Peres and the Horodeckis states that a separable state
$\rho$ stays positive semidefinite after transposing one of the two tensor
factors, $\rho^{\Gamma_B}\succeq 0$
\cite{peres1996separability,horodecki1996separability}.  The
\emph{realignment}, or \emph{computable cross-norm} (CCNR), criterion of
Rudolph and of Chen and Wu states that a separable $\rho$ satisfies
$\|\Rn(\rho)\|_1\leq1$, where $\Rn$ reshuffles the matrix entries of $\rho$
according to $\Rn(X)_{ik,jl}=X_{ij,kl}$ and $\|\cdot\|_1$ is the trace norm
\cite{rudolph2000separability,chen2003matrix}.  The two criteria are
independent: neither implies the other, and, crucially for what follows, the realignment criterion detects states that are \emph{PPT
entangled}, the so-called bound entangled states. Zhang, Zhang, Zhang and Guo later showed
that realignment becomes strictly stronger after \emph{centering} the state,
that is, after replacing $\rho$ by $\rho-\rho_A\otimes\rho_B$; this is the
\emph{enhanced realignment} criterion \cite{zhang2008entanglement}.

At the other end of the spectrum, the Doherty--Parrilo--Spedalieri hierarchy of
semidefinite programs is \emph{complete}, in that it detects every entangled
state at some finite level \cite{doherty2004complete}.  At level $k$, one
asks whether $\rho$ admits a positive semidefinite extension
$\rho_{A B_1\cdots B_k}$ whose marginal on each pair $A B_i$ is $\rho$ and
which is invariant under permutations of the systems $B_1,\ldots,B_k$
(with additional PPT constraints in the PPT-enhanced version).  Every
separable state has such symmetric extensions for all $k$, whereas an
entangled state fails to have one at some finite level; thus the hierarchy
becomes stronger as it manipulates extensions of $\rho$ to an increasing
number of copies.

All of these criteria are computationally tractable, in the sense that one can check them efficiently on a (classical) computer given the complete description of the state $\rho$. In practice, these criteria presuppose \emph{full state tomography}, whose sample
complexity is prohibitive already for moderate system sizes; moreover, what present-day experiments can measure comfortably is very different.  Protocols
based on locally randomized measurements and on classical shadows
\cite{brydges2019probing,huang2020predicting,elben2023randomized} give access
to \emph{low-degree polynomial invariants} of the state: quantities of the
form
$$
\Tr_{\sigma,\tau}(\rho)
=
\Tr\left[(P_\sigma\otimes P_\tau)\,\rho^{\otimes n}\right],
\qquad \sigma,\tau\in S_n,
$$
where $P_\sigma$ permutes the $n$ copies of the $A$ system and $P_\tau$ those
of the $B$ system of the bipartite state $\rho$.  These are precisely the polynomials in the entries of
$\rho$ that are invariant under local unitary conjugation
$\rho\mapsto (U\otimes V)\rho(U\otimes V)^*$, and the number $n$ of copies
they involve controls the experimental effort.  This shift of perspective
turns the separability problem into a question of \emph{truncated moment}
type, which is the one we address in this paper:

\begin{mybox}
Given only finitely many local unitary trace invariants of a bipartite state,
what is the \emph{best} entanglement criterion one can write down, uniformly
in the local dimensions?
\end{mybox}

\medskip

\noindent\textbf{Entanglement criteria from finitely many moments.}
The prototype of such a criterion is the 3 moments PPT criterion of Elben \etal
\cite{elben2020mixed}.  Writing $q_k:=\Tr[(\rho^{\Gamma_B})^k]$ for the
\emph{partial transpose moments}, and noting that $q_1=1$ and $q_2=\Tr(\rho^2)$
are the trace and the purity, they observed that a PPT state must satisfy
$q_3\geq q_2^2$, since the $q_k$ are then the moments of a genuine positive
spectrum; a violation certifies that $\rho^{\Gamma_B}\not\succeq 0$, hence
entanglement.  This criterion only requires three moments instead of the whole matrix, rendering it experimentally accessible \cite{elben2020mixed}.

This observation opened a rapidly growing line of work. Yu, Imai and Gühne \cite{yu2021optimal} gave two systematic strengthenings: a
family indexed by the \emph{Hankel matrices} built from $q_1,\dots,q_m$, whose
lowest level is exactly the 3 moments PPT criterion, and a second, \emph{optimal}
criterion which decides whether the observed moments are compatible with any
nonnegative spectrum at all.  Neven \etal \cite{neven2021symmetry}
resolved the construction by the symmetry sectors of the state, while
Carrasco \etal \cite{carrasco2024entanglement} mapped out the resulting
entanglement phase diagrams; Imai \etal \cite{imai2021bound} showed that
higher moments of randomized measurements already reveal bound entanglement.
Very recently, Miller and Eisert \cite{miller2026detecting} observed that one
need not know a whole initial segment $q_1,\dots,q_m$: any three moments
$q_k,q_l,q_m$ with $k<l<m$ suffice, through the log-convexity inequality
$q_l\leq q_k^{x}q_m^{1-x}$ with $x=(m-l)/(m-k)$, and they identified regimes in
which finitely many moments reproduce the full PPT criterion.  Moments of more general positive maps have been considered as well \cite{nzrc-8yrt,ali2025detection}, including extensions of this moment-based framework to multipartite systems for the detection of genuine multipartite entanglement \cite{ffch-xyv2}.

A parallel and more recent strand replaces the partial transpose by the
realignment map, precisely because the latter can see PPT entanglement.
Zhang, Jing and Fei \cite{zhang2022quantum} derived separability criteria from
low-order realignment moments and from the associated Hankel matrices;
Aggarwal, Adhikari and Majumdar \cite{aggarwal2024entanglement,
aggarwal2024theoretical} developed \emph{partial} realigned moment criteria
for arbitrary local dimensions together with experimental implementations;
Zhao \etal \cite{zhao2025multipartite} and Huang \etal \cite{huang2026note}
obtained multipartite and generalized-moment extensions, and Wang \etal
\cite{wang2024moments} used moments to build entanglement \emph{measures} in
addition to criteria.  Closest to the present work, Tarabunga and Haug
\cite{tarabunga2026quantifying} combined partial transpose and realignment
moments into quantitative witnesses for mixed states, and pointed out that
the even realignment moments are efficiently computable for matrix product
operators.  What all of these constructions have in common is that a
particular algebraic inequality is proposed, and its detection power is then
explored; what is missing is an answer to the \emph{optimality} question.
Given the first $m$ even realignment moments and nothing else, the convex hull of the set of
moment vectors compatible with the realignment criterion is a well-defined
convex body, and the strongest criterion available is the description of that
body.  Determining it is one of the main contributions of this paper.

\medskip

\noindent\textbf{Trace invariants and entanglement detection.}
In our work, we consider an invariant-theoretic approach to moment-based entanglement criteria.
The polynomial local unitary invariants of a bipartite mixed state are generated
by the permutation contractions, up to
dimension-dependent relations; this point of view goes back to Grassl,
Rötteler and Beth \cite{GrasslRoettelerBeth1998} and Makhlin
\cite{Makhlin2002} for qubits, was developed by Sudbery
\cite{sudbery2001local} and Szalay \cite{Szalay2011} in general, and has been
revisited recently for multipartite entanglement classification by Carrozza,
Chevrier and Lionni \cite{carrozza2026tensor}.

Huber's trace-polynomial formalism \cite{Huber2021PositiveMapsTracePolynomials} provides the
dictionary that makes such invariants computationally usable: it translates
the action of $S_n$ on $(\mathbb C^d)^{\otimes n}$ into matrix
multiplication, and thereby converts identities in the group algebra
$\mathbb R[S_n]$ into operator inequalities and \emph{trace polynomials},
that is, polynomials in matrix monomials and their traces.  On this
foundation, Klep, Magron and Volčič
\cite{klep2025sums} built a systematic optimization theory for trace and
moment polynomials, with sums-of-squares certificates; Huber, Klep, Magron
and Volčič \cite{huber2022dimension} used it to characterize entanglement in
multipartite Werner states \emph{dimension-freely}, exhibiting for every
entangled Werner state a witness that is valid in all local dimensions
simultaneously.  Rico and Huber \cite{rico2024entanglement} turned trace
polynomial inequalities (in particular immanant inequalities) into nonlinear
entanglement witnesses implementable by randomized measurements, and showed
that nonlinear detection succeeds on pairs of states and witnesses for which
linear detection fails.  Most recently, Fraser \etal
\cite{fraser2026complete} derived separability criteria as universal bounds
on tensor powers of separable states, obtaining a \emph{complete} (as the number of copies $n$ increases), local-unitary-invariant family of
nonlinear witnesses.

In a different but closely related direction, dimension-free inequalities between the trace
invariants of a single matrix are exactly the subject of the theory of
symmetric function inequalities in the limit of infinitely many variables,
studied by Blekherman and Riener \cite{blekherman2021symmetric}, by Acevedo,
Blekherman, Debus and Riener \cite{acevedo2025symmetric} through
tropicalization of the dual cone, and by Debus and Schabert
\cite{debus2026any} through any-dimensional Positivstellensätze; sharp
dimension-dependent counterparts for the linear entropy were obtained by
Morelli \etal \cite{morelli2020dimensionally}.

\medskip

Our starting point is the observation that these two literatures are the two
halves of the same object.  A criterion built from finitely many moments
is an element of a convex cone of trace-invariant inequalities, and
asking for the optimal criterion at a fixed number of moments is
asking for the facial structure of that cone.  We therefore study, for each
degree $n$, the cone
$$
\ineqsep_n
=
\left\{
w\in\mathbb R[S_n\times S_n]:
\begin{array}{l}
\Tr_w(\rho)\geq0\ \text{for every }d_A,d_B\geq1\\
\text{and every separable }\rho\in\Sep(\mathbb C^{d_A}:\mathbb C^{d_B})
\end{array}
\right\}
$$
of trace-invariant inequalities that are valid for separable states in
\emph{every} local dimension, together with its two one-matrix couterparts
$\ineqall_n$ and $\ineqpsd_n$, obtained by testing arbitrary Hermitian,
respectively positive semidefinite, matrices.  Each of these cones is dual to
a moment cone, so that producing a criterion
and certifying that it cannot be improved become two sides of one convex
duality.  Note that we are interested in this work in dimension-free criteria; this allows us to index the coordinates by combinatorial data: partitions in the one-matrix case, and (orbits of) pairs of permutations in the bipartite case. This data depends only on the number of copies of the state and not on the local dimensions.

The second observation is that the realignment criterion, which at first
sight lies outside this framework because the trace norm $\|\Rn(\rho)\|_1$ is
not a polynomial in $\rho$, can be brought inside it through its \emph{even}
singular value moments.  These moments turn out to be single tensor trace
invariants for an explicit pair of
fixed-point-free involutions. Even moments of the centered, enhanced
version of the criterion can be expressed as a linear combination of
tensor traces.  Estimating the trace norm from finitely many of
these moments is then a one-dimensional truncated moment problem, and its
solution converts the realignment criterion into a hierarchy of genuine elements of the cone $\ineqsep_n$.

\bigskip

\begin{center}
\textbf{\large  ------ Main results ------}
\end{center}

\medskip
Our contributions are threefold:

\begin{itemize}
\item \emph{A convex-geometric framework for dimension-free inequalities.}
We define the cones of trace inequalities for Hermitian and positive
semidefinite matrices, together with the cone of trace-invariant inequalities
for bipartite separable states.  We identify their dual moment cones, determine
the minimal combinatorial parametrization of the bipartite cone, and describe
the embeddings that connect the one-matrix and bipartite settings.

\item \emph{A solution of the moment problem underlying realignment criteria.}
We show that the even singular-value moments of the realignment map (those of its centered version) are (linear combinations of) tensor trace invariants.  This reduces the
problem of extracting the strongest trace-norm information from finitely many
moments to a dimension-free truncated moment problem, which we solve in closed
form for two moments and provide a relaxation, which can be formulated semidefinite program for an arbitrary
number of moments.

\item \emph{Explicit criteria and a comparison of their detection power.}
We use the moment solution to construct a hierarchy of separability criteria,
prove its optimality among criteria based on the same measured moments, and
analyze the resulting inequalities through the geometry of the dual cones.
Along the way, we identify extremal low-degree inequalities and compare the
new few-moment realignment criteria with the standard and centered criteria,
including examples where the combined information gives strictly stronger
detection.
\end{itemize}

We detail below these contribtuins, with pointers to the formal results in the main text.

\medskip

\noindent\textbf{(1) A dimension-free hierarchy of inequality cones, and its
exact parametrization.}

\begin{mybox}
We introduce and study the cones $\ineqall_n\subseteq\ineqpsd_n$ of
dimension-free trace inequalities for one Hermitian, respectively positive
semidefinite, matrix, and the cone $\ineqsep_n\subseteq\mathbb R[S_n\times
S_n]$ of dimension-free trace-invariant inequalities for bipartite separable
states.  Each is the polar dual of the corresponding moment cone.
\end{mybox}

\noindent The positive semidefinite cone can realized inside the bipartite cone in two ways. The two embeddings (\cref{prop:psd-embeddings}) send a partition
monomial $[\lambda]$ that corresponds to a permutation $\alpha_\lambda$ to $(\alpha_\lambda,\alpha_\lambda)$ and to
$(\alpha_\lambda,\alpha_\lambda^{-1})$, and account respectively for the
positivity of $\rho$ and for the positivity of $\rho^{\Gamma_B}$:
$$
\iota_+(\ineqpsd_n)\subseteq\ineqsep_n,
\qquad
\iota_\Gamma(\ineqpsd_n)\subseteq\ineqsep_n.
$$
This immediately organizes the existing moment criteria.  For instance, the
3 moments PPT criterion of \cite{elben2020mixed} is nothing but the image under
$\iota_\Gamma$ of the Hankel determinant $\Delta_{13}=p_1p_3-p_2^2$, and we
prove that $\mathbb R_{\geq0}\Delta_{13}$ is an \emph{extremal ray} of
$\ineqpsd_4$ (\cref{prop:delta13-extremal}). Extremality is the precise sense
in which an inequality is irredundant: $\Delta_{13}$ is not a nonnegative
combination of other dimension-free degree-four inequalities, the 3 moments PPT criterion is genuinely new rather than a
consequence of simpler ones.  Extremality is,
however, not automatic for Hankel minors: already $\Delta_{15}=p_1p_5-p_3^2$
decomposes nontrivially inside $\ineqpsd_6$.  We also show that the
dimension-free cones are non-polyhedral: $\ineqpsd_3$
has infinitely many extremal rays (\cref{prop:ineqpsd3-extremal-rays}).

On the bipartite side, the group algebra $\mathbb R[S_n\times S_n]$ is overcomplete, since simultaneous conjugation and simultaneous
inversion of $(\sigma,\tau)$ leave the invariant unchanged.  We identify the
minimal set of coordinates (\cref{thm:minimal-parametrization-exact-size}):
they are the orbits of $S_n^2$ under $S_n\times C_2$ counted by Bogaerts and
Dukes \cite{bogaerts2014semidefinite}, whose number $b_n$ we record in closed
character-theoretic form.  In the first nontrivial degree we obtain a complete
answer: $\ineqsep_2$ has exactly four extremal rays
(\cref{cor:ineqsep2-extremal-rays})
$$
\pi\geq0,
\qquad
\pi_A\geq\pi,
\qquad
\pi_B\geq\pi,
\qquad
t^2+\pi\geq\pi_A+\pi_B,
$$
where $\pi,\pi_A,\pi_B$ are the global and local purities and $t=\Tr\rho$.
Every degree-two dimension-free separability inequality is a nonnegative
combination of these four.

\medskip

\noindent\textbf{(2) Realignment moments are tensor trace invariants.}

\begin{mybox}
The even singular value moments of the realignment of a bipartite Hermitian
operator are single local unitary trace invariants: for every $m\geq1$,
$$
r_m(X):=\Tr\left[\bigl(\Rn(X)\Rn(X)^*\bigr)^m\right]=\Tr_{\alpha_m,\beta_m}(X),
$$
with $\alpha_m,\beta_m\in S_{2m}$ explicit fixed-point-free involutions. For the centered operator $\tilde X$ of the enhanced realignment criterion,
each $r_m(\tilde X)$ is a linear combination of such invariants, of degree
$4m$ in $X$.
\end{mybox}

\noindent See \cref{prop:even-realignment-moments-as-trace-invariants} and
\cref{prop:centered-realignment-moments}.  The permutations are
$$
\alpha_m=(1\,\,2m)(2\,\,3)(4\,\,5)\cdots(2m-2\,\,2m-1),
\qquad
\beta_m=(1\,\,2)(3\,\,4)\cdots(2m-1\,\,2m),
$$
and the lowest moment is the purity, $r_1(X)=\Tr(X^2) = \pi$.  This places the
realignment criterion on exactly the same footing as the partial transpose
criterion: its moments are polynomial invariants of degree $2m$ in $\rho$,
estimable from $2m$ copies by randomized measurements, and they are elements
of the same group algebra in which $\ineqsep_n$ lives.  It is this
identification that makes the realignment criterion amenable to a
few-moment treatment and places it on the same footing as the moment-PPT criteria, allowing us to compare them.

\medskip

\noindent\textbf{(3) The optimal trace-norm bound from $m$ moments.}
The passage from moments to a criterion is governed by a single scalar
quantity.  For a target even moment vector $r\in\mathbb R_+^m$, set
\begin{equation*}
\alpha(r)
=
\min_{n\in\mathbb N}\ \min_{s\in[0,1]^n}
\left\{\sum_{i=1}^ns_i\ \middle|\ \sum_{i=1}^ns_i^{2k}=r_k
\ \text{for all }k\in[m]\right\}.
\end{equation*}
Thus
$\alpha(r)$ is the smallest $\ell^1$ norm of a nonnegative vector with entries
in $[0,1]$ whose first $m$ even moments are the prescribed $r_k$, with no
constraint on the length, hence none on the local dimensions. If $X$ is
separable, then the singular values of $\Rn(X/t_X)$ form such a vector, and
the realignment criterion bounds their sum by one; therefore
$$
\alpha\bigl(r(X/t_X)\bigr)
\ \leq\
\left\|\frac{\Rn(X)}{t_X}\right\|_1
\ \leq\ 1
$$
 for
every separable $X$ (\cref{thm:realignment-moments-sep}), and $\alpha(r)>1$
certifies entanglement. Conversely, if $\alpha(r)\leq1$ then the observed
moments are realized by some nonnegative vector of $\ell^1$ norm at most one,
so no sharper conclusion can be drawn from $r_1,\dots,r_m$ alone. In this
precise sense, $\alpha$ \emph{is} the optimal $m$-moment relaxation of the
realignment criterion.

\begin{mybox}
We solve this dimension-free truncated moment problem exactly for $m\leq2$,
relax it to a semidefinite program of size $O(m)$ for arbitrary $m$, and
determine precisely when the relaxation is tight.
\end{mybox}

\noindent Concretely (\cref{thm:r1-r2-exact}), the problem is feasible if and
only if $0<r_2\leq r_1^2$ or $r_1,r_2 =0$, and in the first case the optimum is attained on a vector of support
size exactly $n=\lceil r_1^2/r_2\rceil$, and $\alpha(r_1,r_2)$ is given in
closed form; the resulting boundary curve $\alpha=1$ is a countable union of
arcs (\cref{fig:alpha-r1-r2}).  Relaxing atomic measures with unit weights to
arbitrary positive measures on $[0,1]$ gives a lower bound $\alpha_c(r)\leq
\alpha(r)$ which, by the classical Hausdorff truncated moment problem \cite{shohat1943problem,akhiezer2020classical,lasserre2009moments},
is computable by a semidefinite program.  For two moments we obtain the
strikingly simple value $\alpha_c(r_1,r_2)=r_1^{3/2}/\sqrt{r_2}$
(\cref{thm:r1-r2-alpha-c}), and the exact criterion for tightness
$$
\alpha(r_1,r_2)=\alpha_c(r_1,r_2)
\qquad\Longleftrightarrow\qquad
\frac{r_1^2}{r_2}\in\mathbb N .
$$
The gap between the two quantities $\alpha$ and $\alpha_c$ can be understood as an \emph{integrality} gap: what the continuous
relaxation loses is precisely the information that singular values come with
unit multiplicities.

On the dual side, feasible points of the dual program are exactly
the polynomial minorants of $\sqrt x$ on $[0,1]$, and we show that within the
family of Hölder interpolation inequalities the smallest exponents are
already the best using log-convexity of the $p$-norms.  There is thus nothing to be gained from
higher interpolation moments, which sharpens the three-moment philosophy of
\cite{miller2026detecting} in the realignment setting.

\medskip

\noindent\textbf{(4) Two hierarchies of explicit separability inequalities.}
Every polynomial $p(x)=\sum_{k=1}^mp_kx^k$ with $p\leq\sqrt x$ on $[0,1]$
yields, by \cref{thm:realignment-moments-sep}, a dimension-free inequality
for separable states,
$$
\sum_{k=1}^mp_k\,r_k(X)\,t_X^{2(m-k)}\leq t_X^{2m}
\qquad\leadsto\qquad
\text{an element of }\ineqsep_{2m},
$$
and a centered counterpart obtained by the substitutions $X\to\tilde X$ and
$t_X^2\to G(X)$.  We construct two such families of \emph{universal
minorants}:
\begin{itemize}
    \item the binomial minorants $f_k$ of
\cref{eq:binomial-realignment-minorant}, coming from the Taylor expansion of
$\sqrt x$ at $x=1$, which approximate $\sqrt x$ well near $1$
    \item the $L^1_\beta$-weighted minorants $g_{k,\beta}$ of \cref{eq:l1beta}, defined as
optimizers of a weighted $L^1$ distance and computed by an explicit
semidefinite program, which approximate $\sqrt x$ well near $0$.
\end{itemize}
Both converge pointwise to $\sqrt x$, hence both
hierarchies converge to the full (enhanced) realignment criterion
(\cref{fig:minorants}).  
The first two binomial minorants $f_1(x)=x$ and $f_2(x)=\frac32x-\frac12x^2$ give
$r_1(X)\leq t_X^2$, which is just the purity bound $\Tr(X^2)\leq(\Tr X)^2$, and
$\frac32r_1t_X^2-\frac12r_2\leq t_X^4$, which is a genuinely bipartite element of
$\ineqsep_4\setminus\iota_+(\ineqpsd_4)$.  Separately, the two-moment
relaxation $\alpha_c$ of \cref{thm:r1-r2-alpha-c} (equivalently, Hölder
interpolation) yields the degree-six inequality $r_1(X)^3\leq t_X^2r_2(X)$ of
\cref{eq:interpolation-sep-re}, which recovers a criterion of
\cite{tarabunga2026quantifying}.

\medskip

    \noindent\textbf{(5) PPT entangled states through realignment moments.}

\begin{mybox}
As few as nine standard, or five centered, realignment moments
already certify the bound entanglement of the \emph{Tiles} state; only two centered moments suffice for a newly introduced PPT entangled state.
\end{mybox}

\noindent No criterion in $\iota_+(\ineqpsd_n)\cup\iota_\Gamma(\ineqpsd_n)$ can ever
detect a PPT entangled state, whatever the degree, hence genuinely bipartite
invariants are indispensable, and the realignment moment hierarchies
supply them.  In
\cref{sec:tiles} we test our hierarchies on the Tiles state $\rho_{\rm T}$
\cite{bennett1999unextendible} mixed with white noise, and report the exact
detection ranges in \cref{tab:tiles}. From this numerical benchmark, we draw several conclusions. First, the
$L^1_\beta$ hierarchy is more efficient than the binomial one: it
needs nine moments where the binomial hierarchy needs about fifty in the
standard case, and five against twenty in the centered case. Second, the
detection window widens quickly with $k$; already at $k=15$ the $L^1_\beta$
criterion reaches a noise threshold close to that of the full realignment
criterion.

Since the required order could depend on the choice of the Tiles state, we also
optimize numerically over PPT entangled states of $\C^3\otimes\C^3$ and, in
\cref{sec:ppt-state}, identify in closed form a state $\rho_\star$, one quarter
of a rank-four projection with maximally mixed marginals and only two distinct
nonzero realignment singular values, whose detection order is substantially
lower: the binomial hierarchy detects it at order $k=30$ (standard) and $k=13$
(centered), instead of $52$ and $20$ for the Tiles state. More strikingly, the
centered interpolation criterion $r_1(\tilde\rho)^3\leq G\,r_2(\tilde\rho)$,
which uses only two centered realignment moments, detects the noisy family
$p\rho_\star+(1-p)\text{I}_9/9$ on exactly the same range $p>4/5$ as the full
enhanced realignment criterion, whereas the standard two-moment interpolation
criterion detects neither $\rho_\star$ nor the Tiles state
(\cref{tab:rho-star}).

\medskip

\noindent\textbf{(6) A complete comparison in degree four.}
Finally, we compare the criteria against each other rather than against states. On pure states nothing is lost by
truncation: each of the three realignment-moment criteria of
\cref{prop:pure-general} is violated on entangled pure states. On
mixed states the picture is more interesting. We prove the moment inequality
$$
r_2(\rho)\geq p_4\bigl(\rho^{\Gamma_B}\bigr)
\qquad\text{for every }\rho\succeq0,
$$
i.e.\ $\|\Rn(\rho)\|_4^4\geq\|\rho^{\Gamma_B}\|_4^4$
(\cref{prop:realignment-pt4}), and deduce from it that the 3 moments PPT criterion
\emph{dominates} the two-moment interpolation criterion $r_1^3\leq r_2$: every
state detected by the latter is detected by the former
(\cref{thm:interpolation-implied-ppt3}). At fixed degree
four, however, no such ordering survives: among the four criteria
$$
W_{\rm P}=tq_3-q_2^2,
\quad
W_{\rm E}=G-\tilde r_1,
\quad
W_{\rm M}=t^4-\tfrac32t^2r_1+\tfrac12r_2,
\quad
W_{\rm S}=(t^2-r_1)^2-2(r_1^2-r_2),
$$
all lying in $\ineqsep_4$, the \emph{only} inclusions between the
corresponding detection sets, uniformly over all local dimensions, are
$\mathcal V_{\rm M}\subsetneq\mathcal V_{\rm P}$ and $\mathcal V_{\rm
M}\subsetneq\mathcal V_{\rm S}$ (\cref{thm:four-detection-sets}); every other
pair is separated by explicit states.  The criterion $W_{\rm M}$ is
redundant: it is strictly weaker than both $W_{\rm P}$ and $W_{\rm S}$.
Thus a practitioner with a degree-four budget should test only $W_{\rm P}$,
$W_{\rm E}$, and $W_{\rm S}$.

\medskip

Throughout the paper we use the graphical notation for tensors introduced by
Penrose \cite{penrose1971applications}, which is particularly well suited to
quantum theory, where tensor products are ubiquitous; see
\cite{bridgeman2017hand,taylor2024introduction} for modern presentations.  Several of our proofs, notably those
concerning realignment moments, become transparent when read as diagram
manipulations.  Let us also point out that the framework developed here is not
tied to the bipartite setting: the cones $\ineqsep_n$ and the moment
techniques extend naturally to $\mathbb R[S_n^{\times r}]$ for $r$-partite
states, where the relevant trace invariants are those arising in tensor generalizations of free
probability theory \cite{nechita2025tensor}. The cones studied here admit a natural interpretation as tensor
generalizations of the classical moment problem.  In the Hamburger,
Stieltjes, and Hausdorff problems, one is given a sequence of scalar moments
and asks whether it has a representing measure, possibly subject to a
prescribed support constraint.  In the tensor setting, the data are instead a
finite family of trace invariants $\bigl(m_{\sigma,\tau}\bigr)_{\sigma,\tau\in S_n}$, and the question is whether there exists an operator $\rho$ realizing
$m_{\sigma,\tau}=\operatorname{Re}\Tr_{\sigma,\tau}(\rho)$ and satisfying a
prescribed structural constraint, such as Hermiticity, positive
semidefiniteness, the PPT property, or separability.  In particular,
$\momsep_n$ is the closed conic hull of invariant data admitting a separable
realization, whereas its polar cone
$\ineqsep_n=(\momsep_n)^\circ$ consists of all linear compatibility
inequalities obeyed by such data.  Violating one of these inequalities proves
that no separable tensor can have the prescribed trace invariants.

\begin{mybox}
In this paper, we take a first step in this program by studying the
positive-semidefinite inequality cone $\ineqpsd_n$, the separability
inequality cone $\ineqsep_n$, and the embeddings that connect them.  This
viewpoint clarifies how standard bipartite entanglement criteria (such as the
PPT, realignment, and enhanced realignment) give rise to systematic
families of inequalities involving only finitely many trace invariants.
\end{mybox}

We leave the multipartite extension, and the corresponding criteria, for
future work.

\bigskip

The paper is organized as follows.  \Cref{sec:all} introduces the
dimension-free cone $\ineqall_n$ of trace inequalities valid for all Hermitian
matrices, establishes its duality with the corresponding moment cone,
computes its small-degree extremal rays, and derives its Hankel inequalities
from the Hamburger moment problem.  \Cref{sec:psd} carries out the same
programme for positive semidefinite matrices, where the Stieltjes moment
problem replaces the Hamburger one and more inequalities become
available; \cref{sec:hankel-determinants} discusses the Hankel determinants
that underlie the partial transpose moment criteria.  \Cref{sec:sep} passes to
the bipartite setting, defines $\ineqsep_n$ and its moment dual, constructs
the two embeddings of the one-matrix cone, establishes the minimal
parametrization of $\ineqsep_n$, and determines $\ineqsep_2$ completely.
\Cref{sec:realignment} is the technical heart of the paper: it
identifies the even realignment moments as tensor trace invariants, solves the
dimension-free truncated moment problem for the trace norm, and constructs the
universal minorants that convert its dual solutions into elements of
$\ineqsep_n$. \Cref{sec:detection} evaluates and compares the resulting criteria: on pure
states (\cref{sec:pure}), against the 3 moments PPT criterion (\cref{sec:comparison}),
against one another in degree four (\cref{sec:degree-four}), and finally on PPT
entangled states, namely the noisy Tiles state (\cref{sec:tiles}) and an
optimized PPT entangled state of minimal Taylor order (\cref{sec:ppt-state}).
Two appendices collect the semidefinite programs for the weighted minorants and
the proof of the two-moment trace-norm theorem.

\section{Inequalities for Hermitian matrices}\label{sec:all}

We begin by studying the inequalities satisfied by polynomials in traces of powers of Hermitian matrices. In this section we shall consider arbitrary Hermitian matrices without any additional constraints (such as positive semidefiniteness or separability).
For integers $d_1,d_2\geq1$, let $\M_{d_1,d_2}$ denote the vector space of complex
$d_1\times d_2$ matrices, and write $\M_d:=\M_{d,d}$.  We denote by
$\Msa_d\subseteq\M_d$ the real vector space of self-adjoint matrices and by
$\Mpsd_d\subseteq\Msa_d$ the cone of positive semidefinite matrices (which will be the topic of the next section).

Let
$X\in \Msa_d$ be Hermitian and let $\sigma\in S_n$.  We first define the trace invariant
$$
\Tr_\sigma(X)
:=
\Tr(P_\sigma X^{\otimes n}),
$$
where $P_\sigma$ permutes the $n$ tensor factors of
$(\mathbb C^d)^{\otimes n}$.  If
$$
p_k(X):=\Tr(X^k)
$$
and $\lambda(\sigma)=(\lambda_1,\ldots,\lambda_\ell)\vdash n$ is the
cycle type of $\sigma$, then contraction along the cycles of $\sigma$ gives
$$
\Tr_\sigma(X)
=
p_{\lambda(\sigma)}(X),
\qquad
p_\lambda(X):=\prod_{j=1}^{\ell}p_{\lambda_j}(X).
$$
Equivalently, after diagonalizing $X$ with real eigenvalues
$x=(x_1,\ldots,x_d)$,
$$
p_k(X)=p_k(x):=\sum_{i=1}^d x_i^k.
$$
Thus the trace invariant depends only on the cycle type of the
permutation.  We shall use throughout the vector space
\begin{equation}\label{eq:def-partition-space}
V_n:=\spanop\{[\lambda]:\lambda\vdash n\},
\end{equation}
where the basis vector $[\lambda]$ represents the trace monomial $p_\lambda$.
In other words, the evaluation map
$$
\mathbb R[S_n]\longrightarrow V_n,
\qquad
\sigma\longmapsto[\lambda(\sigma)],
$$
is the cycle-type quotient of the real group algebra.  This notation will
also be used for positive semidefinite matrices in the next section, while the separability inequalities will require a different combinatorial structure.

This section is the first of a series of sections dealing with increasingly specialized matrices.  We restrict test
matrices in later sections to nonnegative spectra (\cref{sec:psd}) and to separable
bipartite states (\cref{sec:sep}).  At each step we define an inequality cone and a
dual moment cone linked by the same coefficient pairing. The mathematical framework changes
from Hamburger moment problems to Stieltjes moment problems to moments of separable bipartite qauntum states.

\begin{definition}
The dimension-free all-Hermitian inequalities cone in degree $n$ is
$$
\ineqall_n
:=
\left\{
f=\sum_{\lambda\vdash n}a_\lambda[\lambda]\in V_n:
\begin{array}{l}
f(X)\geq0\ \text{for every }d\geq1\\
\text{and every Hermitian }X\in \Msa_d
\end{array}
\right\}.
$$
Equivalently,
$$
\ineqall_n
=
\left\{
f=\sum_{\lambda\vdash n}a_\lambda p_\lambda:
f(x_1,\ldots,x_d)\geq0
\ \text{for all }d\geq1\text{ and all }x_i\in\mathbb R
\right\}.
$$
The corresponding all-Hermitian moment cone is
$$
\momall_n
:=
\overline{\cone\left\{
\bigl(p_\lambda(x)\bigr)_{\lambda\vdash n}:
x=(x_1,\ldots,x_d),\ d\geq1,\ x_i\in\mathbb R
\right\}}.
$$
\end{definition}

\begin{proposition}
\label{prop:all-parity}
If $n$ is odd, then
$$
\ineqall_n=\{0\}.
$$
If $n$ is even, every product of even power sums of total degree $n$ belongs
to $\ineqall_n$.  In particular,
$$
[2m]\in\ineqall_{2m}.
$$
\end{proposition}

\begin{proof}
Let $f\in V_n$ be homogeneous of odd degree and nonnegative on every real vector $x$.  Since $f(-x)=(-1)^n f(x)=-f(x)$,
both $f(x)$ and $-f(x)$ are nonnegative.  Thus $f$ vanishes on every real
vector.  A symmetric function of degree $n$ is determined by its
specializations to finitely many variables, so $f=0$ in $V_n$.

Finally, $p_{2m}(x)=\sum_i x_i^{2m}\geq0$ on every real vector.  Products
of even power sums are therefore nonnegative as well.
\end{proof}

We pair $V_n$ with the real vector space indexed by partitions of $n$ by the standard scalar product
$$
\left\langle
\sum_{\lambda\vdash n}a_\lambda[\lambda],
(y_\lambda)_{\lambda\vdash n}
\right\rangle
=
\sum_{\lambda\vdash n}a_\lambda y_\lambda.
$$
For any cone $C$, we use the nonnegative
polar convention
$$
C^\circ:=\{f:\langle f,y\rangle\geq0\ \text{for every }y\in C\}.
$$

\begin{proposition}
For every $n\geq1$,
$$
\ineqall_n=(\momall_n)^\circ.
$$
\end{proposition}

\begin{proof}
Let $f=\sum_{\lambda\vdash n}a_\lambda[\lambda]\in V_n$.  For every
$x=(x_1,\ldots,x_d)\in\mathbb R^d$, one has
$$
f(x_1,\ldots,x_d)
=
\left\langle f,
\bigl(p_\lambda(x)\bigr)_{\lambda\vdash n}
\right\rangle.
$$
Thus, by definition, $f\in\ineqall_n$ if and only if this pairing is
nonnegative on every evaluation vector.  Since the pairing is linear and
continuous, this is equivalent to its being nonnegative on the closed conic
hull of all evaluation vectors, which is precisely $\momall_n$.  Hence
$f\in\ineqall_n$ if and only if $f\in(\momall_n)^\circ$.
\end{proof}

\subsection{Small-degree extremal rays}

To separate dimension-free positivity from trace identities, define
$$
\ineqall_{n,d}
:=
\left\{f\in V_n:f(x_1,\ldots,x_d)\geq0
\ \text{for every }x\in\mathbb R^d\right\}.
$$
Then
$$
\ineqall_n=\bigcap_{d\geq1}\ineqall_{n,d}.
$$
For fixed $d$, the evaluation map $V_n\to\mathbb R[x_1,\ldots,x_d]^{S_d}$ (the space of symmetric polynomials)
can have a non-empty kernel; hence result in a trivial inequality. In what follows, when discussing extremal rays of these cones, we shall quotient first the cone by the kernel of the evaluation map.

For $d=1$, every $p_k$ equals $x^k$.  The quotient cone is zero for odd $n$
and is the single half-line of nonnegative multiples of $x^n$ for even $n$.

\begin{proposition}[All-Hermitian degrees one, two, and three]
\label[proposition]{prop:all-small-degrees}
One has
$$
\ineqall_1=\ineqall_3=\{0\}.
$$
The cone $\ineqall_2$ has exactly two extremal rays,
$$
\mathbb R_{\geq0}[2]
\qquad\text{and}\qquad
\mathbb R_{\geq0}[1^2].
$$
Equivalently,
$$
\ineqall_2=\cone\{p_2,p_1^2\}.
$$
\end{proposition}

\begin{proof}
The odd-degree assertions follow from \cref{prop:all-parity}.  Write a
quadratic trace polynomial as
$$
f=a[2]+b[1^2],
\qquad
f(x)=a p_2(x)+b p_1(x)^2.
$$
For a nonzero vector, $p_2>0$, so
$$
\frac{f(x)}{p_2(x)}=a+b t,
\qquad
t:=\frac{p_1(x)^2}{p_2(x)}.
$$
In dimension $d$, Cauchy--Schwarz gives $0\leq t\leq d$.  For every $d\geq2$
all values in $[0,d]$ are attained: on the unit sphere, the continuous
function $p_1$ has range $[-\sqrt d,\sqrt d]$.  Since $d$ is arbitrary, the
dimension-free range of $t$ is $[0,\infty)$.  Therefore $f\geq0$ on all real
vectors if and only if
$$
a+bt\geq0\qquad\text{for all }t\geq0,
$$
which is equivalent to $a\geq0$ and $b\geq0$.  These two coordinate
half-lines are the asserted extremal rays.
\end{proof}

\begin{remark}
In degree two there are no non-trivial trace identities that hold in all dimensions, so
\cref{prop:all-small-degrees} gives the full dimension-free cone directly.
In fixed dimension $d$ the cone $\ineqall_{2,d}$ is strictly larger.
For example, in dimension $d=2$ the range
$$
0\leq\frac{p_1^2}{p_2}\leq2
$$
gives
$$
\ineqall_{2,2}
=
\{a[2]+b[1^2]:a\geq0,\ a+2b\geq0\},
$$
whose quotient cone has the additional extremal ray
$$
\mathbb R_{\geq0}\bigl(2[2]-[1^2]\bigr),
\qquad
2p_2-p_1^2=(x_1-x_2)^2.
$$  
This ray is nonnegative in dimension two but not dimension-free: on the
all-ones vector in dimension $d>2$ it equals $d(2-d)<0$. In this paper however we shall not be overly concerned with the case of fixed dimensions; we shall focus on the dimension-free case.
\end{remark}

\subsection{Hamburger moment inequalities}

Having classified the first degrees directly, we now describe the general
even-shift source.  Let $X=X^*$ have eigenvalues $x_1,\ldots,x_d\in\mathbb R$.  The power sums
are the moments of the positive atomic measure called the \emph{empirical eigenvalue distribution} of $X$:
$$
\nu_X:=\sum_{i=1}^d\delta_{x_i}
$$
on $\mathbb R$.  If $s\geq0$ is even and
$q(t)=\sum_{j=0}^r c_jt^j$, then
$$
\sum_{i,j=0}^r c_i c_j p_{s+i+j}(X)
=
\int_{\mathbb R}t^s q(t)^2\,\d\nu_X(t)
\geq0.
$$
Consequently every even-shifted Hankel matrix
$$
H_s(X):=\bigl(p_{s+i+j}(X)\bigr)_{i,j\geq0},
\qquad s\ \text{even},
$$
is positive semidefinite \cite{shohat1943problem,akhiezer2020classical}.  Every principal minor that does not involve
$p_0=d$ therefore gives a dimension-free trace inequality in the appropriate
homogeneous degree.

This is the Hamburger moment problem; its analogue for positive spectra, called the Stieltjes moment problem will be discussed in the following section \cite{akhiezer2020classical,lasserre2009moments}. The distinction is that a positive measure on
$[0,\infty)$ remains positive after multiplication by $t^s$ for every
$s\geq0$, whereas on $\mathbb R$ this is automatic only for even $s$.
Accordingly, odd Hankel shifts do not give all-Hermitian inequalities.

The first nontrivial example of the even shift Hankel matrix inequality is
$$
p_2p_4-p_3^2
=
\det
\begin{pmatrix}
p_2&p_3\\
p_3&p_4
\end{pmatrix}
\in\ineqall_6.
$$
Equivalently, this is the Cauchy--Schwarz inequality for the vectors
$(x_i)_i$ and $(x_i^2)_i$.  By contrast,
$$
p_1p_3-p_2^2
=
\det
\begin{pmatrix}
p_1&p_2\\
p_2&p_3
\end{pmatrix}
$$
comes from an odd shift and is not nonnegative on signed spectra:
for $X=(1,-1)$ one has $p_1=p_3=0$ and $p_2=2$, so
$p_1p_3-p_2^2=-4<0$; thus $[13]-[2^2] \notin \ineqall_4$.

We leave as an open problem the description of the (extremal rays of the) cone $\ineqall_4$. This is probably a difficult task since there exist such inequalities that do not admit a sum of squares decomposition (see \cite[Theorem 3.6]{acevedo2025symmetric} and also \cite{blekherman2021symmetric}):
$$4p_1^4-5p_2p_1^2-\frac{139}{20}p_1p_3+4p_2^2+4p_4 \in \ineqall_4.$$

\section{Inequalities for PSD matrices}\label{sec:psd}

In this section we still work in the one-matrix framework introduced above but restrict the test
matrices to being \emph{positive semidefinite}: $X\in \Mpsd_d$.

\begin{definition}
The dimension-free one-matrix PSD cone in degree $n$ is
$$
\ineqpsd_n
=
\left\{
\sum_{\lambda\vdash n}a_\lambda[\lambda]:
\sum_{\lambda\vdash n}a_\lambda p_\lambda(X)\geq0
\ \text{for every }d\geq1\text{ and every }X\in \Mpsd_d
\right\}.
$$
Equivalently, if $x_1,\ldots,x_d$ are the eigenvalues of $X$, then
$$
\ineqpsd_n
=
\left\{
f=\sum_{\lambda\vdash n}a_\lambda p_\lambda:
f(x_1,\ldots,x_d)\geq0
\ \text{for all }d\geq1,\ x_i\geq0
\right\}.
$$
The corresponding PSD moment cone is
$$
\mompsd_n
:=
\overline{\cone\left\{
\bigl(p_\lambda(x)\bigr)_{\lambda\vdash n}:
x=(x_1,\ldots,x_d),\ d\geq1,\ x_i\geq0
\right\}}.
$$
\end{definition}

This is the cone of homogeneous symmetric power-sum polynomials of degree $n$ which are nonnegative on every finite positive vector.  It is the spectral shadow of the moment-polynomial positivity problems studied in much greater generality in \cite{klep2025sums}.  The general even-shift source for trace inequalities in \cref{sec:all} is replaced here by the Stieltjes moment problem on $[0,\infty)$.

With the coefficient pairing introduced above, the defining condition gives
$$
\ineqpsd_n=(\mompsd_n)^\circ.
$$
The closure does not change the polar and records dimension-free limit points.
For example, after normalizing $p_1=\sum_i x_i=1$, the uniform spectra
$(1/m,\ldots,1/m)$ satisfy $p_k\to0$ for every $k\geq2$, although the limit
is not the evaluation vector of any finite positive vector.

\subsection{Relation with the all-Hermitian cone}

First, note that since we are only evaluating on positive semi-definite matrices, there are possibly more inequalities in this case.

\begin{proposition}
For every $n\geq1$ one has
$$
\ineqall_n\subseteq\ineqpsd_n.
$$
\end{proposition}

The inclusion is generally strict.  For example,
$$
p_1^2-p_2=2\sum_{i<j}x_ix_j\in\ineqpsd_2,
$$
whereas at the signed vector $(1,-1)$ it takes the value $-2$.  Likewise,
$$
p_1p_3-p_2^2=\sum_{i<j}x_ix_j(x_i-x_j)^2\in\ineqpsd_4,
$$
but this polynomial takes the value $-4$ at $(1,-1)$.  The strictness of these
examples is revisited after the Stieltjes/Hankel block and the degree-two
comparison in \cref{prop:all-small-degrees}.

\subsection{Hankel determinants from the Stieltjes moment problem}\label{sec:hankel-determinants}

Let $X\in \Mpsd_d$ have eigenvalues $x_1,\ldots,x_d\geq0$.  The sequence
$$
p_k(X)=\sum_i x_i^k
$$
is the Stieltjes moment sequence of the positive measure $\nu_X=\sum_i\delta_{x_i}$
on $[0,\infty)$, see \cite{shohat1943problem,akhiezer2020classical}.  Hence for every $s\geq0$ and every polynomial $q(t)=\sum_i c_i t^i$,
$$
\sum_{i,j}c_i c_j p_{s+i+j}(X)
=
\int_0^\infty t^s q(t)^2\,d\nu_X(t)
\geq0.
$$
Therefore every shifted Hankel matrix
$$
H_s(X)=\bigl(p_{s+i+j}(X)\bigr)_{i,j\geq0}
$$
is positive semidefinite.  In particular, every principal minor containing no
$p_0$ gives an element of $\ineqpsd_n$ in its homogeneous degree.  This is
automatic for $s\geq1$; for $s=0$, one must choose only positive row and
column indices.

The first examples are
$$
p_a p_c-p_b^2\in\ineqpsd_{a+c},
\qquad
a+c=2b,\quad a,b,c\geq1,
$$
because this is the determinant of
$$
\begin{pmatrix}
p_a&p_b\\
p_b&p_c
\end{pmatrix}.
$$
For instance, the degree-four Hankel inequality
\begin{equation} \label{p3-ppt}
    p_1p_3-p_2^2\in\ineqpsd_4 \nonumber
\end{equation}
and we also have
$$
p_2p_4-p_3^2\in\ineqpsd_6.
$$
Larger principal Hankel minors give higher-degree elements, such as
$$
\det
\begin{pmatrix}
p_1&p_2&p_3\\
p_2&p_3&p_4\\
p_3&p_4&p_5
\end{pmatrix}
\in\ineqpsd_9.
$$
These are exactly the inequalities coming from the one-variable Stieltjes moment problem.  They are universal for positive semidefinite matrices because they use no matrix structure beyond the nonnegative spectrum.

\subsection{Small-degree extremal rays}

We now compute the PSD extremal rays in the first degrees and compare them
with their all-Hermitian counterparts.  Recall that we write
$$
[3]=p_3,\qquad [21]=p_2p_1,\qquad [1^3]=p_1^3,
$$
and similarly in other degrees.

\begin{proposition}
The cone $\ineqpsd_1$ has the single extremal ray $\mathbb R_{\geq0}[1]$,
while the cone $\ineqpsd_2$ has exactly two extremal rays:
$$
\mathbb R_{\geq0}[2],
\qquad
\mathbb R_{\geq0}\bigl([1^2]-[2]\bigr).
$$
\end{proposition}

\begin{proof}
Degree one is immediate.  In degree two, normalize $p_1=1$.  Then $0\leq p_2\leq1$ in the dimension-free closure.  An element $a[2]+b[1^2]$ is nonnegative exactly when
$$
a t+b\geq0\qquad\text{for all }t\in[0,1].
$$
This is equivalent to $b\geq0$ and $a+b\geq0$, whose two boundary rays are $[2]$ and $[1^2]-[2]$.
\end{proof}

\begin{remark}
By \cref{prop:all-small-degrees}, the dimension-free Hermitian cone in degree
two is
$$
\ineqall_2=\cone\{[2],[1^2]\},
\quad \text{whereas} \quad
\ineqpsd_2=\cone\{[2],[1^2]-[2]\}.
$$
On the trace-one positive slice, $p_2$ ranges over $[0,1]$ in the
dimension-free closure, which produces the mixed PSD inequality
$p_1^2-p_2\geq0$.  Signed spectra instead make $p_1^2/p_2$ arbitrarily large, making two Hermitian rays the manifest squares
$\sum_i x_i^2$ and $(\sum_i x_i)^2$; see \cref{fig:all-vs-psd-deg-2}.
\end{remark}

\begin{figure}
\centering
\includegraphics[scale=.75]{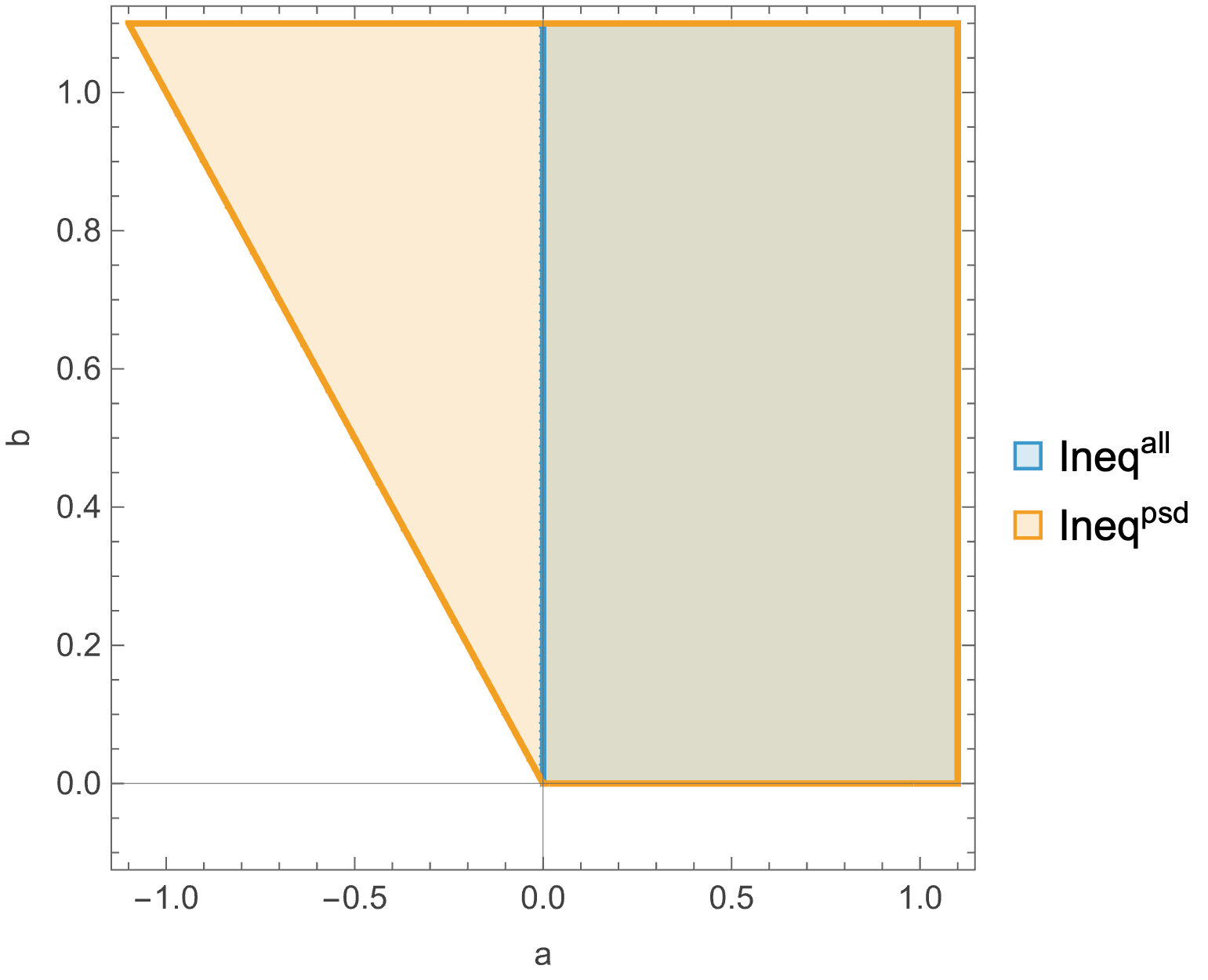}
\caption{Comparison of the Hermitian cone $\ineqall_2$ with the PSD cone $\ineqpsd_2$ in degree two. The two cones are parametrizes by $a[2] + b[1^2]$.}
\label{fig:all-vs-psd-deg-2}
\end{figure}

\begin{proposition}
\label{prop:ineqpsd3-extremal-rays}
The extremal rays of $\ineqpsd_3$ are
$$
\mathbb R_{\geq0}[3], \qquad \mathbb R_{\geq0}\bigl([21]-[3]\bigr),
$$
and, for every integer $k\geq1$,
\begin{equation}\label{eq:ineq-3-k}
\mathbb R_{\geq0}
\left(
[1^3]-(2k+1)[21]+k(k+1)[3]
\right).
\end{equation}
\end{proposition}

\begin{proof}
Let
$$
f=a[1^3]+b[21]+c[3].
$$
For a positive semidefinite matrix with eigenvalues $x_i\geq0$, this is the
symmetric homogeneous cubic
$$
f(x)=a p_1(x)^3+b p_1(x)p_2(x)+c p_3(x).
$$
Timofte's half-degree principle for symmetric polynomial inequalities on
the nonnegative orthant \cite[Theorem 5.1]{timofte2003positivity} says that a
symmetric inequality of degree $d$ is valid if and only if it is valid on
points having at most $\max\{\lfloor d/2\rfloor,2\}$ distinct components,
where zero is counted as a component.  Thus, in degree $3$, the test points
may have two distinct component values.

Normalize a nonzero two-level test point by $p_1=1$, and write $y=p_2$, $z=p_3$.
Then $0<y\leq1$.  Choose an integer $j\geq1$ such that
$$
\frac1{j+1}\leq y\leq\frac1j.
$$
The fixed-moment calculation in the following remark gives
$$
z\geq z_j^{\min}(y)\geq L_j(y),
\qquad
L_j(y):=\frac{(2j+1)y-1}{j(j+1)}.
$$
Indeed, with $s=\sqrt{j((j+1)y-1)}\in[0,1]$, the difference is
$$
z_j^{\min}(y)-L_j(y)
=\frac{(j-1)s^2(1-s)}{j^2(j+1)^2}\geq0.
$$

We claim that this implies, for every integer $k\geq1$,
\begin{equation}\label{eq:cubic-chord-halfspace}
1-(2k+1)y+k(k+1)z\geq0.
\end{equation}
Since the coefficient of $z$ is positive, it is enough to replace $z$ by
$L_j(y)$.  The resulting expression
$$
1-(2k+1)y+k(k+1)L_j(y)
$$
is affine in $y$, so its minimum on
$[1/(j+1),1/j]$ is attained at an endpoint.  At $y=1/m$, where
$m\in\{j,j+1\}$, one has $L_j(1/m)=1/m^2$, and hence
$$
1-\frac{2k+1}{m}+\frac{k(k+1)}{m^2}
=\frac{(m-k)(m-k-1)}{m^2}\geq0.
$$
The last inequality holds because $m-k$ and $m-k-1$ are consecutive
integers.  This proves \eqref{eq:cubic-chord-halfspace} on the two-level test
set.  After undoing the normalization, it is the homogeneous cubic inequality
$$
p_1^3-(2k+1)p_1p_2+k(k+1)p_3\geq0.
$$
Timofte's principle therefore proves it for every nonnegative spectrum.
Moreover, $z\geq0$ and $y-z\geq0$, since normalization gives
$0\leq x_i\leq1$ and hence $x_i^3\leq x_i^2$.

We now identify the moment cone.  If
$$
M_3(X):=\bigl(p_1(X)^3,\ p_1(X)p_2(X),\ p_3(X)\bigr),
$$
then $f(X)=\langle f,M_3(X)\rangle$ in the dual bases
$([1^3],[21],[3])$ and $(p_1^3,p_1p_2,p_3)$.  Consider the normalized slice
$$
K_3:=\{(y,z):(1,y,z)\in\mompsd_3\}
$$
and the closed region
$$
H_3:=\left\{(y,z):
z\geq0,\quad y-z\geq0,\quad
1-(2k+1)y+k(k+1)z\geq0\ \text{for every }k\geq1
\right\}.
$$
The preceding inequalities hold on every normalized evaluation point and
hence on their closed convex hull, so $K_3\subseteq H_3$.

For the reverse inclusion, consider the uniform rank-$m$ spectra
\begin{equation}\gamma_m := \label{eq:uniform-spectrum}
\left(\frac1m,\ldots,\frac1m,0,\ldots,0\right),
\end{equation}
whose normalized moment points are
$$
v_m:=\left(\frac1m,\frac1{m^2}\right),
\qquad m=1,2,\ldots.
$$
As the matrix size is arbitrary, all $m\geq1$ occur, and the limit
$m\to\infty$ gives $v_\infty=(0,0)$.  Since $K_3$ is closed and convex, we
therefore have
$$
C_3:=\overline{\conv\{v_m:m=1,2,\ldots\}}\subseteq K_3.
$$
We show that $H_3=C_3$.  The upper edge of $C_3$ is the segment from
$v_\infty=(0,0)$ to $v_1=(1,1)$, supported by $y-z$.  For each $k\geq1$,
the identity
$$
1-\frac{2k+1}{m}+\frac{k(k+1)}{m^2}
=\frac{(m-k)(m-k-1)}{m^2}\geq0
$$
shows that the line
$$
1-(2k+1)y+k(k+1)z=0
$$
supports $C_3$ along the lower edge $[v_{k+1},v_k]$.  Finally, $z=0$ is the
limiting supporting line at $v_\infty$.  Thus $C_3\subseteq H_3$.

These supporting lines also exhaust the boundary.  Indeed, let
$(y,z)\in H_3$.  The inequality with $k=1$, together with $z\leq y$, implies
$y\leq1$.  If $y=0$, then $z=0$.  If $0<y\leq1$, choose $j\geq1$ such that
$$
\frac1{j+1}\leq y\leq\frac1j.
$$
The $j$th lower inequality gives $z\geq L_j(y)$.  The point
$(y,L_j(y))$ lies on $[v_{j+1},v_j]$, while $(y,y)$ lies on
$[v_\infty,v_1]$.  Since also $z\leq y$, the point $(y,z)$ lies on the
vertical segment joining these two points.  Both endpoints belong to the
convex set $C_3$, so $(y,z)\in C_3$.  Hence $H_3\subseteq C_3$.

We have proved
$$
C_3\subseteq K_3\subseteq H_3=C_3,
$$
and therefore all three sets coincide.  In particular,
$$
\mompsd_3
=\overline{\cone\left\{
\left(1,\frac1m,\frac1{m^2}\right):m=1,2,\ldots
\right\}},
$$
where the closure includes $(1,0,0)$.  Consequently
$\ineqpsd_3=(\mompsd_3)^\circ$ is the cone of affine functions
$$
\ell(y,z)=a+b\,y+c\,z
$$
which are nonnegative on $K_3$.

We can now translate the boundary supports back into trace polynomials.  The
three types of support found above give
$$
z
\quad\Longleftrightarrow\quad
[3],
$$
$$
y-z
\quad\Longleftrightarrow\quad
[21]-[3],
$$
and, for $k\geq1$,
$$
1-(2k+1)y+k(k+1)z
\quad\Longleftrightarrow\quad
[1^3]-(2k+1)[21]+k(k+1)[3].
$$
Each of these spans an extremal ray of the cone of affine functions
nonnegative on $K_3$, because each is an extreme boundary normal of the
closed convex set $K_3$: the first is the limiting lower normal at $v_\infty$,
the second is the normal to the upper edge, and the third is the normal to one
lower edge.  It remains to check that a normal supported only at a vertex
does not give a new extremal ray.

Let $\ell$ span an extremal ray of the affine cone.  If $\ell$ is
strictly positive on $K_3$, then it is an interior point and cannot be
extremal.  Hence $\ell$ supports $K_3$.  If its zero set contains a relative
interior point of an edge, then $\ell$ is proportional to that edge's
supporting affine function, already listed.  If $\ell$ vanishes only at a
finite vertex $v_k$, then its normal lies in the two-dimensional normal cone
generated by the two adjacent edge normals.  At $v_1$ these are $y-z$ and
$1-3y+2z$; at $v_k$ for $k\geq2$ they are the lower-edge normals for
$[v_{k-1},v_k]$ and $[v_k,v_{k+1}]$.  Such a ray is therefore not new unless
it is one of the adjacent edge rays.

Finally, if $\ell$ vanishes only at $v_\infty$, then
$\ell(y,z)=b\,y+c\,z$.
Nonnegativity on the reciprocal-rank points, for which $z/y\to0$, gives
$b\geq0$, and nonnegativity on the upper edge $z=y$ gives $b+c\geq0$.
Therefore
$$
b\,y+c\,z=b(y-z)+(b+c)z,
$$
so the only extremal possibilities at $v_\infty$ are the already listed rays
$y-z$ and $z$.  This proves the displayed classification.
\end{proof}

\begin{remark}
The infinite family \eqref{eq:ineq-3-k} is closely related to, but not
identical with, the optimized condition involving moments up to degree three
derived in the \emph{Methods} section of \cite{neven2021symmetry}.  In that paper, the
first two moments are fixed and the third moment is minimized over positive
semidefinite matrices.  In particular, for
$y\in[1/(k+1),1/k]$ and normalized moments $p_1=1$, $p_2=y$, the minimum on
the two-level test set (indeed, on all nonnegative spectra) is attained at a
spectrum consisting of $k$ equal eigenvalues $\lambda_a$ and one eigenvalue
$\lambda_b$, with
\begin{equation}
\lambda_a=\frac{k+s}{k(k+1)},\qquad
\lambda_b=\frac{1-s}{k+1},\qquad
s=\sqrt{k\bigl((k+1)y-1\bigr)},
\end{equation}
and hence
\begin{equation}
z_k^{\min}(y)
=
k\lambda_a^3+\lambda_b^3
=
\frac{1}{(k+1)^2}
\left(
1+\frac{3s^2}{k}-\frac{(k-1)s^3}{k^2}
\right).
\end{equation}
For $k=1$ this reduces to $(3y-1)/2$, which is the $D_3$ condition of
\cite{neven2021symmetry}; for $k>1$ it is the exact, nonlinear lower
boundary of the feasible region in the slice $p_1=1$.

The normalized member of our family \eqref{eq:ineq-3-k} is the affine
function
\begin{equation}
L_k(y)=\frac{(2k+1)y-1}{k(k+1)}.
\end{equation}
At the endpoints $y=1/(k+1)$ and $y=1/k$, which correspond to the uniform
spectra of ranks $k+1$ and $k$, one has $z_k^{\min}(y)=L_k(y)$.  In the
interior of the interval,
\begin{equation}
z_k^{\min}(y)-L_k(y)
=
\frac{(k-1)s^2(1-s)}{k^2(k+1)^2}
>
0
\qquad
\text{for } k>1,
\end{equation}
so the exact optimized \emph{curve} lies strictly above the chord joining the two
adjacent uniform moment points.  Thus the formulas are not the same: the
coefficients $(2k+1)$ and $k(k+1)$ coincide because both constructions are
governed by the same rank transition, but \eqref{eq:ineq-3-k} is the
homogeneous linear inequality whose zero set is that chord, whereas the
condition of \cite{neven2021symmetry} is the nonlinear pointwise minimum.
In particular, the latter is not an element of $\ineqpsd_3$: it cannot be
written as a single homogeneous cubic inequality, but only as a
piecewise-defined necessary condition depending on the value of $p_2$. Our (linear) inequalities capture the fact that the set $\ineqpsd_3$ is dual to the moment set $\mompsd_3$: the functions $z_k^{\min}(y)$ capture the exact boundary of the set of tuples of moments, while $L_k(y)$ correspond to its convex hull. For similar considerations related to the realignment criterion, see \cref{sec:moment-interpolation} and \cref{fig:alpha-r1-r2}.
\end{remark}

The dual set cone can be readily obtained, see \cref{fig:mom-psd-deg-3}.

\begin{figure}
\centering
\includegraphics[scale=.75]{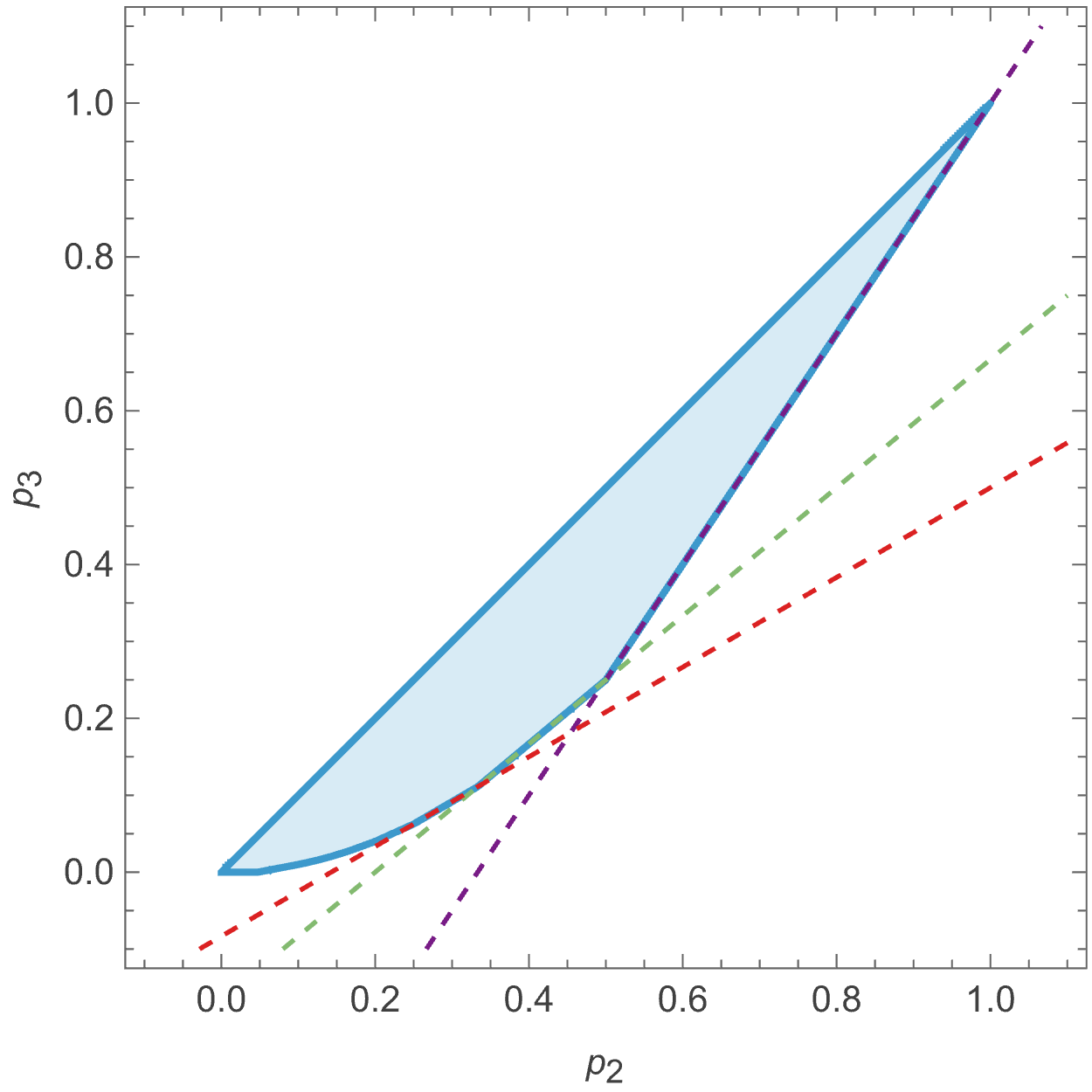}
\caption{The slice $p_1=1$ through $\mompsd_3$. The top boundary corresponds to the extremal ray $p_1p_2 \geq p_3$ of $\ineqpsd_3$. The bottom boundary is a union of infinitely many segments correspodning to the extremal inequalities $p_1^3 - (2k+1)p_1p_2 + k(k+1)p_3 \geq 0$, for $k\geq1$; the first three such support inequalities are shown using colored dashed lines. The inflection points correspond to the uniform spectra from \cref{eq:uniform-spectrum}.}
\label{fig:mom-psd-deg-3}
\end{figure}

\begin{corollary}
Let $(x,y,z)$ denote the dual coordinates to the ordered basis
$([1^3],[21],[3])$, so that
$$
\left\langle
a[1^3]+b[21]+c[3],(x,y,z)
\right\rangle
=
a x+b y+c z.
$$
Then
$$
\mompsd_3
=
\left\{
(x,y,z):
\begin{array}{l}
z\geq0,\quad y-z\geq0,\\
x-(2k+1)y+k(k+1)z\geq0\ \text{for all }k\geq1
\end{array}
\right\}.
$$
\end{corollary}
\begin{proof}
By the general duality $\mompsd_3=(\ineqpsd_3)^\circ$, this is the intersection of
the half-spaces dual to the extremal rays found in the proposition. The normalized slice used above can be understood as follows. If $p_1(X)>0$ and
$\lambda_i=x_i/p_1(X)$, then
$$
\bigl(p_1(X)^3,\ p_1(X)p_2(X),\ p_3(X)\bigr)
=
p_1(X)^3
\left(
1,\sum_i\lambda_i^2,\sum_i\lambda_i^3
\right),
$$
so the nonzero part of the evaluation cone is exactly the cone over
$\{1\}\times K_3$.

For
$\rho\succeq 0$ with $\Tr\rho=p_1(\rho)=1$, the slice $x=1$ of $\mompsd_3$ is
$$
\left\{
(1,y,z):
z\geq0,\quad y-z\geq0,\quad
1-(2k+1)y+k(k+1)z\geq0\ \text{for all }k\geq1
\right\},
$$
where $y=\Tr(\rho^2)$ and $z=\Tr(\rho^3)$.  Equivalently,
$$
\{(y,z):(1,y,z)\in\mompsd_3\}
=
\overline{\conv\left\{
(0,0),\ (1,1),\
\left(\frac1m,\frac1{m^2}\right)\ (m=2,3,\ldots)
\right\}}.
$$
Its upper boundary is the segment $z=y$ from $(0,0)$ to $(1,1)$, and its lower
boundary consists of the segments
$$
z=
\frac{(2k+1)y-1}{k(k+1)},
\qquad
\frac1{k+1}\leq y\leq\frac1k,
\qquad k=1,2,\ldots,
$$
joining the adjacent uniform-rank points
$(1/(k+1),1/(k+1)^2)$ and $(1/k,1/k^2)$.
\end{proof}

\subsection{Some properties of degree four inequalities}

In this section we prove the extremality of a classical inequality that comes from a Hankel determinant. Consider the Hankel ray $p_1p_3-p_2^2$ is an extremal ray of $\ineqpsd_4$. This ray is particularly important in the context of entanglement detection in quantum information theory using the moments of the partial transposition \cite{elben2020mixed,carrasco2024entanglement}.

\begin{proposition}
\label{prop:delta13-extremal}
Let
$$
\Delta_{13}:=p_1p_3-p_2^2=[31]-[22].
$$
Then $\Delta_{13}\in\ineqpsd_4\setminus\ineqall_4$ and $\mathbb R_{\geq0}\Delta_{13}$
is an extremal ray of $\ineqpsd_4$.
\end{proposition}

\begin{proof}
The quantity $\Delta_{13}$ is precisely the determinant of the first
$2\times2$ odd-shifted Hankel block discussed in \cref{sec:hankel-determinants}:
$$
\Delta_{13}
=
\det\begin{pmatrix}p_1&p_2\\ p_2&p_3\end{pmatrix}.
$$
One can also directly prove that it correponds to a valid inequality on non-negative spectra:
\begin{align*}
\Delta_{13}(X)
&=
\left(\sum_i x_i\right)\left(\sum_i x_i^3\right)
-\left(\sum_i x_i^2\right)^2\\
&=
\sum_{i<j}\left(x_i x_j^3+x_j x_i^3-2x_i^2x_j^2\right)\\
&=
\sum_{i<j}x_ix_j(x_i-x_j)^2 \geq 0.
\end{align*}

The inequality is not valid on all Hermitian matrices.  Indeed, for
$X=\diag(1,-1)$ one has $p_1(X)=p_3(X)=0$ and $p_2(X)=2$, and hence
$\Delta_{13}(X)=-4<0$. Thus $\Delta_{13}\notin\ineqall_4$.

It remains to prove extremality.  Suppose that
$\Delta_{13}=g+h$ with $g,h\in\ineqpsd_4$. Using the partition basis of $V_4$ from \cref{eq:def-partition-space}, write
$$
g
=
a p_4+b p_1p_3+c p_2^2+d p_1^2p_2+e p_1^4
$$
for some $a,b,c,d,e\in\mathbb R$.  For every integer $m\geq1$, consider the
uniform spectrum
$
X_m:=I_m / m$. Its power sums are
$
p_k(X_m)=m^{1-k}$; in particular,
$$
p_1(X_m)=1,\qquad
p_2(X_m)=\frac1m,\qquad
p_3(X_m)=\frac1{m^2},\qquad
p_4(X_m)=\frac1{m^3}.
$$
Consequently $\Delta_{13}(X_m)=0$.  Since $g(X_m)$ and $h(X_m)$ are
nonnegative and have sum zero, they must both be zero.  In particular,
$$
0=g(X_m)
=
\frac{a}{m^3}+\frac{b+c}{m^2}+\frac{d}{m}+e
\qquad(m=1,2,\ldots).
$$
The polynomial
$q(t):=a t^3+(b+c)t^2+d t+e$
therefore vanishes at every $t=1/m$.  Hence
$a=d=e=0$ and $c=-b$.  Therefore
$$
g=b\bigl(p_1p_3-p_2^2\bigr)=b\Delta_{13}.
$$
Now $\Delta_{13}(\diag(1,2))=2>0$.  Because $g\geq0$ on positive
semidefinite matrices, this forces $b\geq0$.  Moreover,
$$
h=\Delta_{13}-g=(1-b)\Delta_{13},
$$
and the nonnegativity of $h$ at $\diag(1,2)$ forces $1-b\geq0$.  Thus every
decomposition of $\Delta_{13}$ into two elements of $\ineqpsd_4$ has both
summands on the ray $\mathbb R_{\geq0}\Delta_{13}$, which is precisely
the definition of an extremal ray.
 \end{proof}

It is natural to question whether all $2 \times 2$ Hankel inequalities are extremal elements in $\ineqpsd_n.$ Unfortunately, this is not true in general.
\begin{proposition}
    Let
$$
\Delta_{15}:=p_1p_5-p_3^2=[51]-[33].
$$
Then $\Delta_{15}\in\ineqpsd_6\setminus\ineqall_6$ and $\mathbb R_{\geq0}\Delta_{15}$
is \underline{not} extremal in $\ineqpsd_6$.
\end{proposition}
\begin{proof}
    We first verify that $\Delta_{15}$ does not belong to
    $\ineqall_6$.  Consider the Hermitian matrix
    $X=\diag(1,-2)$.  Its power sums are
    $p_k(X)=1+(-2)^k$, so that
    $$
    p_1(X)=-1,\qquad p_3(X)=-7,\qquad p_5(X)=-31.
    $$
    Consequently,
    $$
    \Delta_{15}(X)=p_1(X)p_5(X)-p_3(X)^2
    =(-1)(-31)-(-7)^2=-18<0.
    $$
    Since elements of $\ineqall_6$ must be nonnegative on every
    Hermitian matrix, this proves
    $\Delta_{15}\notin\ineqall_6$.

    It remains to exhibit a nontrivial decomposition inside the PSD cone.
    Set
    $$
    A:=p_1p_5-p_2p_4,
    \qquad
    B:=p_2p_4-p_3^2.
    $$
    Clearly $\Delta_{15}=A+B$.  If $x_1,\ldots,x_d\geq0$ are the
    eigenvalues of a positive semidefinite matrix, then
    \begin{align*}
    A
    &=\sum_{i<j}x_ix_j(x_i-x_j)(x_i^3-x_j^3)\geq0,
    \\
    B
    &=\sum_{i<j}x_i^2x_j^2(x_i-x_j)^2\geq0.
    \end{align*}
    Indeed, the first summand is nonnegative because $t\mapsto t^3$ is
    increasing on $[0,\infty)$, while the second identity is the usual
    $2\times2$ Hankel-determinant expansion.  Thus
    $A,B\in\ineqpsd_6$.

    Finally, in the partition basis of $V_6$ we have
    $$
    A=[51]-[42],\qquad B=[42]-[33], \qquad
    \Delta_{15}=[51]-[33].
    $$
    Comparing coefficients shows that neither $A$ nor $B$ is
    proportional to $\Delta_{15}$.  Hence the displayed decomposition is
    nontrivial, and $\mathbb R_{\geq0}\Delta_{15}$ is not an extremal ray
    of $\ineqpsd_6$.
\end{proof}

\section{Inequalities for separable matrices}\label{sec:sep}

We now switch from one positive matrix to a bipartite matrix algebra.  Let
$H_A\simeq\mathbb C^{d_A}$ and $H_B\simeq\mathbb C^{d_B}$, and let
$\rho\in \M_{d_A}\otimes \M_{d_B}$ have matrix entries $\rho_{ia,jb}$,
where $i,j$ are $A$-indices and $a,b$ are $B$-indices.

We recall the \emph{tensor trace invariant} notation from \cite{nechita2025tensor}. For $(\sigma,\tau)\in S_n^2$, we consider the local unitary invariant
$$
\Tr_{\sigma,\tau}(\rho)
:=
\Tr_{(\sigma,\tau)}(\rho,\ldots,\rho).
$$
Using indices this is
$$
\Tr_{\sigma,\tau}(\rho)
=
\sum_{\substack{i_1,\ldots,i_n\\ a_1,\ldots,a_n}}
\prod_{r=1}^n
\rho_{i_ra_r,\,i_{\sigma(r)}a_{\tau(r)}}.
$$
Equivalently,
$$
\Tr_{\sigma,\tau}(\rho)
=
\Tr\left[(P_\sigma\otimes P_\tau)\rho^{\otimes n}\right],
$$
where $P_\sigma$ permutes the $n$ copies of the $A$ system and $P_\tau$
permutes the $n$ copies of the $B$ system.  Since permutation operators
commute with tensor powers of local unitaries, these quantities are invariant
under local unitary conjugation:
$$
\Tr_{\sigma,\tau}\bigl((U\otimes V)\rho(U\otimes V)^*\bigr)
=
\Tr_{\sigma,\tau}(\rho).
$$
For an element
\begin{equation}\label{eq:def-w-bipartite}
w=\sum_{\sigma,\tau\in S_n}c_{\sigma,\tau}(\sigma,\tau)
\in\mathbb R[S_n\times S_n],
\end{equation}
we need a real-valued polynomial in order to discuss inequalities.  Individual
tensor trace invariants need not be real, but for Hermitian $\rho$ they satisfy
$$
\overline{\Tr_{\sigma,\tau}(\rho)}
=\Tr_{\sigma^{-1},\tau^{-1}}(\rho).
$$
In what follows, we write for $w$ as in \cref{eq:def-w-bipartite}
$$
\Tr_w(\rho):=
\operatorname{Re}\!\left(
\sum_{\sigma,\tau\in S_n}c_{\sigma,\tau}
\Tr_{\sigma,\tau}(\rho)
\right).
$$

\subsection{Dimension-free separable cone}

For local dimensions $d_A,d_B$, the unnormalized separable cone is
$$
\Sep(H_A:H_B)
=
\cone\{X\otimes Y:X\in \Mpsd_{d_A},\ Y\in \Mpsd_{d_B}\}.
$$
The normalized separable states are the elements of this cone with trace $1$.

\begin{definition}
The dimension-free separability-valid trace-invariant cone in degree $n$ is
$$
\ineqsep_n
=
\left\{
w\in\mathbb R[S_n\times S_n]:
\begin{array}{l}
\Tr_w(\rho)\geq0\ \text{for every }d_A,d_B\geq1\\
\text{and every }\rho\in\Sep(\mathbb C^{d_A}:\mathbb C^{d_B})
\end{array}
\right\}.
$$
The corresponding separable moment cone is
$$
\momsep_n
:=
\overline{
\cone\left\{
\bigl(\operatorname{Re}\Tr_{\sigma,\tau}(\rho)\bigr)_{
\sigma,\tau\in S_n}:
\rho\in\Sep(\mathbb C^{d_A}:\mathbb C^{d_B}),\
d_A,d_B\geq1
\right\}}.
$$
\end{definition}

The closure does not change the dual cone, but it makes the moment cone closed.
With the evident pairing between coefficients
$w=\sum_{\sigma,\tau}c_{\sigma,\tau}(\sigma,\tau)$ and evaluation vectors, we
have
$$
\ineqsep_n
=
\left(\momsep_n\right)^\circ.
$$
Indeed, $\Tr_w(\rho)$ is precisely the pairing of $w$ with the evaluation
vector
$$
\bigl(\operatorname{Re}\Tr_{\sigma,\tau}(\rho)\bigr)_{
\sigma,\tau\in S_n},
$$
and the condition of nonnegativity is unchanged after passing to the conic
hull and then to its closure.
This is the exact analogue of the duality
$\ineqpsd_n=(\mompsd_n)^\circ$ in the one-matrix problem.  The main difference is
that the evaluation vector is formed from sums of product positive
semidefinite matrices.  Hence positivity on product rays
$X\otimes Y$ is only a necessary condition for membership in $\ineqsep_n$;
one must test the whole separable cone.

\begin{remark}
The polynomial local-unitary invariants of bipartite mixed states are
generated by such permutation contractions, up to the usual
dimension-dependent relations.  This is the invariant-theoretic viewpoint of
\cite{GrasslRoettelerBeth1998,Makhlin2002,Szalay2011,carrozza2026tensor}; Huber's
trace-polynomial formalism \cite{Huber2021PositiveMapsTracePolynomials}
gives a systematic translation between permutations, partial traces, matrix
products, and invariant trace polynomials.
\end{remark}

The embeddings in the next subsection connect the one-matrix notation of
\cref{sec:all,sec:psd} to the bipartite group algebra: a partition monomial
$[\lambda]\in V_n$ is sent to the pair
$(\alpha_\lambda,\alpha_\lambda)$ or $(\alpha_\lambda,\alpha_\lambda^{-1})$
in $\mathbb R[S_n\times S_n]$, where $\alpha_\lambda\in S_n$ has cycle type
$\lambda$.

\subsection{Embeddings of the positive semidefinite cone}

Write $\PPT$ for the set of states that remain positive semidefinite under
partial transposition; every separable state is $\PPT$.

\begin{proposition}
\label{prop:psd-embeddings}
There are two natural embeddings of the one-matrix cone $\ineqpsd_n$ into
$\ineqsep_n$.  For a partition
$\lambda=(\lambda_1\geq\lambda_2\geq\cdots\geq\lambda_r)$ of $n$, let
$\alpha_\lambda\in S_n$ be the permutation with cycle type $\lambda$ obtained
by writing the cycles with the largest blocks first and with the elements in
each cycle ordered increasingly:
$$
\alpha_\lambda
=
(1\,2\,\cdots\,\lambda_1)
(\lambda_1+1\,\lambda_1+2\,\cdots\,\lambda_1+\lambda_2)
\cdots .
$$
Let $f=\sum_{\lambda\vdash n}a_\lambda[\lambda]\in\ineqpsd_n$ and define
$$
\iota_+(f)
:=
\sum_{\lambda\vdash n}
a_\lambda(\alpha_\lambda,\alpha_\lambda)
\qquad \text{and} \qquad
\iota_\Gamma(f)
:=
\sum_{\lambda\vdash n}
a_\lambda(\alpha_\lambda,\alpha_\lambda^{-1}).
$$
Then
$$
\iota_+(\ineqpsd_n)\subseteq\ineqsep_n
\qquad \text{and} \qquad
\iota_\Gamma(\ineqpsd_n)\subseteq\ineqsep_n.
$$
\end{proposition}

\begin{proof}
For the first embedding, regard $\rho$ as a single matrix on
$H_A\otimes H_B$.  Under the identification
$(H_A\otimes H_B)^{\otimes n}\simeq H_A^{\otimes n}\otimes H_B^{\otimes n}$,
the global permutation of the $n$ copies is the tensor product of the two
local permutations.  Hence, for every $\lambda\vdash n$,
$$
p_\lambda(\rho)
=
\Tr_{\alpha_\lambda,\alpha_\lambda}(\rho).
$$
If $\rho$ is separable, then $\rho\succeq 0$, so
$$
\Tr_{\iota_+(f)}(\rho)
=
\sum_{\lambda\vdash n}a_\lambda p_\lambda(\rho)
\geq0.
$$
Thus $\iota_+(f)\in\ineqsep_n$.  This is the diagonal embedding, corresponding
to the choice $\beta=\alpha$.

For the second embedding, use that every separable matrix is $\PPT$.  Let
$\rho^{\Gamma_B}$ denote the partial transpose on the $B$ tensor factor.  If
$\rho$ is separable, then $\rho^{\Gamma_B}\succeq 0$.  Since partial transposition
is self-adjoint for the trace pairing and the transpose of a permutation
operator is the inverse permutation operator, for every $\lambda\vdash n$,
$$
p_\lambda(\rho^{\Gamma_B})
=
\Tr_{\alpha_\lambda,\alpha_\lambda^{-1}}(\rho).
$$
Therefore
$$
\Tr_{\iota_\Gamma(f)}(\rho)
=
\sum_{\lambda\vdash n}a_\lambda p_\lambda(\rho^{\Gamma_B})
\geq0,
$$
and $\iota_\Gamma(f)\in\ineqsep_n$.  This is the $\PPT$ embedding,
corresponding to the choice $\beta=\alpha^{-1}$.
\end{proof}

\begin{remark}
An important example of the preceding construction is the partial-transpose
moment inequality used in entanglement detection
\cite{elben2020mixed,carrasco2024entanglement}.  If $\rho$ is separable, then
$\rho^{\Gamma_B}\geq0$, so applying the Hankel inequality
$$
\Delta_{13}(X)=p_1(X)p_3(X)-p_2(X)^2\geq0
$$
from \cref{sec:hankel-determinants} to $X=\rho^{\Gamma_B}$ gives
$$
p_1(\rho^{\Gamma_B})p_3(\rho^{\Gamma_B})
-p_2(\rho^{\Gamma_B})^2\geq0.
$$
Since $p_1(\rho^{\Gamma_B})=\Tr(\rho)=1$, this is equivalently
$$
p_3(\rho^{\Gamma_B})\geq p_2(\rho^{\Gamma_B})^2.
$$
Thus this separability inequality is precisely the image under
$\iota_\Gamma$ of the positive-semidefinite Hankel inequality
$\Delta_{13}\in\ineqpsd_4$.
\end{remark}

Summarizing the chain of inclusions established in the current and previous sections, we have shown that
$$
\ineqall_n \subseteq \ineqpsd_n,
\qquad
\iota_+(\ineqpsd_n) \subseteq \ineqsep_n,
\qquad
\iota_\Gamma(\ineqpsd_n) \subseteq \ineqsep_n.
$$
Positivity on product rays $X\otimes Y$ gives necessary conditions for
separability, but the separable cone is strictly stronger than positivity of
the underlying one-matrix spectra.  We next make the dimension-free
parametrization more economical by identifying its universal quotient
coordinates, before turning to concrete criteria derived from realignment
moments.

\subsection{An economical parametrization of
\texorpdfstring{$\ineqsep_n$}{Ineq sep n}}

In \cref{sec:psd}, we have seen that the parametrization of the moments of a single matrix by permutations is overcomplete, in the sense that permutations belonging to the same conjugacy class yield the same trace invariant. Hence, we have decided to parametrize trace invariants by conjugacy classes which, in turn, are in one-to-one correspondence with partitions of the degree.

In this subsection, we take the same approach for the tensor trace invariants of bipartite Hermitian matrices. Let us point out that the parametrization discussed in this section has nothing to do with separability, it is simply a convenient way to parametrize trace invariants of bipartite Hermitian matrices by conjugacy classes. The economical coordinates for the degree-$n$ separable inequalities are
the orbits considered by Bogaerts and Dukes
\cite{bogaerts2014semidefinite}.  Let
$G_n=S_n\times C_2$ act on $S_n^2$ by simultaneous conjugation and
simultaneous inversion:
\begin{align*}
\pi\cdot(\sigma,\tau)
&=(\pi\sigma\pi^{-1},\pi\tau\pi^{-1}),\\
\kappa\cdot(\sigma,\tau)
&=(\sigma^{-1},\tau^{-1}).
\end{align*}
These two actions commute.  We write
$$
\mathcal O_n:=S_n^2/G_n,
\qquad
b_n:=|\mathcal O_n|.
$$
For each $\omega\in\mathcal O_n$, choose a representative
$(\sigma_\omega,\tau_\omega)\in\omega$ and define
$$
K_\omega(\rho):=\operatorname{Re}
\Tr_{\sigma_\omega,\tau_\omega}(\rho).
$$
This is independent of the representative.  Indeed, the tensor-trace
contractions satisfy, for all $\rho=\rho^*$:
\begin{equation*}
\Tr_{\pi\sigma\pi^{-1},\,\pi\tau\pi^{-1}}(\rho)
=\Tr_{\sigma,\tau}(\rho),
\qquad
\Tr_{\sigma^{-1},\tau^{-1}}(\rho)
=\overline{\Tr_{\sigma,\tau}(\rho)}.
\end{equation*}

There is an equivalent way to enumerate these coordinates that is useful
in practice.  For every partition $\lambda\vdash n$, fix a representative
$\gamma_\lambda\in S_n$ of cycle type $\lambda$.  Given
$(\sigma,\tau)\in S_n^2$, conjugate both permutations so that
$\sigma$ becomes $\gamma_\lambda$, and denote the resulting second
permutation by $\widetilde\tau$.  The ambiguity in this choice is exactly
conjugation by $Z_\lambda:=C_{S_n}(\gamma_\lambda)$.
Consequently, simultaneous-conjugacy classes of pairs are parametrized by
$(\lambda,[\widetilde\tau]_{Z_\lambda})$, where
$[\widetilde\tau]_{Z_\lambda}$ denotes a $Z_\lambda$-conjugacy
orbit.  One then identifies the resulting data under simultaneous
inversion.  In this partition--permutation notation, define
$$
K_{\lambda;\tau}(\rho):=K_{\omega(\lambda,\tau)}(\rho)
=\operatorname{Re}\Tr_{\gamma_\lambda,\tau}(\rho),
$$
where $\omega(\lambda,\tau)$ is the $G_n$-orbit of
$(\gamma_\lambda,\tau)$.  Thus $\tau$ is understood modulo
$Z_\lambda$-conjugacy and the simultaneous-inversion identification, and
$K_{\lambda;\tau}$ is the same reduced invariant as $K_\omega$ in these
coordinates.

Let $(\ineqsep_n)^{\mathrm{quot}}$ denote the image of $\ineqsep_n$ after
factoring out universal evaluation identities, that is, relations between
permutation contractions that vanish for every local dimension and every
separable $\rho$.  The preceding orbit construction gives the following
exact parametrization.

\begin{theorem}
\label{thm:minimal-parametrization-exact-size}
The minimal parametrization of the cone $\ineqsep_n$ is given by
\begin{equation*}
(\ineqsep_n)^{\mathrm{quot}}
\cong
\left\{
c\in\mathbb R^{\mathcal O_n}:
\sum_{\omega\in\mathcal O_n}c_\omega K_\omega(\rho)\geq0
\ \text{for every }d_A,d_B
\ \text{and }\rho\in\Sep(\mathbb C^{d_A}:\mathbb C^{d_B})
\right\}.
\end{equation*}
By \cite[Proposition~3.2]{bogaerts2014semidefinite}, its number of real coordinates is
\begin{equation*}
b_n=
\frac{1}{2n!}\sum_{\alpha\in S_n}
\left[
\left(\sum_{\chi}\chi(\alpha)^2\right)^2
+
\left(\sum_{\chi}\chi(\alpha)\right)^3
\right],
\end{equation*}
where both inner sums run over the irreducible characters of $S_n$.
\end{theorem}

The resulting exact dimensions form the sequence \cite{A241584}:
$$
\begin{array}{c|rrrrrrrr}
n&1&2&3&4&5&6&7&8\\
\hline
b_n&1&4&11&43&155&761&4\,043&27\,190
\end{array}
$$

For $n=2$, the 4 conjugacy classes are trivial. We display the tensor invariants in \cref{fig:tensor-invariants-2}.

\begin{figure}[htb]
    \centering
    \includegraphics[width=0.2\linewidth]{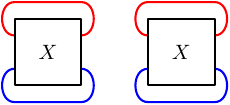}\qquad
    \includegraphics[width=0.2\linewidth]{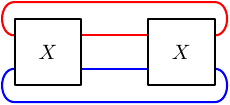}\qquad
    \includegraphics[width=0.2\linewidth]{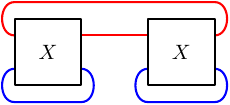}\qquad
    \includegraphics[width=0.2\linewidth]{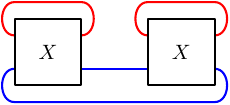}
    \caption{The four tensor trace invariants for $n=2$: $t_X^2 = (\Tr X)^2$, $p_2(X) = \Tr(X^2)$, $\pi_A = \Tr(X_A^2)$, and $\pi_B = \Tr(X_B^2)$.}
    \label{fig:tensor-invariants-2}
\end{figure}

For $n=3$, write
\[
e=\mathrm{id},\qquad s=(12),\qquad t=(13),\qquad u=(23),
\qquad r=(123),\qquad q=(132)=r^{-1}.
\]
The $36$ pairs in $S_3^2$ split into the following $11$ $G_3$-orbits,
where $G_3$ acts by simultaneous conjugation and simultaneous inversion:
\begin{align*}
\mathcal C_{111;e}
 &=\{(e,e)\},\\[2pt]
\mathcal C_{111;(12)}
 &=\{(e,s),(e,t),(e,u)\},\\[2pt]
\mathcal C_{111;(123)}
 &=\{(e,r),(e,q)\},\\[6pt]
\mathcal C_{21;e}
 &=\{(s,e),(t,e),(u,e)\},\\[2pt]
\mathcal C_{21;(12)}
 &=\{(s,s),(t,t),(u,u)\},\\[2pt]
\mathcal C_{21;(23)}
 &=\{(s,t),(s,u),(t,s),(t,u),(u,s),(u,t)\},\\[2pt]
\mathcal C_{21;(123)}
 &=\{(s,r),(s,q),(t,r),(t,q),(u,r),(u,q)\},\\[6pt]
\mathcal C_{3;e}
 &=\{(r,e),(q,e)\},\\[2pt]
\mathcal C_{3;(12)}
 &=\{(r,s),(r,t),(r,u),(q,s),(q,t),(q,u)\},\\[2pt]
\mathcal C_{3;(123)}
 &=\{(r,r),(q,q)\},\\[2pt]
\mathcal C_{3;(132)}
 &=\{(r,q),(q,r)\}.
\end{align*}
Their cardinalities of these classes are $1,3,2,3,3,6,6,2,6,2,2$, which sum to $3!^2=36$.

We represent graphically a representative from each class in
\cref{fig:tensor-invariants-3}.  The representatives are ordered from left
to right and then from top to bottom according to the list above.
\begin{figure}[htb]
    \centering
    \begin{tabular}{@{}c@{\qquad}c@{\qquad}c@{}}
        \includegraphics[width=0.30\linewidth]{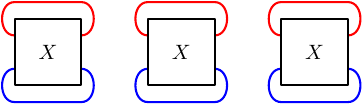} &
        \includegraphics[width=0.30\linewidth]{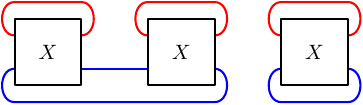} &
        \includegraphics[width=0.30\linewidth]{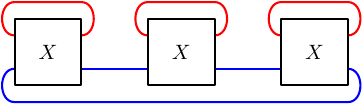} \\[6pt]
        \includegraphics[width=0.30\linewidth]{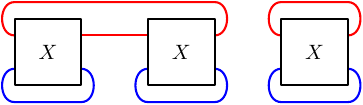} &
        \includegraphics[width=0.30\linewidth]{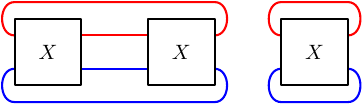} &
        \includegraphics[width=0.30\linewidth]{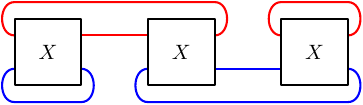} \\[6pt]
        \includegraphics[width=0.30\linewidth]{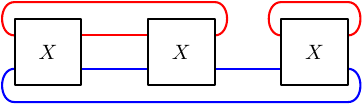} &
        \includegraphics[width=0.30\linewidth]{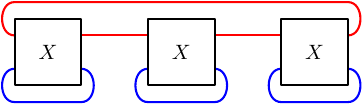} &
        \includegraphics[width=0.30\linewidth]{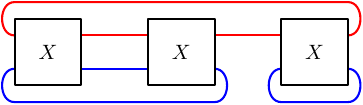} \\[6pt]
        \multicolumn{3}{c}{
            \includegraphics[width=0.30\linewidth]{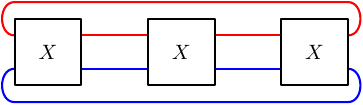}
            \qquad
            \includegraphics[width=0.30\linewidth]{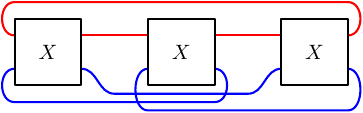}}
    \end{tabular}
    \caption{Representatives of the eleven $G_3$-orbits of pairs of
    permutations in $S_3^2$, ordered according to
    $\mathcal C_{111;e}$, $\mathcal C_{111;(12)}$,
    $\mathcal C_{111;(123)}$, $\mathcal C_{21;e}$,
    $\mathcal C_{21;(12)}$, $\mathcal C_{21;(23)}$,
    $\mathcal C_{21;(123)}$, $\mathcal C_{3;e}$,
    $\mathcal C_{3;(12)}$, $\mathcal C_{3;(123)}$,
    $\mathcal C_{3;(132)}$. The first argument (integer partition) corresponds to Alice's permutation, while the second one to Bob's permutation.}
    \label{fig:tensor-invariants-3}
\end{figure}

\subsection{Structure and extremal rays of \texorpdfstring{$\ineqsep_2$}{Ineq sep 2}}

The orbit parametrization of \cref{thm:minimal-parametrization-exact-size}
can be made completely explicit in degree two.

From the previous subsection, we recall the four tensor trace invariants of degree two (see \cref{fig:tensor-invariants-2}):
$$
\begin{array}{c|c|c|c}
\text{quotient coordinate}&(\sigma,\tau)&(r_A,r_B)&
K_{\lambda;\tau}(\rho)\\
\hline
((2),[s])&(s,s)&(1,1)&\Tr(\rho^2)\\
((2),[e])&(s,e)&(1,2)&\Tr(\rho_A^2)\\
((1,1),[s])&(e,s)&(2,1)&\Tr(\rho_B^2)\\
((1,1),[e])&(e,e)&(2,2)&(\Tr\rho)^2.
\end{array}
$$
Here, $e$ is the identity permutation, $s$ is the transposition $(12)$, $\rho_A=\Tr_B\rho$, $\rho_B=\Tr_A\rho$, and $r_A,r_B$ denote the
numbers of cycles of the two local permutations.  For the remainder of the
subsection, abbreviate
$$
t:=\Tr(\rho),
\qquad
\pi:=p_2(\rho)=\Tr(\rho^2),
\qquad
\pi_A:=\Tr(\rho_A^2),
\qquad
\pi_B:=\Tr(\rho_B^2).
$$

Let us normalize states ($t=1$) and introduce
$$
\mathcal K_2
:=
\left\{(\pi,\pi_A,\pi_B):
\rho\in\Sep(\mathbb C^{d_A}:\mathbb C^{d_B}),\quad
\Tr(\rho)=1,\quad d_A,d_B\geq1\right\}.
$$
\begin{proposition}
The closed convex hull of $\mathcal K_2$ is
\begin{equation}
\label{eq:ineqsep2-evaluation-tetrahedron}
\overline{\conv(\mathcal K_2)}
=
P_2
:=\left\{(\pi,\pi_A,\pi_B):
\pi\geq0,\quad \pi_A \geq \pi,\quad \pi_B \geq \pi,\quad
1+\pi \geq \pi_A+\pi_B\right\}.
\end{equation}
Equivalently, $P_2$ is the tetrahedron with vertices
\[
v_0=(0,0,0),
\qquad
v_A=(0,1,0),
\qquad
v_B=(0,0,1),
\qquad
v_{AB}=(1,1,1).
\]
\end{proposition}

\begin{proof}
We first prove that $\mathcal K_2\subseteq P_2$.  Let
$(\pi,\pi_A,\pi_B)\in\mathcal K_2$, and let $\rho$ be a separable
density matrix that gives this triple.  By separability, $\rho$ has a
finite decomposition
$$
\rho=\sum_i t_i X_i\otimes Y_i,
\qquad
t_i\geq0,
\qquad
\sum_i t_i=1,
$$
where $X_i$ and $Y_i$ are local density matrices; in particular,
$$
X_i,Y_i \succeq 0,
\qquad
\Tr(X_i)=\Tr(Y_i)=1.
$$
Define $A_{ij}:=\Tr(X_iX_j)$ and $B_{ij}:=\Tr(Y_iY_j)$.
Clearly, $A_{ij},B_{ij}\in[0,1]$.
Taking partial traces in the separable decomposition and using
$\Tr(X_i)=\Tr(Y_i)=1$ gives
$$
\rho_A=\sum_i t_iX_i,
\qquad
\rho_B=\sum_i t_iY_i.
$$
We can express now the purities as follows:
\begin{align*}
\pi=\Tr(\rho^2)
&=\sum_{i,j}t_it_jA_{ij}B_{ij},\\
\pi_A=\Tr(\rho_A^2)
&=\sum_{i,j}t_it_jA_{ij},\\
\pi_B=\Tr(\rho_B^2)
&=\sum_{i,j}t_it_jB_{ij}.
\end{align*}
Since $\sum_{i,j}t_it_j=(\sum_i t_i)^2=1$, it follows that
\begin{align*}
\pi
&=\sum_{i,j}t_it_jA_{ij}B_{ij}\geq0,\\
\pi_A-\pi
&=\sum_{i,j}t_it_jA_{ij}(1-B_{ij})\geq0,\\
\pi_B-\pi
&=\sum_{i,j}t_it_j(1-A_{ij})B_{ij}\geq0,\\
1+\pi-\pi_A-\pi_B
&=\sum_{i,j}t_it_j(1-A_{ij})(1-B_{ij})\geq0.
\end{align*}
Every summand on the right-hand sides is nonnegative.  Thus
$\mathcal K_2\subseteq P_2$.  Because $P_2$ is an intersection of closed
half-spaces, it is closed and convex, and hence
\[
\overline{\conv(\mathcal K_2)}\subseteq P_2.
\]

It remains to prove the reverse inclusion.  Let
$\mu_n:=I_n/n$ denote the maximally mixed state on $\mathbb C^n$, and
let $P_A$ and $P_B$ be rank-one density matrices.  All the states in the
following table are product states and therefore belong to the class
used in the definition of $\mathcal K_2$:
$$
\begin{array}{c|c}
\text{product state}&(\pi,\pi_A,\pi_B)\\
\hline
\mu_n\otimes\mu_n
    &(n^{-2},n^{-1},n^{-1})\\
P_A\otimes\mu_n
    &(n^{-1},1,n^{-1})\\
\mu_n\otimes P_B
    &(n^{-1},n^{-1},1)\\
P_A\otimes P_B
    &(1,1,1).
\end{array}
$$
Here we used $\Tr(\omega_n^2)=1/n$ and
$\Tr(P_A^2)=\Tr(P_B^2)=1$.  Letting $n\to\infty$ in the first three
rows, and using the last row directly, shows that
$$
v_0,v_A,v_B,v_{AB}\in\overline{\mathcal K_2}
\subseteq\overline{\conv(\mathcal K_2)}.
$$

Finally, the four affine functions defining $P_2$ give the barycentric
identity
$$
(\pi,\pi_A,\pi_B)
=
(1+\pi-\pi_A-\pi_B)v_0
+(\pi_A-\pi)v_A+(\pi_B-\pi)v_B+\pi\,v_{AB}.
$$
The four coefficients on the right sum to one, and they are nonnegative
precisely when $(\pi,\pi_A,\pi_B)\in P_2$.  Thus $P_2$ is the convex
hull of the four displayed vertices.  Since
$\overline{\conv(\mathcal K_2)}$ is convex and contains all four
vertices, we conclude that $P_2\subseteq\overline{\conv(\mathcal K_2)}$. Together with the first inclusion, this proves
\cref{eq:ineqsep2-evaluation-tetrahedron}.
\end{proof}

\begin{remark}
The assumption that $\rho$ is separable is essential.  For example, consider the Bell (or maximally entangled) state and the corresponding density matrix
$$
\ket{\Phi^+}:=\frac{\ket{00}+\ket{11}}{\sqrt2},
\qquad
\rho_{\Phi}:=\ketbra{\Phi^+}{\Phi^+}.
$$
This state is entangled, and it is pure, so
$\pi=\Tr(\rho_{\Phi}^2)=1$.  On the other hand,
$(\rho_{\Phi})_A=(\rho_{\Phi})_B=I_2/2$, whence
$$
\pi_A=\pi_B=\Tr\!\left((I_2/2)^2\right)=\frac12.
$$
Therefore
$$
\pi_A-\pi=\pi_B-\pi=-\frac12<0.
$$
Thus two of the inequalities defining $P_2$ fail for this entangled
state; they are consequences of separability, and are not valid for
arbitrary bipartite states.
\end{remark}

\begin{remark}
The closure in
\cref{eq:ineqsep2-evaluation-tetrahedron} cannot be omitted.  Indeed,
let $\rho$ be any density matrix on the finite-dimensional space
$\mathbb C^{d_A}\otimes\mathbb C^{d_B}$, and let
$\lambda_1,\ldots,\lambda_r>0$ be its nonzero eigenvalues.  Since
$\sum_{k=1}^r\lambda_k=1$, the Cauchy--Schwarz inequality gives
$$
1=\left(\sum_{k=1}^r\lambda_k\right)^2
\leq r\sum_{k=1}^r\lambda_k^2
=r\Tr(\rho^2).
$$
Consequently,
$$
\pi=\Tr(\rho^2)\geq\frac1r
\geq\frac1{d_Ad_B}>0.
$$
In particular, the first coordinate of every point of $\mathcal K_2$ is
strictly positive.  The same remains true for every finite convex
combination of points of $\mathcal K_2$: if
$$
x=\sum_{\ell=1}^m s_\ell x_\ell,
\qquad
s_\ell\geq0,
\qquad
\sum_{\ell=1}^m s_\ell=1,
\qquad
x_\ell\in\mathcal K_2,
$$
then the first coordinate of $x$ is
$\sum_{\ell=1}^m s_\ell\pi_\ell>0$.  Hence no point of
$\conv(\mathcal K_2)$ can have first coordinate equal to zero.

On the other hand, the three vertices $v_0$, $v_A$, and $v_B$ of $P_2$
all have first coordinate zero.  The product-state sequences constructed
in the proof show that these vertices belong to
$\overline{\mathcal K_2}$, but the preceding argument shows that they do
not belong to $\conv(\mathcal K_2)$.  Therefore
$$
\conv(\mathcal K_2)
\subsetneq
\overline{\conv(\mathcal K_2)}=P_2.
$$
Thus the closure is essential since it accounts for boundary points obtained
by allowing the local dimensions to tend to infinity.
\end{remark}
\begin{remark}
The last inequality in \cref{eq:ineqsep2-four-extremal-rays}
$$1 + \pi \geq \pi_A + \pi_B \iff 1 + \Tr(\rho_{AB}^2) \geq \Tr(\rho_{A}^2) + \Tr(\rho_{B}^2) $$
is precisely \emph{subadditivity of the linear entropy} (or the Tsallis $q=2$ entropy), see \cite[Appendix C]{rungta2001universal} or \cite[Theorem 2]{audenaert2007subadditivity}.
\end{remark}

Going back to un-normalized separable matrices, we can rephrase the previous result as follows.

\begin{corollary}
\label{cor:ineqsep2-extremal-rays}
In the partition--permutation quotient basis, the dimension-free cone
$\ineqsep_2$ has exactly four extremal rays, represented by
\begin{equation}
\label{eq:ineqsep2-four-extremal-rays}
\begin{gathered}
\mathbb R_{\geq0}K_{(2);s},
\qquad
\mathbb R_{\geq0}\bigl(K_{(2);e}-K_{(2);s}\bigr),\\[2pt]
\mathbb R_{\geq0}\bigl(K_{(1,1);s}-K_{(2);s}\bigr),
\qquad
\mathbb R_{\geq0}\bigl(K_{(1,1);e}+K_{(2);s}-K_{(2);e}-K_{(1,1);s}\bigr).
\end{gathered}
\end{equation}
Their evaluations are, respectively,
\[
\pi \geq 0,
\qquad
\pi_A \geq \pi,
\qquad
\pi_B \geq \pi,
\qquad
t^2+\pi \geq \pi_A +\pi_B.
\]

More explicitly, an arbitrary quotient element
\[
w=c_{t^2}K_{(1,1);e}+c_\pi K_{(2);s}
+c_{\pi_A}K_{(2);e}+c_{\pi_B}K_{(1,1);s}
\]
belongs to $\ineqsep_2$ if and only if
\begin{equation}
\label{eq:ineqsep2-facet-coefficient-test}
c_{t^2}\geq0,
\qquad
c_{t^2}+c_{\pi_A}\geq0,
\qquad
c_{t^2}+c_{\pi_B}\geq0,
\qquad
c_{t^2}+c_\pi+c_{\pi_A}+c_{\pi_B}\geq0.
\end{equation}
\end{corollary}

\begin{proof}
Let $\rho$ be a separable positive semidefinite operator.  If $\rho=0$,
then every homogeneous degree-two invariant vanishes on $\rho$, so there
is nothing to prove.  Suppose, therefore, that $\rho\neq0$.  Then
$t=\Tr(\rho)>0$, and $\rho/t$ is a separable density matrix.  Since all
four invariants are homogeneous of degree two,
\[
\Tr_w(\rho)
=t^2\Tr_w\!\left(\frac{\rho}{t}\right)
=t^2L_w\!\left(
\frac{\pi}{t^2},\frac{\pi_A}{t^2},\frac{\pi_B}{t^2}
\right),
\]
where
\[
L_w(\pi,\pi_A,\pi_B)
=c_{t^2}+c_\pi\pi+c_{\pi_A}\pi_A+c_{\pi_B}\pi_B.
\]
It follows that $w$ is nonnegative on all separable positive
semidefinite operators, in all local dimensions, if and only if $L_w$ is
nonnegative on $\mathcal K_2$.  Because $L_w$ is affine and continuous,
it is nonnegative on $\mathcal K_2$ if and only if it is nonnegative on
$\overline{\conv(\mathcal K_2)}$.  By
\eqref{eq:ineqsep2-evaluation-tetrahedron}, the latter set is $P_2$.

Every point of $P_2$ is a convex combination of
$v_0,v_A,v_B,v_{AB}$.  Hence an affine function is nonnegative on $P_2$
if and only if it is nonnegative at these four vertices.  The
corresponding values of $L_w$ are
\[
L_w(v_0)=c_{t^2},
\qquad
L_w(v_A)=c_{t^2}+c_{\pi_A},
\qquad
L_w(v_B)=c_{t^2}+c_{\pi_B},
\qquad
L_w(v_{AB})=c_{t^2}+c_\pi+c_{\pi_A}+c_{\pi_B}.
\]
This proves \eqref{eq:ineqsep2-facet-coefficient-test}.

Set these values equal to
\[
\lambda_0=c_{t^2},
\qquad
\lambda_A=c_{t^2}+c_{\pi_A},
\qquad
\lambda_B=c_{t^2}+c_{\pi_B},
\qquad
\lambda_{AB}=c_{t^2}+c_\pi+c_{\pi_A}+c_{\pi_B}.
\]
A direct coefficient comparison gives the unique decomposition
\begin{align*}
w={}&
\lambda_{AB}K_{(2);s}
+\lambda_A\bigl(K_{(2);e}-K_{(2);s}\bigr)
+\lambda_B\bigl(K_{(1,1);s}-K_{(2);s}\bigr)\\
&+\lambda_0\bigl(K_{(1,1);e}+K_{(2);s}-K_{(2);e}-K_{(1,1);s}\bigr).
\end{align*}
Hence \eqref{eq:ineqsep2-facet-coefficient-test} is equivalent to a unique
nonnegative linear combination of the four generators in
\eqref{eq:ineqsep2-four-extremal-rays}.  The map
$w\mapsto(\lambda_0,\lambda_A,\lambda_B,\lambda_{AB})$ is an invertible
linear map from $\ineqsep_2$ onto $\mathbb R_{\geq0}^4$.  Since the
positive orthant has exactly its four coordinate axes as extremal rays, the
four displayed rays are both extremal and exhaustive.
\end{proof}

\section{Inequalities from few moments of realignment}\label{sec:realignment}
In this section, we show that the \emph{realignment criterion} \cite{rudolph2000separability,chen2003matrix} for quantum states can be used to derive some of the elements of the convex cone $\ineqsep_n$, the separable inequalities that are agnostic of the underlying dimension. The primary issue with considering the realignment test itself is that the trace norm considered in realignment is not an element of a group algebra, so the criterion cannot be directly used to formulate such an inequality. However, we show that its \emph{even} singular-value moments are
bipartite tensor trace invariants of the form
$$
\Tr_{\sigma,\tau}(\rho)
=
\Tr\left[(P_\sigma\otimes P_\tau)\rho^{\otimes n}\right],
$$
with different permutations on the two local tensor legs.  They are therefore
genuinely mixed local-unitary invariants in the sense of the preceding
section. We first recall the realignment criterion for separability.

\noindent \emph{Realignment criterion}. The realignment (or
reshuffling) map is given by
\begin{align*}
\Rn:\M_{d_A}\otimes \M_{d_B} &\longrightarrow
\M_{d_A^2,d_B^2}\\
\Rn(X)_{ik,jl} &= X_{ij,kl} \qquad i,k \in [d_A],\ j,l \in [d_B].
\end{align*}
Graphically, this reads as
\begin{center}
    \includegraphics{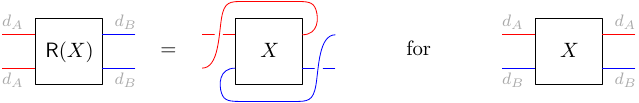}
\end{center}
Thus, this map $\Rn(X)$ is obtained by grouping the two ``$A$'' indices (top wires) and the two
``$B$'' indices (bottom wires) together to construct a rectangular matrix.
For a matrix $X$ and $1\leq p<\infty$, we use the Schatten norms
\begin{equation}
\label{eq:schatten-norm}
\|X\|_p=\left(\sum_i s_i(X)^p\right)^{1/p},
\end{equation}
where $s_i(X)$ are the singular values of $X$.  In particular, $\|X\|_1$ is
the trace norm and $\|X\|_2$ is the Hilbert--Schmidt norm.
The realignment criterion (sometimes called the computable cross norm, or CCNR, criterion) states that any separable state $X \in \mathcal{M}_d \otimes \mathcal{M}_d$ satisfies the following homogenous inequality \cite{rudolph2000separability,chen2003matrix}: $$\| \Rn(X) \|_1 \leq \operatorname{Tr}(X),$$
where the trace norm is the sum of the singular values of $\Rn(X)$.
Thus, any violation of this criterion certifies that the state is entangled. The realignment criterion is independent of the PPT criterion, as in neither one implies the other, in effect, this means that the realignment criterion can detect states that are PPT entangled \cite{chen2003matrix}.

It was later shown in \cite{zhang2008entanglement} that the realignment criterion can be improved by \emph{centering} the operator, yielding the so-called \emph{enhanced entanglement criterion}. We first introduce notation for some particular local unitary trace invariants. Let
$$
t_X=\Tr(X),
\qquad
X_A=\Tr_B(X),
\qquad
X_B=\Tr_A(X),
$$
and
$$
\pi_A(X)=\Tr(X_A^2),
\qquad
\pi_B(X)=\Tr(X_B^2),
\qquad
p_2(X)=\Tr(X^2),
\qquad
\pi_{\mathrm{mix}}(X)=\Tr\!\left(X(X_A\otimes X_B)\right).
$$
The first three are degree two local unitary trace invariants of $X$ and correspond to the local, respectively, global purities of $X$:
$$\pi_A(X) = \Tr_{(12),(1)(2)}(X),\qquad \pi_B(X) = \Tr_{(1)(2),(12)}(X), \qquad p_2(X) = \Tr_{(12),(12)}(X).$$
The last invariant, $\pi_{\mathrm{mix}}(X)$, is a degree three local unitary trace invariant of $X$:
$$\pi_{\mathrm{mix}}(X) = \Tr_{(1)(23),(12)(3)}(X)$$
which can be seen from the following graphical representaiton:
\begin{center}
    \includegraphics{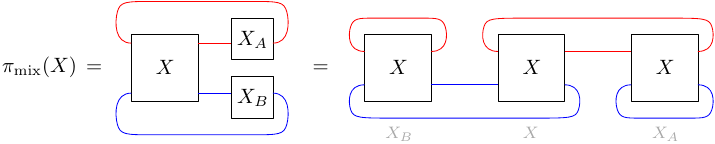}
\end{center}

\begin{theorem}[Enhanced realignment criterion \cite{zhang2008entanglement}]
For any bipartite matrix $X \in \M_{d_A} \otimes \M_{d_B}$, define the following centered operator, homogeneous in $X$:
\begin{equation}\label{eq:def-tildeX}
\tilde X := t_X X - X_A \otimes X_B.
\end{equation}
Then, for any separable state $X$, we have
$$\|\Rn(\tilde X)\|_1 \leq \sqrt{G(X)},$$
where
\begin{equation}\label{eq:def-G}
G(X) := \left(t_X^2-\pi_A\right)\left(t_X^2-\pi_B\right).
\end{equation}
\end{theorem}

The enhanced realignment criterion is actually \emph{strictly stronger} than the realignment criterion. If a normalized state satisfies the centered bound, then it satisfies the
plain realignment bound $\|\Rn(\rho)\|_1\leq1$.  Indeed,
$$
\|\Rn(\rho)\|_1
\leq
\|\Rn(\rho-\rho_A\otimes\rho_B)\|_1
+\|\Rn(\rho_A\otimes\rho_B)\|_1
$$
and
$$
\|\Rn(\rho_A\otimes\rho_B)\|_1
=
\|\rho_A\|_2\|\rho_B\|_2
=
\sqrt{\pi_A\pi_B}.
$$
Therefore the centered bound implies (here, $X = \rho$ and thus $t_X  = 1$)
$$
\|\Rn(\rho)\|_1
\leq
\sqrt{(1-\pi_A)(1-\pi_B)}
+\sqrt{\pi_A\pi_B}
\leq1,
$$
where the last step is the Cauchy--Schwarz ineqquality applied to the two unit vectors
$(\sqrt{\pi_A},\sqrt{1-\pi_A})$ and
$(\sqrt{\pi_B},\sqrt{1-\pi_B})$.  Thus, every state detected by the plain
realignment criterion is detected by the centered criterion.

\begin{example}\label{ex:enhanced-strict}
The enhanced realignment criterion is strictly stronger than the standard
realignment criterion. To illustrate this, consider the family of two-qubit
states
\begin{equation}\label{eq:ex22-mixture}
\rho_c=
\frac12\proj{00}
+\frac{1+2c}{4}\proj{\psi^+}
+\frac{1-2c}{4}\proj{\psi^-},
\qquad
\ket{\psi^\pm}
=\frac1{\sqrt2}\bigl(\ket{01}\pm\ket{10}\bigr), \nonumber
\end{equation}
where $0\le c\le \frac12$. A direct calculation yields
\begin{equation}\label{eq:ex22-values}
\bigl\|\Rn(\rho_c)\bigr\|_1
=c+\frac1{\sqrt2},
\qquad
\bigl\|\Rn(\widetilde{\rho_c})\bigr\|_1
=c+\frac18,
\qquad
\sqrt{G(\rho_c)}
=\frac38, \nonumber
\end{equation}
where, the centered operator $\widetilde{\rho_c}$ and the $G$ quantity are defined in \eqref{eq:def-tildeX}, resp.~\eqref{eq:def-G}:
\[
 \widetilde{\rho_c}
=t_{\rho_c}\rho_c-(\rho_c)_A\otimes(\rho_c)_B,
\qquad
G(\rho_c)
=
\bigl(t_{\rho_c}^2-\pi_A(\rho_c)\bigr)
\bigl(t_{\rho_c}^2-\pi_B(\rho_c)\bigr).
\]
Therefore, we obtain
\[
\|\Rn(\rho_c)\|_1>t_{\rho_c}
\quad\Longleftrightarrow\quad
c>1-\frac1{\sqrt2},
\]
whereas
\[
\|\Rn(\widetilde{\rho_c})\|_1>\sqrt{G(\rho_c)}
\quad\Longleftrightarrow\quad
c>\frac14.
\]
Hence, for every
\[
\frac14<c\le 1-\frac1{\sqrt2}
\]
the state $\rho_c$ is detected by the enhanced realignment criterion but not
by the standard realignment criterion.
\end{example}

As we mentioned before, the trace norm of realignment $\| \Rn(X) \|_1 \leq \operatorname{Tr}(X)$ is not a local trace invariant of the matrix $X$. We will instead consider the higher moments of the realignment operator $\Rn(X)$. More precisely, in the next proposition, we show that the \emph{even} singular value moments of realignment are tensor trace invariants with the permutations $\alpha_m$ and $\beta_m$ which are defined as
\begin{align*}
\alpha_m&:=(1\,\,2m)(2\,\,3)(4\,\,5)\cdots(2m-2\,\,2m-1)\\
\beta_m&:=(1\,\,2)(3\,\,4)\cdots(2m-1\,\,2m).
\end{align*}
Both are fixed-point-free involutions in $S_{2m}$. Let us mention here that similar ideas were put forward in \cite[Fig.~2, panel~(b)]{tarabunga2026quantifying}), where $X$ was assumed to have a matrix product operator (MPO) structure, and also in \cite[Section 3.1.6]{carrozza2026tensor}, in even more generality.

\begin{proposition}\label{prop:even-realignment-moments-as-trace-invariants}
For a Hermitian bipartite operator $X$, let $s_1 \geq s_2 \geq s_3 \ldots \geq s_{d^2} \geq 0$ denote the singular values of $\Rn(X)$, with $d = \min(d_A, d_B)$. Then, we have the following equality:
$$r_m(X)
:=
\Tr\left[(\Rn(X)\Rn(X)^*)^m\right]
= \sum^{d^2}_{i = 1} s^{2m}_i =
\Tr_{\alpha_m,\beta_m}(X).
$$
In other words, $r_m(X)$ is the linear evaluation of
$(\alpha_m, \beta_m) \in \mathbb R[S_{2m}\times S_{2m}]$ on $X^{\otimes 2m}$.
\end{proposition}

\begin{proof}
We prove the result using tensor diagrams. Firstly, note that $\Rn{(X)} \Rn{(X)^*}$ can be represented by the following tensor diagram:
\begin{center}
    \includegraphics{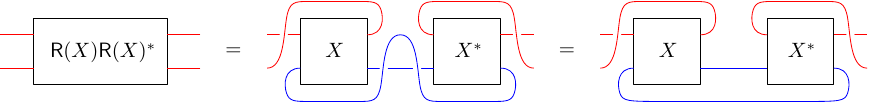}
\end{center}

The quantity $r_m=\Tr\big((\Rn{(X)} \Rn{(X)^*)}^m\big)$ can be represented by repeating the above diagram $m$ times and contracting the wires (we also use $X=X^*$):
\begin{center}
    \includegraphics{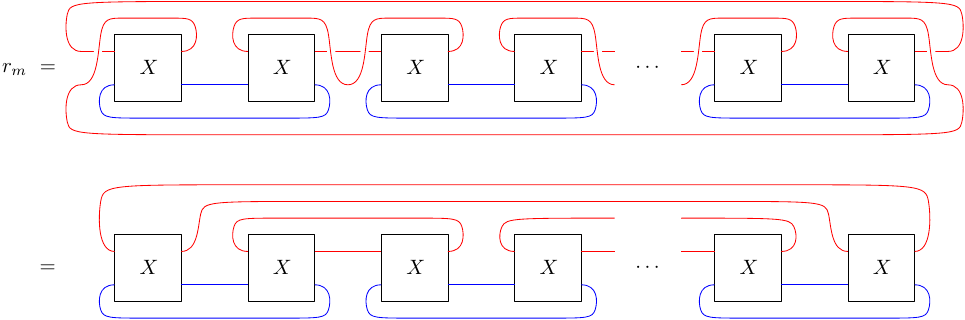}
\end{center}
The final diagram shows that the top wires are connected by the permutation $\alpha_m$ while the bottom wires are connected with the permutation $\beta_m.$
\end{proof}

Note that the first even moment of the realignment of a bipartite quantum state $\rho$, $r_1(\rho) = \Tr[(\Rn(\rho) \Rn(\rho)^*)] = \operatorname{Tr}(\rho^2) = p_2(\rho)$ is exactly the purity of $\rho$. For all Hermitian operators $X$, it holds that $r_m(X) \geq 0$, as they are traces of positive semidefinite matrices.

Analogous to realigment, we now show that the
the even moments of the \emph{centered} realignment operator are elements of $\mathbb R[S_n\times S_n].$ We first define the even moments,
$$
r_m(\tilde X)
:=
\Tr\!\left[\bigl(R(\tilde X)R(\tilde X)^*\bigr)^m\right],
\qquad m\geq1 .
$$
Each $r_m(\tilde X)$ is homogeneous of degree $4m$ in $X$ since $\tilde X := t_X X - X_A \otimes X_B$ is already a quadratic function of $X$.  For example, the first even moment is
\begin{equation}\label{eq:r1-tildeX}
r_1(\tilde  X)
=
\Tr(\tilde X^{\,2})
=
t^2_{X} p_2(X)-2t_X \pi_{\mathrm{mix}}(X)+\pi_A(X)\pi_B(X).
\end{equation}
Note that the centered realignment moment family already contains the quantity $$\pi_{\mathrm{mix}}(X)=\Tr(X(X_A\otimes X_B))$$ in the first moment. We can show that these moments are also expressible in terms of elements in the group algebra $\mathbb{R}[S_n \times S_n]$.

\begin{proposition}
\label{prop:centered-realignment-moments}
For every $m\geq1$, the invariant $r_m(\tilde X)$ is a linear combination of
tensor trace invariants $\Tr_{\sigma,\tau}$ with $\sigma,\tau\in S_{4m}$. In particular,
$$
r_1(\tilde X) = \Tr_{(1)(2)(34), (1)(2)(34)}(X) - 2 \Tr_{(1)(2)(34), (1)(23)(4)}(X) +  \Tr_{(12)(3)(4), (1)(2)(34)}(X).
$$
\end{proposition}

\begin{proof}
The general statement follows from expanding the result in \cref{prop:even-realignment-moments-as-trace-invariants} for the centered matrix $\tilde X$ from \cref{eq:def-tildeX}. The case $m=1$ is just a rewriting of \cref{eq:r1-tildeX}.
\end{proof}

Next, we want to express both the standard and the centered realignment criteria in terms of tensor invariants of the quantum state. To do this, we consider the mathematical theory for moment based optimization for the $1$ norm.

\subsection{Moment based estimation of the 1 norm}\label{sec:moment-interpolation}

Consider the following \emph{discrete} moment problem for a given $r$-tuple of non-negative real numbers $r \in \R_+^m$:

\begin{equation}
\label{eq:alpha-minimization}
\alpha(r) = \min_{n \in \mathbb{N}} \min_{s \in [0,1]^n} \left\{ \sum^{n}_{i=1} s_i \;\middle|\; \forall k \in [m], \; \sum^n_{i=1} s^{2k}_i= r_k\right\}.
\end{equation}

The quantity $\alpha(r)$ is the minimization of the $1$ norm of $s$ subject to specified first $m$ even moments given by $r$, independent of the size of the vector $s$. Note that if $\alpha(r) > 1$, it implies that there is no set (of arbitrary dimension) of singular values with the even moments $r_k$ such that $\|s\| \leq 1$.  In this paper, we only look at the the case $r_1 \leq 1$, which also implies $r_k \leq 1$ for all $k \in [m]$.

Consider the case $m=1$. Then, the problem \cref{eq:alpha-minimization} has an exact solution $\alpha(r_1) =   \sqrt{r_1}.$ This follows from the fact that ${\sum^{n}_{i=1} s_i} \geq \sqrt{\sum^{n}_{i=1} s^2_i}$ and the minimum can be obtained with choosing $n=1$ (a single atom) and $s_1 = \sqrt{r_1} \leq 1$. For $m=2$, the exact solution is given in the next theorem.

\begin{theorem}\label{thm:r1-r2-exact}
The problem in \cref{eq:alpha-minimization} is feasible (i.e.~$\alpha(r)< \infty$) iff $0<r_2\le r_1^{2} \leq 1$ or $r_1 = r_2 = 0$. For the first case, the minimum is attained for a vector with support of size
\[
n=\left\lceil \frac{r_1^{2}}{r_2}\right\rceil.
\]
Setting $\Delta:=\sqrt{n r_2-r_1^2}$ and $m:=n-1$, we have
\begin{equation}
\label{eq:solution-alphar1r2}
\alpha(r_1,r_2)=\frac{\sqrt{m^{2}r_1+m^{3/2}\Delta}+\sqrt{r_1-\sqrt{m}\Delta}}{\sqrt{n}}.
\end{equation}
If $r_1^2/r_2\in \N$ then $\Delta=0$ and
$\alpha(r) = r_1^{3/2}/\sqrt{r_2}.$ For $r_1 = r_2 = 0$, we have $\alpha(r) = 0.$
\end{theorem}

We present the proof of the theorem in \cref{sec:proof-1-norm} which uses results of \cite{sakai2017sharp} about Rényi entropies; see also \cite{harremoes2009joint}. Using \cref{thm:r1-r2-exact}, we can now obtain the characterization of the curve $\alpha(r_1,r_2)=1.$ Let $$(n-1) r_2<r^2_1 \leq n r_2.$$
In this region, $\alpha(r_1,r_2) = \frac{\sqrt{m^{2}r_1+m^{3/2}\Delta}+\sqrt{r_1-\sqrt{m}\Delta}}{\sqrt{n}} = 1 \implies \sqrt{m^{2}r_1+m^{3/2}\Delta}+\sqrt{r_1-\sqrt{m} \Delta} = \sqrt{n}$. Therefore, we can obtain $\alpha(r_1,r_2)$ as a collection of arcs. When, $r^2_1 = nr_2$, $\alpha(r_1,r_2) = \frac{r^{3/2}_1}{\sqrt{r_2}} = 1 \iff r_1^3 = r_2 = \frac{r_1^2}{n} \implies r_1 = \frac{1}{n}, r_2 = \frac{1}{n^3}$. For $n=2$, we can obtain the following curve for the interval $r_1=[\frac{1}{2}, 1].$

\begin{corollary} \label{cor:alpha-leq-1-r1-1/2-1}
    For $r_1 \in [\frac{1}{2},1]$  and $r_2 \leq r^2_1$, $$\alpha(r_1,r_2)\leq1 \iff 2r_2-r^2_1 \geq 2r_1-1.$$
\end{corollary}
The inequality above is valid for all values of $r_1$, but forms the boundary of $\alpha(r_1,r_2) \leq 1$ for $r_1 \in [1/2,1].$

\begin{proposition}\label{prop:pairsum}
Every separable $X$ with $t_X=\Tr X$ satisfies
\begin{equation}\label{eq:pairsum}
  2\Bigl(r_1(X)^2-r_2(X)\Bigr)
  \leq \bigl(t_X^2-r_1(X)\bigr)^2.
\end{equation}
\end{proposition}

\begin{proof}
Let $s_1,s_2,\ldots$ be the singular values of $\Rn(X)$.  Then
\[
 r_1(X)=\sum_i s_i^2,
 \qquad
 r_2(X)=\sum_i s_i^4.
\]
Put $e_2:=\sum_{i<j}s_i s_j$.  The CCNR criterion gives
$\sum_i s_i\leq t_X$, and therefore
\[
 e_2=\frac12\left(\left(\sum_i s_i\right)^2-r_1(X)\right)
 \leq \frac12\bigl(t_X^2-r_1(X)\bigr).
\]
Here $t_X^2-r_1(X)\geq0$, since
$r_1(X)=\sum_i s_i^2\leq(\sum_i s_i)^2\leq t_X^2$.
Since all $s_i s_j$ are nonnegative,
\[
 e_2^2
 \geq \sum_{i<j}s_i^2s_j^2
 =\frac12\bigl(r_1(X)^2-r_2(X)\bigr).
\]
Combining the two inequalities proves \eqref{eq:pairsum}.
\end{proof}

\begin{figure}[htb]
    \centering
    \includegraphics[width=\linewidth]{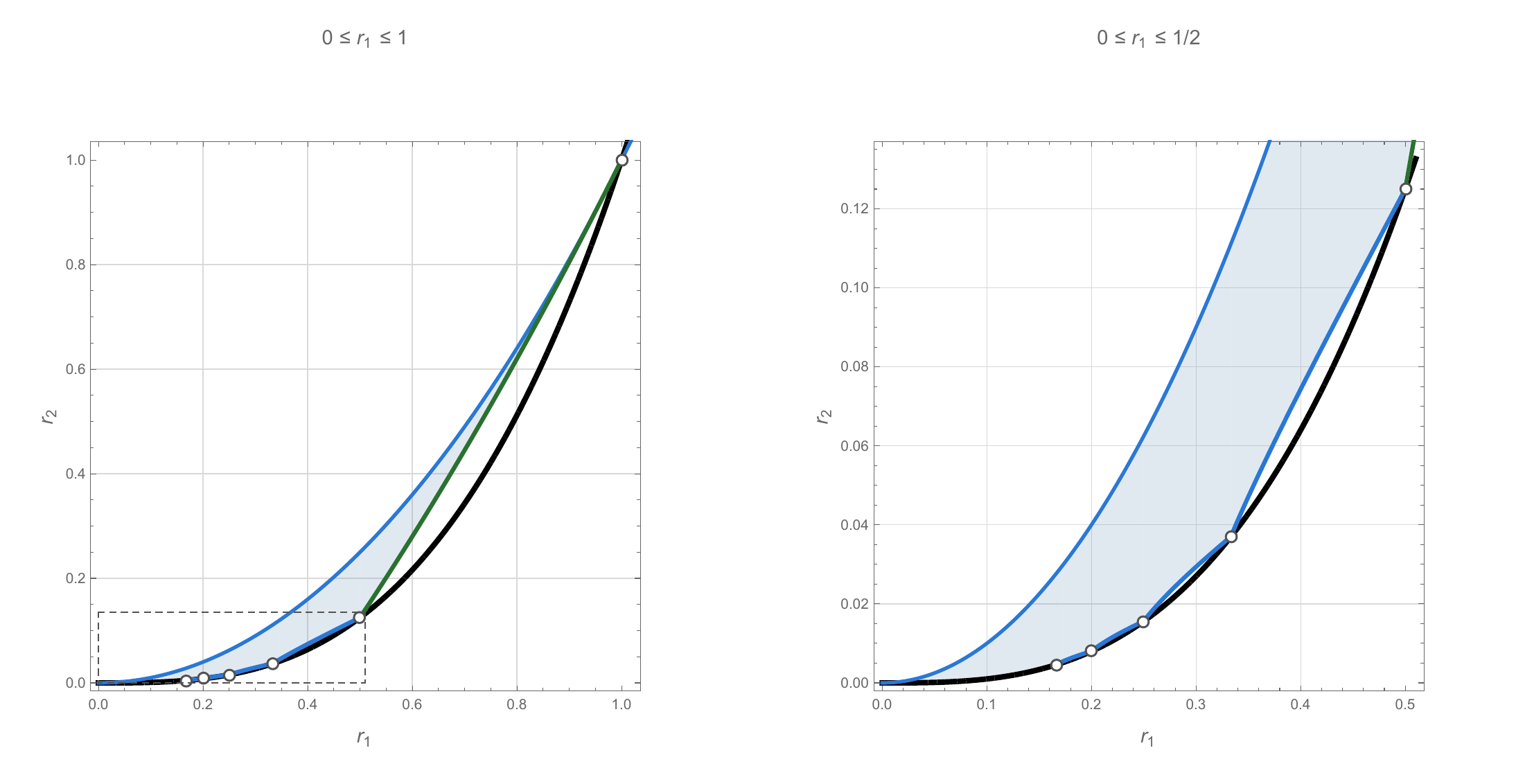}
    \caption{The figure shows the exact solution of $\alpha(r_1,r_2) \leq 1$ (in light blue). The boundary of this region (in blue) is given by the collection of arcs (ref. \cref{thm:r1-r2-exact}), and the upper curve $r^2_1=r_2$ that encodes feasibility of $\alpha(r).$ The first arc (in green) is given by \cref{cor:alpha-leq-1-r1-1/2-1}. The black curve is $r^3_1=r_2$, which is obtained from the continuous approximation $\alpha_c(r)=1$ in \cref{thm:r1-r2-alpha-c}. The right figure is the zoomed rectangle in the left ones.}
\label{fig:alpha-r1-r2}
\end{figure}

The previous results for $m=2$ show that the exact computation of $\alpha(r)$ is likely to be non-trivial for general $m$ and $r \in \R_+^m$. In the second part of this section, we introduce a lower bound on $\alpha(r)$ by relaxing the problem from the \emph{discrete} to the \emph{continuous} case, and show that the relaxation can be computed as a \emph{semidefinite program} (SDP).

\begin{remark}
    Note that if $r_1 > 1$, it implies that $\alpha(r) > 1$ as ${\sum^{n}_{i=1} s_i} \geq \sqrt{\sum^{n}_{i=1} s^2_i} = \sqrt{r_1}.$ 
\end{remark}

\medskip

\noindent\textbf{Relaxations of $\alpha(r)$}. Let $\mathbb{B}[0,1]$ denote (un-normalized) Borel measures on the semialgebraic set $[0,1] = \{t \in \R \, \mid \, t(1-t) \geq 0 \}$. Consider the following relaxation of the optimization problem for $\alpha(r)$ given in \cref{eq:alpha-minimization}.

\begin{equation*}
\alpha_c(r) := \min_{\mu \in \mathbb{B}[0,1]} \left\{ \int x d\mu \;\middle|\; \forall k \in [m], \; \int x^{2k} d \mu= r_k\right\}.
\end{equation*}

It is easy to infer that $\alpha(r) \geq \alpha_c(r)$, since any feasible singular value distribution in \cref{eq:alpha-minimization} corresponds to a Borel measure $\mu = \sum^n_{i = 1}\delta_{s_i}$ which supported in $[0,1]$ and satisfies all the moment conditions. The lower bound $\alpha_c(r)$ can be \emph{strictly} less than $\alpha(r)$ for general $r$, as it is shown in \cref{ex:alpha-c-strict}.

We now show that the continuous relaxation can be recast as a semidefinite program. A truncated moment sequence $y=(y_j)_{j=0}^{2m} \in \mathbb{R}_+^{2m+1}$ is said to have a \emph{representing measure} over $[0,1]$ if there is a measure $\mu$ such that
$$y_j = \int^1_0 x^j d\mu \qquad \text{for all } j=0,\ldots,2m.$$
By using results about the moment problem in one variable (or the truncated Hausdorff moment problem), this is true if and only if the corresponding localizing matrices satisfy the following conditions (see \cite[Theorem 3.2(c)]{lasserre2009moments} or \cite[Theorem 4.10]{laurent2008sums}):
$$M(gy) \succeq 0 \qquad \text{ for } \qquad g(x)= x(1-x), 1$$
where $M(y)_{\alpha,\beta} = y_{\alpha+\beta}$ is a Hankel matrix and
$(gy)_\gamma = \sum_{\delta \geq 0} g_\delta y_{\delta+\gamma}$, with all indices denoting subscripts of the moment sequence. The sizes of the localizing matrices can be truncated to $m \times m$ and $(m+1) \times (m+1)$ as the constraints only apply up to $y_{2m}$.  Hence, the SDP formulation is given by
$$\alpha_c(r) = \min\{y_1 \mid \left(y_{i+j}\right)_{i,j=0}^{m}\succeq 0, \left(y_{i+j+1}-y_{i+j+2}\right)_{i,j=0}^{m-1}\succeq 0\text{ and }\forall k \in [m], \, y_{2k} = r_k\}.$$

We now proceed to show that interpolation inequalities between $L^p$ norms provides a lower bound for $\alpha_c(r)$ and even the exact solutions for $m\leq 2$. The standard interpolation (Hölder's) inequality for $L^p$ norms gives
$$
\left(\int x^{2a} d \mu\right)^{1/(2a)}
\leq
\left(\int x d \mu\right)^\theta
\left(\int x^{2b} d\mu\right)^{(1-\theta)/(2b)},
$$
where $\theta:=(b-a)/(a(2b-1))$ is such that
$$\frac{1}{2a} = \frac{\theta}{1} + \frac{1-\theta}{2b}.$$
For any $1< 2a < 2b$, such a $\theta \in (0,1)$ exists and provides a lower bound for $\int x d \mu$. In particular, applying the interpolation inequality with $a=1$ and $b=2$ (which is the smallest integer setting), we obtain

$$\left(\int x^{2} d \mu\right)^{1/2}
\leq
\left(\int x d \mu\right)^{1/3}
\left(\int x^{4} d\mu\right)^{1/6} \qquad \implies \qquad\alpha_c(r) \geq \frac{r^{3/2}_1}{\sqrt{r_2}}.$$
Note that, unsing different considerations, the same inequality has been obtained in \cite[Section 3.2]{tarabunga2026quantifying}. The lower bound above is attained by the unnormalized measure
$$
\mu_* = \frac{r_1^2}{r_2}\,\delta_{\sqrt{r_2/r_1}}.
$$
Indeed, we have
$$\int x\,d\mu_*=\frac{r_1^{3/2}}{r_2^{1/2}}, \qquad\qquad \int x^2\,d\mu_*=r_1, \qquad\qquad \int x^4\,d\mu_*=r_2.$$
We have thus proven the following result.
\begin{theorem}\label{thm:r1-r2-alpha-c}
In the case $m=2$, for every $0<r_2 \leq r_1$, we have
$$
\alpha_c(r_1,r_2)=\frac{r_1^{3/2}}{\sqrt{r_2}}.
$$
\end{theorem}

We illustrate this result with an example showing that $\alpha_c < \alpha$.

\begin{example}\label{ex:alpha-c-strict}
Consider the two-moment vector $r=(r_1,r_2)=(1, 2/3)$.
We claim that the continuous relaxation is strictly smaller than the discrete optimum for this vector. The formula above gives
$$
\alpha_c\left(1,\frac{2}{3}\right)=\sqrt{\frac{3}{2}}.
$$
This value is attained by the positive, unnormalized measure $\mu_* = (3/2)\delta_{\sqrt{2/3}}$; note that this measure has a \emph{non-integer weight} for its Dirac mass, so it is not feasible for the optimization problem in \cref{eq:alpha-minimization}. Indeed, in the notation of
\cref{thm:r1-r2-exact},
$$
n=\left\lceil\frac{r_1^2}{r_2}\right\rceil=2,
\qquad
\Delta=\sqrt{2r_2-r_1^2}=\frac{1}{\sqrt{3}}.
$$
Consequently,
$$
\alpha\left(1,\frac{2}{3}\right)
=\frac{1}{\sqrt{2}}
\left(
\sqrt{1+\frac{1}{\sqrt{3}}}
+\sqrt{1-\frac{1}{\sqrt{3}}}
\right)=\sqrt{1+\sqrt{\frac{2}{3}}},
$$
proving the claim. Actually, one can easily see from \cref{thm:r1-r2-exact,thm:r1-r2-alpha-c} that
$$\alpha(r_1,r_2) = \alpha_c(r_1,r_2) \quad \iff \quad \frac{r_1^2}{r_2} \in \N.$$
\end{example}

It is natural to ask whether other such interpolation inequalities provide independent bounds for $\alpha_c(r)$. We show next that the interpolation inequality at $a=1, b=2$ provides the tightest lower bound for $\alpha_c(r)$.
\begin{proposition}
The supremum of the following optimization problem ($s \neq 0$):
    $$\gamma(s) =\max_{a < b \in \N} \frac{\|s\|_{2a}^{2a(2b-1)/(b-a)}}{\|s\|_{2b}^{2b(2a-1)/(b-a)}}$$
    is attained at $a=1$ and $b=2$.
\end{proposition}
\begin{proof}
    The proof follows from the log-convexity of the $p$-norms, i.e.~
    $P(t) := \operatorname{log}\|s\|^{t}_{t} = \operatorname{log}\sum_i |s_i|^t$
    is a convex function. First, note that by the monotonicity of the logarithm, it follows that
    $$\operatorname{log}\gamma(s) = \max_{a < b \in \N} \frac{(2b-1) P(2a) - (2a-1)  P(2b)}{b-a}  $$
    also attains its maximum at the same arguments. Now, note that the secant
    $\frac{P(2b) - P(2a)}{b-a}$ is a nondecreasing function in both $a$ and $b$ due to the convexity of $P$.

    $$\frac{(2b-1) P(2a) - (2a-1)  P(2b)}{b-a} =  2P(2a) - (2a - 1) \frac{P(2b) - P(2a)}{b-a}$$ which shows that this function is nonincreasing in $b$. Similarly, we have,

    $$\frac{(2b-1) P(2a) - (2a-1)  P(2b)}{b-a} =  2 P(2b) - (2b - 1) \frac{P(2b) - P(2a)}{b-a}$$ which shows it is nonincreasing in $a$. This implies that the maximum is attained at the minimum possible values of $a$ and $b$, proving the result.
\end{proof}

\subsection{Universal minorants}
We describe here a systematic approach to provide lower bounds for the quantity $\alpha_c(r)$ by considering the \emph{dual optimization problem}. This approach also enables us to construct various elements in $\ineqsep$ of various degrees in the later section. By the generalized moment-problem duality
\cite[Section~1.2, Theorems~1.2--1.3 and Eq.~(1.6)]
{lasserre2009moments}, and after the change of variables $t=x^2$,
the dual problem is
\begin{equation}
\beta_c(r) := \operatorname{sup}\{\sum^m_{k=1}p_k r_k \mid \forall x \in [0,1], \, \sum^m_{k=1} p_k x^k \leq \sqrt{x}\}
\end{equation}

Here, instead of polynomials that provide optimum of $\beta_c(r)$ for some $r$, we would like to look at points that approximate $\beta_c(r)$ for all $r \in \mathbb{R}^m$; we call such points \emph{universal} for $\beta_c(r)$.

\begin{enumerate}
    \item For any $m\geq1$, we can choose the \emph{universal binomial minorant}
\begin{equation}\label{eq:binomial-realignment-minorant}
f_m(x):=
x\sum_{a=0}^{m-1}\frac{(2a)!}{4^a (a!)^2}(1-x)^a,
\qquad 0\leq x\leq1,
\end{equation} since it satisfies $f_m(x) \leq \sqrt{x}$ for $x \in [0,1].$ It is easy to see that $f_m(x)$ is a minorant by considering the Taylor expansion around $x=1^-$:
$$\sqrt{x} =\frac{x}{\sqrt{1-(1-x)}} = x \sum_{a=0}^{\infty} \frac{(2a)!}{4^a (a!)^2} (1-x)^a.$$ Since each term of the expansion is positive for all $0 \leq x \leq 1$, we can conclude that $f_m(x) \leq \sqrt{x}.$
\item We also introduce a minorant based on $\beta$-weighing the cost function; we call these the \emph{universal $L^1_\beta$-weighted minorants}:
\begin{equation}
\label{eq:l1beta}
g_{m,\beta}\in
\operatorname*{arg\,min}_{h\in\mathcal H_m}
\int_0^1 x^{-\beta}\bigl(\sqrt{x}-h(x)\bigr)\,\dd x,
\end{equation}
where we define the set of polynomials
\begin{equation}
\label{eq:Hm}
\mathcal H_m:=\left\{
h(x)=\sum_{k=1}^m a_kx^k:
h(1)=1,\ h(x)\leq\sqrt{x}\ \forall x\in[0,1]
\right\}.
\end{equation}

We present the SDP formulation of this in the appendix, and solutions to this optimization problem for all $\beta \in [0,3/2)$. These universal minorants better approximate the function $\sqrt{x}$ near $0$ as $\beta$ increases (whereas the previous $f_m$ minorant provides a good approximation around $1$).

\end{enumerate}

The binomial minorant families $f_m$ converge pointwise to the function $\sqrt{x}$ in the interval $[0,1]$. The $L^1_\beta$ minorants are obtained by minimizing an integrated approximation error and, unlike the binomial family, need not form a nested hierarchy. Nevertheless, their flexibility can yield substantially improved detection power at low orders.
We plot these polynomial minorants in the  \cref{fig:minorants}, which allows to judge the quality of the approximation of $\sqrt{x}$ by these minorants.

\begin{figure}[htb]
    \centering
    \includegraphics[width=0.5\linewidth]{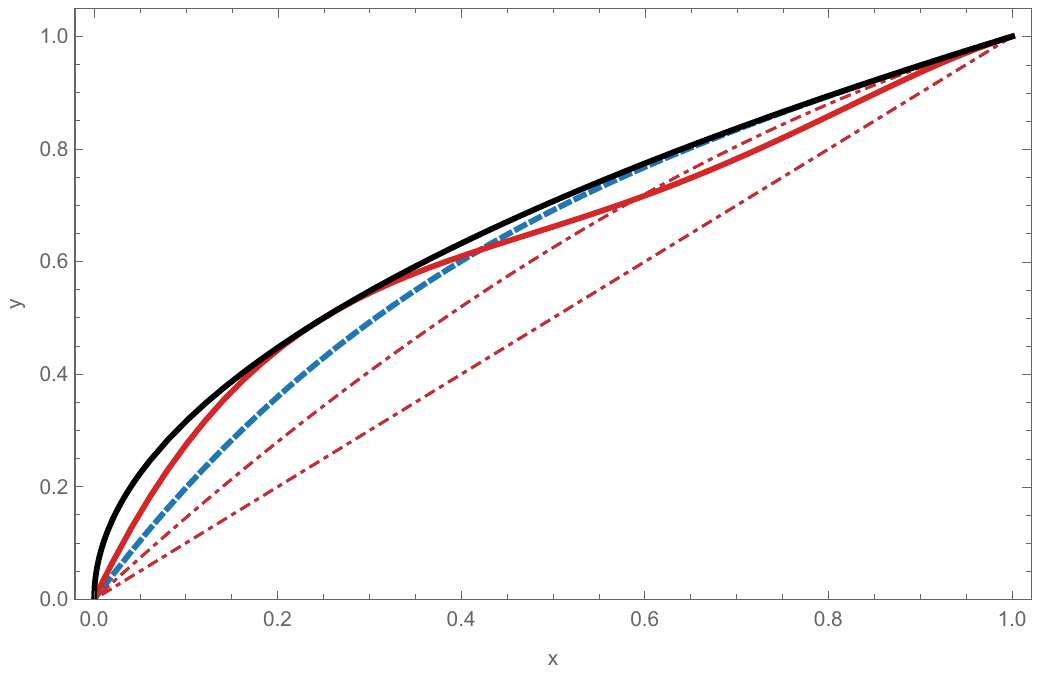}
    \caption{The figure shows the polynomial minorants of $y=\sqrt{x}$. The binomial minorant $f_4$ is presented in dashed blue, as well as the the minorant $g_{4,1}$ obtained from \cref{eq:l1beta} in red. The first two minorants $f_1, f_2$, are represented by $y=x$ and $y =\frac{3}{2}x - \frac{1}{2}x^2$ in dot-dashed lines. For every $\beta$, the minorants $g_{1,\beta}$ and $g_{2,\beta}$ coincide with these two curves. Qualitatively, the binomial minorant approximates better the curve at $1$, while the SDP based minorant is better near $0.$}
    \label{fig:minorants}
\end{figure}

\subsection{Inequalities for separable states}
Recall that we have the convex cone of dimension independent inequalities on separable states $$
\ineqsep_n
=
\left\{
w\in\mathbb R[S_n\times S_n]:
\begin{array}{l}
\Tr_w(\rho)\geq0\ \text{for every }d_A,d_B\geq1\\
\text{and every }\rho\in\Sep(\mathbb C^{d_A}:\mathbb C^{d_B})
\end{array}
\right\}.
$$

We will now show that the results of the previous section can be used to derive elements in the convex cone $\ineqsep_n.$ Firstly, combining the realignment criterion and enhanced realignment criterion with the results of the previous section, we have the following simple corollary.
\begin{proposition}
    \label{thm:realignment-moments-sep}
    Let $r(X) \in \mathbb{R}^m$ be the vector of the first $m$ even moments of the realignment of a separable state $X$. Then, the quantity $\alpha$ from \cref{eq:alpha-minimization} is bounded by
    $$\alpha\left(r\left(\frac{X}{t_X}\right)\right) \leq\left\| \frac{\Rn(X)}{t_X}\right\|_1 \leq 1.$$

    Similarly, let $\tilde X := t_XX - X_A \otimes X_B$ be the centered operator associated with the separable state $X$. Assume that $G(X)= (t^2_X-\pi_A(X))(t^2_X-\pi_B(X)) \neq 0$. Then,
    $$\alpha\left(r\left(\frac{\tilde X}{\sqrt{G(X)}}\right)\right) \leq \left\| \frac{\Rn{(\tilde X)}}{\sqrt{G(X)}}\right\|_1 \leq 1.$$
\end{proposition}

\begin{proof}
We use the following elementary observation.  Let $Y$ be a bipartite
operator, and let $s_1,\ldots,s_N$ be the singular values of $\Rn(Y)$,
including zero singular values if necessary.  By the definition of the
realignment moments, we have, for all $k \geq 1$,
\begin{equation}\label{eq:realignment-moments-singular-values}
r_k(Y)=\Tr\!\left[\bigl(\Rn(Y)\Rn(Y)^*\bigr)^k\right]
=\sum_{i=1}^N s_i^{2k}.
\end{equation}
Moreover, for every $c\geq0$, linearity of the realignment map gives
$\Rn(cY)=c\Rn(Y)$.  Consequently, the singular values of $\Rn(cY)$ are
$cs_1,\ldots,cs_N$, and
\begin{equation}\label{eq:realignment-moment-homogeneity}
r_k(cY)=c^{2k}r_k(Y).
\end{equation}

We first consider the uncentered statement.  Since $X$ is a state,
$t_X=\Tr(X)>0$.  Let $s_1,\ldots,s_N$ be the singular values of
$\Rn(X)$ and set $u_i:=s_i/t_X$ for $i=1,\ldots,N$.
The realignment criterion for separable states implies
$$
\sum_{i=1}^N u_i
=\frac{\|\Rn(X)\|_1}{t_X}
\leq1.
$$
In particular, every $u_i$ belongs to $[0,1]$.  Furthermore,
\eqref{eq:realignment-moments-singular-values} and
\eqref{eq:realignment-moment-homogeneity} show that, for every
$k\in[m]$,
$$
\sum_{i=1}^N u_i^{2k}
=\frac{r_k(X)}{t_X^{2k}}
=r_k\!\left(\frac{X}{t_X}\right).
$$
Thus, $u=(u_1,\ldots,u_N)$ is a feasible vector in the minimization
problem defining $\alpha\bigl(r(X/t_X)\bigr)$.  Since $\alpha$ is the
minimum of the sum of the entries over all such feasible vectors, we obtain
$$
\alpha\left(r\left(\frac{X}{t_X}\right)\right)
\leq \sum_{i=1}^N u_i
=\left\|\frac{\Rn(X)}{t_X}\right\|_1
\leq1.
$$

The proof of the centered statement is very similar and left to the reader.
\end{proof}

\begin{remark}\label{rem:G-zero}
    Regarding the second point in the result above, if $G(X) = 0$, then either $\pi_A(X) = t^2_X.$ or $\pi_B(X) = t^2_X.$ Then, considering the first case, $\pi_A(X) = t^2_X \iff \pi_A(\frac{X}{t_X})=1.$ This allows us to write $\frac{X_A}{t_X} = \ketbra{\psi}{\psi}$ with $\|\psi\|_2=1$. Hence, $X$ is a product matrix, $X = Y \otimes Z$ for some $Y,Z \in \mathcal{M}^+$, which consequently implies that $\tilde X = 0.$ Therefore, all moments are $0$, while the state is \emph{detected} as separable (in particular, it is a product state) by just estimating $r_1(\tilde X)$. 
\end{remark}

\begin{remark}
    Note that $r_1\Big(\frac{\tilde X}{\sqrt{G(X)}}\Big)$ can be strictly $>1$ $\iff r_1(\tilde X) > G(X)$, but this immediately implies that $X$ is entangled. 
\end{remark}

Recall that we showed that $\alpha(r(X))$ can be computed exactly for $r(X) \in \mathbb{R}^2$, and also has a neat lower bound $\alpha_c(r(X))$ for higher number of moments. This lower bound can be computed exactly by solving an SDP of $O(m)$ where $m$ are the number of accessible moments. Dually, lower bounds on $\alpha_c(r(X))$ can be obtained by \emph{feasible points} of the optimization problem $\beta_c(r)$, that correspond to polynomial minorants of $\sqrt{x}$ on the interval $[0,1].$
We present some of the inequalities we can derive on separable states using \cref{thm:realignment-moments-sep}.

The first inequality can be derived from the interpolation lower bound \cref{thm:r1-r2-alpha-c}; this is exactly the degree $6$ inequality
\begin{equation}
\label{eq:interpolation-sep-re}
    r^3_1(X) \leq r_2(X)t_X^2.
\end{equation}

The same inequality has been identified in \cite[eq. (23)]{tarabunga2026quantifying} in terms of a new quantity called the $r_4$-negativity. Then, the \cref{cor:alpha-leq-1-r1-1/2-1} implies the following inequality on all separable states with $r_1 = \operatorname{Tr}(X^2) \geq \frac{1}{2}$:

\begin{equation}
\label{eq:chi-inequality}
    2r_2(X)+t_X^4 \geq r_1(X)^2+2t_X^2r_1(X).
\end{equation}

Although this inequality holds for $r_1 \leq \frac{1}{2}$ as well, it is weaker than the inequality in \cref{eq:interpolation-sep-re} in this regime. Finally, for any polynomial minorant $p := \sum^m_{k=1} p_k x^k$ of $\sqrt{x}$ in the interval $[0,1]$, we have the following inequality:

\begin{equation}
\label{eq:minorant-ineq}
    \sum^m_{k=1} p_k r_k t_X^{2(m-k)} \leq t^{2m}_X.
\end{equation}

By choosing either the binomial minorant/$L^1$-beta approximation, we obtain the first two such inequalities:

\begin{equation}
\label{eq:1-and-2-minorants}
r_1 \leq t^2_X \qquad\text{ and }\qquad \frac{3}{2} r_1 t^2_X- \frac{1}{2} r_2 \leq t^4_X.
\end{equation}

Note that the first inequality actually belongs to the smaller set $\ineqpsd_2$ while the second one is the non-trivial inequality belonging to $\ineqsep_4 \backslash \iota_+(\ineqpsd_4)$.

Finaly, all the enhanced realignment based inequalities can be obtained by formally replacing $t^2_X \to G(X), X \to \tilde X$, following the similarity of the result in \cref{thm:realignment-moments-sep}. For example, we have the following inequalities for all separable states which follow from \cref{eq:interpolation-sep-re} and, respectivlely, \cref{eq:chi-inequality}:

\begin{enumerate}
    \item $r^3_1(\tilde X) \leq r_2(\tilde X)G(X)$
    \item $2 r_2(\tilde X) + G(X)^2 \geq r_1(\tilde X)^2 + 2G(X)r_1(\tilde X)$.
\end{enumerate}

Moreover, for any polynomial minorant $p := \sum^m_{k=1} p_k x^k$ of $\sqrt{x}$ in the interval $[0,1]$, we have the following inequality.

\begin{equation}
\label{eq:}
    \sum^m_{k=1} p_k r_k(\tilde X) G(X)^{m-k} \leq G(X)^m
\end{equation}

Recall that the enhanced realignment criterion implies the realignment criterion for all quantum states. It is natural to ask whether this also applies to the moment-based criterion, i.e.:

$$\forall X \in \mathcal{M}^+,  \quad \alpha\left(r\left(\frac{\tilde X}{\sqrt{G(X)}}\right)\right) \leq 1 \implies \alpha\left(r\left(\frac{X}{t_X}\right)\right) \leq 1.$$

In the implication above, we tacitly assume that $X\neq0$ and
$G(X)\neq0$, so that the normalization $\tilde X/\sqrt{G(X)}$
makes sense, see \cref{rem:G-zero}.

\section{Detection efficiency of the criteria} \label{sec:detection}
In this section, we examine the detection efficiency of the following degree-four
entanglement criteria:
\begin{align}
\label{eq:crit-one-moment}
\tilde r_1&\leq G,\\
\label{eq:crit-minorant2}
\tfrac32 r_1-\tfrac12 r_2&\leq1,\\
\label{eq:crit-pairsum}
2\bigl( r_1^{2}- r_2\bigr)&\leq(1- r_1)^2,
\end{align}
Their provenance is as follows.  Criterion \eqref{eq:crit-one-moment} is obtained from
\cref{thm:realignment-moments-sep} for $m=1$ and it
uses a single centered realignment moment.  Criterion \eqref{eq:crit-minorant2} is the
second binomial minorant $f_2(x)=\frac32x-\frac12x^2$ of \cref{eq:binomial-realignment-minorant} and
Criterion \cref{eq:crit-pairsum} is \cref{prop:pairsum}: for
$ r_1\in[\frac12,1]$ it is \emph{equivalent} to $\alpha( r_1,
r_2)\leq1$. 

Finally, we shall use the standard two-moment interpolation criterion 
\begin{equation}\label{eq:crit-interpolation}
r_1^{3}\leq r_2,
\end{equation}
which follows from the interpolation bound
$\alpha_c(r)\geq r_1^{3/2}r_2^{-1/2}$ with $(a,b)=(1,2)$, and is homogeneous
of degree six before restricting to normalized states: for an unnormalized
separable operator $X$, it reads
$r_1(X)^3\leq t_X^2r_2(X)$.

\Cref{sec:pure} shows that each of the four
criteria discussed above detects every entangled pure state. \Cref{sec:comparison} establishes an ordering between the
standard interpolation and 3 moments PPT criteria, and \cref{sec:degree-four} compares the
degree-four criteria against one another, both on explicit states and statistically.
\Cref{sec:tiles} is devoted to PPT entanglement, where we consider the $3 \times 3$ Tiles state, and finally, \cref{sec:ppt-state}
exhibits a PPT entangled state in $\C^3\otimes\C^3$ that is detected at
substantially lower binomial orders, and for which the centered interpolation
criterion already achieves detection using only two centered moments.

\subsection{Pure states}\label{sec:pure}

The next proposition shows that each of the four criteria above detects every entangled pure state.
\begin{proposition}\label{prop:pure-general}
Let $\rho=\proj\psi$ be a pure state.  Each of the following three, degree four criteria is
violated if and only if $\psi$ is entangled:
\begin{enumerate}
\item the one-moment centered criterion \cref{eq:crit-one-moment}
$\tilde r_1\le G$;
\item the degree-two minorant criterion \cref{eq:crit-minorant2}
$\tfrac32 r_1-\tfrac12 r_2\le1$;
\item the pair-sum criterion \cref{eq:crit-pairsum}
$2( r_1^{\,2}- r_2)\le(1- r_1)^2$;
\end{enumerate}
\end{proposition}

\begin{proof}
Write the Schmidt decomposition of the pure state as
$$
\ket\psi=\sum_{i=1}^{r}s_i\ket{ii},
\qquad
s_i>0,
\qquad
\sum_{i=1}^{r}s_i^{\,2}=1,
$$
where $r$ is the Schmidt rank of $\ket\psi$; the state is separable precisely
when $r=1$.  In the Schmidt bases,
$$
\rho=\proj\psi=\sum_{i,j=1}^{r}s_is_j\ket{e_if_i}\bra{e_jf_j},
\qquad\text{so that}\qquad
\Rn(\rho)=\sum_{i,j=1}^{r}s_is_j\ket{e_ie_j}\bra{f_if_j}.
$$
The latter is a sum of rank-one terms with orthonormal left and right vectors,
so the singular values of $\Rn(\rho)$ are exactly the products $s_is_j$ of
Schmidt coefficients.  Consequently the realignment moments are available in
closed form:
$$
r_1=\sum_{i,j=1}^{r}s_i^{\,2}s_j^{\,2}=\Bigl(\sum_{i=1}^{r}s_i^{\,2}\Bigr)^{2}=1,
\qquad
r_2=\sum_{i,j=1}^{r}s_i^{\,4}s_j^{\,4}=\Bigl(\sum_{i=1}^{r}s_i^{\,4}\Bigr)^{2}.
$$
Write $x=(s_i^{\,2})_{i=1}^{r}$ for the vector of Schmidt coefficients, so that
$\|x\|_1=\sum_is_i^{\,2}=1$ and $\|x\|_2^{\,2}=\sum_is_i^{\,4}$.  Recall that a
nonnegative vector satisfies $\|x\|_2\leq\|x\|_1$, with equality if and only if
$x$ has support one.  Hence
\begin{equation}\label{eq:pure-support-one}
\sum_{i=1}^{r}s_i^{\,4}\leq1,
\qquad\text{with equality if and only if }r=1 .
\end{equation}
\emph{(1).}  The centered criterion requires the two marginals.  From the
Schmidt decomposition, $\rho_A=\sum_is_i^{\,2}\proj{e_i}$ and
$\rho_B=\sum_is_i^{\,2}\proj{f_i}$, so that
$$
\Tr(\rho_A^{\,2})=\Tr(\rho_B^{\,2})=\sum_{i=1}^{r}s_i^{\,4},
\qquad
\Tr\bigl(\rho(\rho_A\otimes\rho_B)\bigr)
=\bra\psi\rho_A\otimes\rho_B\ket\psi
=\sum_{i=1}^{r}s_i^{\,6}.
$$
Together with $\Tr(\rho^2)=1$ this gives
$$
G=\bigl(1-\Tr\rho_A^{\,2}\bigr)\bigl(1-\Tr\rho_B^{\,2}\bigr)
=\Bigl(1-\sum_{i}s_i^{\,4}\Bigr)^{2}
$$
and, expanding $\tilde\rho=\rho-\rho_A\otimes\rho_B$,
$$
\tilde r_1=\Tr(\tilde\rho^{\,2})
=\Tr(\rho^2)-2\Tr\bigl(\rho(\rho_A\otimes\rho_B)\bigr)+\Tr(\rho_A^{\,2})\Tr(\rho_B^{\,2})
=1-2\sum_is_i^{\,6}+\Bigl(\sum_is_i^{\,4}\Bigr)^{2}.
$$
Subtracting the above two terms, we get
$$
G-\tilde r_1
=-2\Bigl(\sum_{i=1}^{r}s_i^{\,4}-\sum_{i=1}^{r}s_i^{\,6}\Bigr)
=-2\sum_{i=1}^{r}s_i^{\,4}\bigl(1-s_i^{\,2}\bigr).
$$
Since $0<s_i^{\,2}\leq1$ for every $i$, each summand $s_i^{\,4}(1-s_i^{\,2})$
is nonnegative, and it vanishes if and only if $s_i^{\,2}=1$.  The
normalization $\sum_is_i^{\,2}=1$ forces this to happen for every $i$ exactly
when $r=1$.  Hence $G-\tilde r_1=0$ for product vectors, while
$G-\tilde r_1<0$, i.e.\ the criterion is violated, as soon as $r\geq2$.\\
\emph{(2).}  Substituting the moments computed above into the degree-two
minorant criterion $\tfrac32r_1-\tfrac12r_2\leq1$ gives
$$
\tfrac32-\tfrac12\Bigl(\sum_{i=1}^{r}s_i^{\,4}\Bigr)^{2}\leq1
\qquad\Longleftrightarrow\qquad
\sum_{i=1}^{r}s_i^{\,4}\geq1 .
$$
Comparing with \cref{eq:pure-support-one}, the criterion holds if and only if
$\sum_is_i^{\,4}=1$, that is, if and only if $\|x\|_2=\|x\|_1$, which happens
exactly when $x$, and hence $s$, has support one.  Therefore the criterion is
satisfied precisely for product vectors and violated precisely for entangled
ones.\\
\emph{(3).}  Since $r_1=1$, the right-hand side of the pair-sum criterion
$2(r_1^{\,2}-r_2)\leq(1-r_1)^2$ vanishes, and the criterion reduces to
$$
2\Bigl(1-\bigl(\textstyle\sum_is_i^{\,4}\bigr)^{2}\Bigr)\leq0
\qquad\Longleftrightarrow\qquad
\sum_{i=1}^{r}s_i^{\,4}\geq1 ,
$$
which is the same condition as in \emph{(2)}.  The conclusion follows by the
same argument.

\end{proof}

Note that \emph{(2)} and \emph{(3)} reduce to the same condition
$\sum_is_i^{\,4}\geq1$, so on pure states they are equivalent.

\subsection{Comparison of the standard interpolation criterion with the
3 moments PPT criterion}
\label{sec:comparison}

Recall the 3 moments PPT criterion of \cref{sec:sep}:
$$
q_3(\rho) = p_3\bigl(\rho^{\Gamma_B}\bigr)\geq p_2\bigl(\rho^{\Gamma_B}\bigr)^2 = q_2(\rho)^2.
$$
We compare this criterion
with the standard interpolation criterion \cref{eq:crit-interpolation}.  In
this case there is a strict ordering: every state detected by the
interpolation criterion is already detected by the 3 moments PPT criterion.
We shall use two auxiliary moment inequalities.  The first is a scalar statement for the
eigenvalues of an arbitrary Hermitian matrix which, in particular, does not
require the matrix itself to be positive semidefinite.

\begin{lemma}\label{lem:ppt3-to-ppt4}
Let $X\in\Msa_d$ be Hermitian with $p_1:=\Tr(X)>0$.  Then
$$
p_1(X)p_3(X)\geq p_2(X)^2
\quad\Longrightarrow\quad
p_1(X)^2p_4(X)\geq p_2(X)^3.
$$
\end{lemma}

\begin{proof}
Let $\mu_1,\ldots,\mu_d\in\mathbb R$ be the eigenvalues of $X$.  Since
$p_1>0$, the matrix $X$ is nonzero and hence
\[
p_2=\sum_{i=1}^d\mu_i^2>0.
\]
Consequently,
\[
w_i:=\frac{\mu_i^2}{p_2},\qquad i=1,\ldots,d,
\]
defines a probability vector.  If $\mathbb E_w$ denotes expectation with
respect to these weights, then
\[
\mathbb E_w[\mu]
=\frac{1}{p_2}\sum_i\mu_i^3
=\frac{p_3}{p_2},
\qquad
\mathbb E_w[\mu^2]
=\frac{1}{p_2}\sum_i\mu_i^4
=\frac{p_4}{p_2}.
\]
The hypothesis gives $\mathbb E_w[\mu]\geq p_2/p_1$.  Jensen's inequality for
the convex function $t\mapsto t^2$ therefore yields
\[
\frac{p_4}{p_2}
=\mathbb E_w[\mu^2]
\geq \mathbb E_w[\mu]^2
\geq \frac{p_2^2}{p_1^2}.
\]
Multiplying by $p_1^2p_2$ proves the claim.
\end{proof}

The second
inequality compares the fourth moments of the realignment and the partial
transpose.  Unlike the preceding scalar implication, it uses the positivity
of the original bipartite operator.

\begin{proposition}\label{prop:realignment-pt4}
For every positive semidefinite bipartite operator $\rho$,
\begin{equation}\label{eq:realignment-pt4}
r_2(\rho)
\geq q_4(\rho) =
p_4(\rho^{\Gamma_B}).
\end{equation}
Equivalently,
$$
\|\Rn(\rho)\|_4^4\geq\|\rho^{\Gamma_B}\|_4^4.
$$
In tensor trace notation, \cref{eq:realignment-pt4} says that
$$
\Tr_{(14)(23),(12)(34)}(\rho)
-
\Tr_{(1234),(1432)}(\rho)
\geq0.
$$
Thus the corresponding element of $\mathbb R[S_4\times S_4]$ belongs to
$\ineqsep_4$. In fact, the inequality holds on the full positive
semidefinite cone, not only on separable operators; we present it as an element of $\ineqsep_4$ and not $\ineqpsd_4$ because it requires the bipartite structure of $\rho$.
\end{proposition}

\begin{proof}
Start from the spectral decomposition of $\rho$, with $\rank \rho = s$
$$
\rho=\sum_{a=1}^s\ket{\varphi_a}\bra{\varphi_a},
$$
where the vectors $\varphi_a$ absorb the square roots of the nonzero
eigenvalues of $\rho$ and are therefore generally not normalized.  Let
$M_a\in\M_{d_A \times d_B}$ be the matricization of $\varphi_a$ defined by
$\operatorname{vec}(M_a)=\varphi_a$.  We shall express the quantities $r_2(\rho)$ and $q_4(\rho)$ in the statement terms of the 3-tensor $M \in \mathbb C^s \otimes \mathbb C^{d_A} \otimes \mathbb C^{d_B}$ using the graphical tensor calculus.

First, the tensor $M$ and the positive semidefinite operator $\rho$ are represented by

\begin{center}
\includegraphics[scale=1]{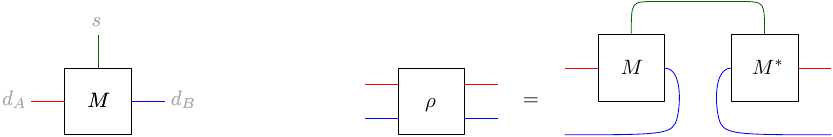}
\end{center}

The diagrams for the realignment $\Rn(\rho)$ and the partial transposition $\rho^{\Gamma_B}$ are

\begin{center}
\includegraphics[scale=1]{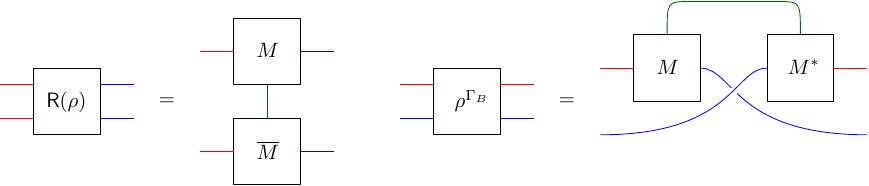}
\end{center}

The quantity $r_2(\rho) = \Tr ([\Rn(\rho) \Rn(\rho)^*]^2)$  is given by

\begin{center}
\includegraphics[scale=1]{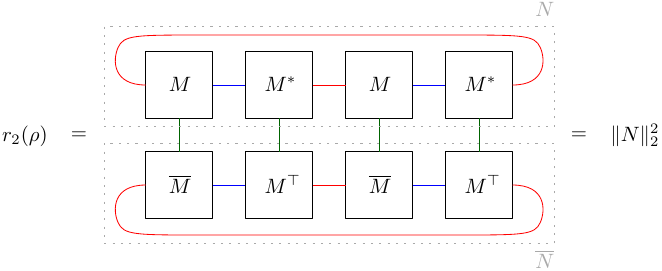}
\end{center}

Define the tensor $N \in (\mathbb C^s)^{\otimes 4}$ by $N_{\alpha\beta\gamma\delta}
:=\Tr\bigl(M_\alpha M_\beta^* M_\gamma M_\delta^*\bigr)$.
After graphically rearranging the $M$ tensors in the diagram for
$q_4(\rho) = \Tr( [\rho^{\Gamma_B}]^4)$ we obtain

\begin{center}
\includegraphics[scale=1]{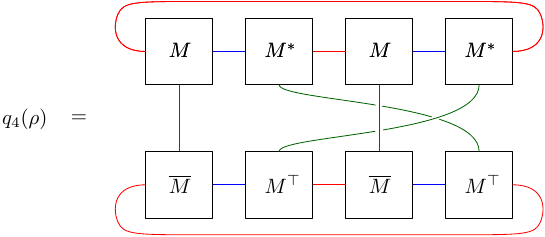}
\end{center}

In terms of the tensor $N$, the diagram above can be written as $q_4(\rho) = \langle N, P_{24} N \rangle$, where $P_{24}$ is the operator permutting the second and fourth indices of $N$.
Since $\rho^{\Gamma_B}$ is Hermitian, $q_4(\rho)=\Tr([\rho^{\Gamma_B}]^4)\geq 0$,
so $q_4(\rho)=\bigl|\langle N,P_{24}N\rangle\bigr|$.
Note that the operator $P_{24}$ only permutes the entries of its input tensor, hence $\|P_{24}N\|_2 = \|N\|_2$.
A simple application of the Cauchy-Schwarz inequality gives
$$q_4(\rho) = \bigl|\langle N, P_{24} N \rangle\bigr|
\leq \|N\|_2\|P_{24}N\|_2 = \|N\|_2^2 = r_2(\rho),$$
proving the claim.
\end{proof}

Combining the two inequalities gives the desired relation between
the criteria.

\begin{theorem}\label{thm:interpolation-implied-ppt3}
Let $\rho$ be a normalized bipartite state.  Then
\begin{equation}\label{eq:ppt3-implies-interpolation}
q_3(\rho)\geq r_1(\rho)^2
\quad\Longrightarrow\quad
r_2(\rho)\geq r_1(\rho)^3.
\end{equation}
Equivalently, every state detected by the standard interpolation criterion
is detected by the 3 moments PPT criterion:
$$
r_2(\rho)<r_1(\rho)^3
\quad\Longrightarrow\quad
q_3(\rho)<r_1(\rho)^2.
$$
\end{theorem}

\begin{proof}
For a quantum state $\rho$, we have $p_1 = q_1 = 1$ and the purity $\Tr(\rho^2) = p_2 = q_2 = r_1$. From \cref{lem:ppt3-to-ppt4} we have
$$
q_3 \geq r_1^2  = q_2^2 \implies q_4 \geq q_2^3 = r_1^3.
$$
Using \cref{prop:realignment-pt4} we also have $q_4 \leq r_2$. Hence, $r_2 \geq r_1^3$, as claimed.
\end{proof}

The inclusion in \cref{thm:interpolation-implied-ppt3} is strict.  Indeed,
consider the two qutrit isotropic state
$$
\rho_p=p\proj{\Phi_3}+(1-p)I_9/9.
$$
The eigenvalues of
$\rho_p^{\Gamma_B}$ are $(1+2p)/9$ with multiplicity 6 and $(1-4p)/9$
with multiplicity 3.  The singular values of $\Rn(\rho_p)$ are $1/3$
with multiplicity 1 and $p/3$ with multiplicity 8.  Therefore
$$
r_1(\rho_p)=\frac{1+8p^2}{9},
\qquad
r_2(\rho_p)=\frac{1+8p^4}{81},
\qquad
q_3(\rho_p)
=\frac{6(1+2p)^3+3(1-4p)^3}{729}.
$$
At $p=3/10$, these formulas give
$$
q_3(\rho_{3/10})-r_1(\rho_{3/10})^2
=-\frac{16}{5625}<0,
$$
whereas
$$
r_2(\rho_{3/10})-r_1(\rho_{3/10})^3
=\frac{140461}{22781250}>0.
$$
Thus the 3 moments PPT criterion detects $\rho_{3/10}$, while the standard
interpolation criterion does not.

\subsection{Comparison of degree-four criteria} \label{sec:degree-four}

In this section we compare the 4 degree-four criteria for separability discussed in this paper. We are able to place them in the same setting, as elements of $\ineqsep_4$.  For an arbitrary positive semidefinite bipartite
operator $\rho$, write
$$
t=\Tr\rho,
\qquad
p_k=\Tr(\rho^k),
\qquad
q_k=\Tr\bigl[(\rho^{\Gamma_B})^k\bigr],
\qquad
r_k=\Tr\bigl[(\Rn(\rho)\Rn(\rho)^*)^k\bigr].
$$
Clearly $p_1=q_1=t$ while the purity is given by $p_2=q_2 = r_1=\Tr(\rho^2)$.  The four homogeneous
separability criteria we consider are:
\begin{align}
W_{\rm P}&:=tq_3-q_2^2,
&
W_{\rm E}&:=G-\tilde r_1,
\label{eq:four-defects-first}\\
W_{\rm M}&:=t^4-\frac32t^2r_1+\frac12r_2,
&
W_{\rm S}&:=(t^2-r_1)^2-2(r_1^2-r_2).
\label{eq:four-defects-second}
\end{align}
Here the subscripts stand respectively for 3 moments PPT criterion, the first enhanced
realignment moment criterion, the second binomial \emph{minorant}, and the
\emph{pair-sum} criterion.  Thus $W_{\rm P},W_{\rm E},W_{\rm M},W_{\rm
S}\in\ineqsep_4$.  A normalized state is detected by a criterion precisely
when the corresponding $W$ trace invariant is strictly negative.

Note that although it is displayed in degree-four form, the enhanced realignment defect $W_E$ factors by $t=\Tr \rho$; this can be seen by expanding the fomula and noticing that the term $\pi_A \pi_B$ which is the constant term in $t$ cancells. Its degree-four presentation is therefore only the corresponding homogeneous lift of a genuine inequality from $\ineqsep_3$.

For clarity, we formulate implications in terms of these detection sets:
$$
\mathcal V_j:=\{\rho \succeq 0\, : \, \Tr\rho=1,\ W_j(\rho)<0\},
\qquad j\in\{\mathrm P,\mathrm E,\mathrm M,\mathrm S\}.
$$

\begin{theorem}
\label{thm:four-detection-sets}
Uniformly over all finite local dimensions, the only inclusions among the four detection
sets are
\begin{equation}
\mathcal V_{\rm M}\subsetneq\mathcal V_{\rm P}
\qquad \text{ and } \qquad
\mathcal V_{\rm M}\subsetneq\mathcal V_{\rm S}.
\label{eq:only-degree-four-inclusions}
\end{equation}
\end{theorem}

\begin{proof}
Let $\rho$ be normalized and set
$$
x:=r_1(\rho)=\Tr(\rho^2), \quad
y:=r_2(\rho), \quad
a:=\Tr(\rho_A^2), \quad
b:=\Tr(\rho_B^2), \quad
c:=\Tr[\rho(\rho_A\otimes\rho_B)].$$
We can now express the 4 invariants in terms of these parameters and $q_3$:
\begin{equation}
(W_{\rm P},W_{\rm E},W_{\rm M},W_{\rm S})
=\left(q_3-x^2,\ 1+2c-x-a-b,\
1-\frac32x+\frac12y,\ 1-2x-x^2+2y\right).
\label{eq:four-normalized-defects}
\end{equation}
Since $\rho$ is a density matrix, its purity satisfies $0<x\leq1$.  Violation of the second
binomial minorant means
\begin{equation}
W_{\rm M}<0
\quad\Longleftrightarrow\quad
y<3x-2.
\label{eq:minorant-violation-threshold}
\end{equation}
For $0\leq x\leq1$ one has
$$
x^3-(3x-2)=(1-x)^2(x+2)\geq0.
$$
It follows from \eqref{eq:minorant-violation-threshold} that $y<x^3$.
Now \cref{thm:interpolation-implied-ppt3} gives $q_3<x^2$
which is $W_{\rm P}<0$.  This proves
$\mathcal V_{\rm M}\subseteq\mathcal V_{\rm P}$.

Likewise,
\begin{equation}
\left(x+\frac12x^2-\frac12\right)-(3x-2)
=\frac12(1-x)(3-x)\geq0.
\label{eq:pairsum-dominates-minorant}
\end{equation}
Combining \eqref{eq:minorant-violation-threshold} and
\eqref{eq:pairsum-dominates-minorant} gives
$$
y<x+\frac12x^2-\frac12,
$$
which is exactly $W_{\rm S}<0$. Hence
$\mathcal V_{\rm M}\subseteq\mathcal V_{\rm S}$.

It remains to prove strictness and to rule out every other inclusion.  We consider the following four matrices:
\begin{align*}
\rho_1&=\frac18\begin{bmatrix}6&0&0&1\\0&0&0&0\\0&0&0&0\\1&0&0&2\end{bmatrix}
\in\M_2\otimes\M_2 \\
\rho_2&=\frac12\proj{\Phi_3}+\frac12\cdot\frac{I_9}{9}\in\M_3\otimes\M_3\\
\rho_3&:=\frac{51}{100}\proj{\Phi_2}
       +\frac{49}{100}\proj{12}\in\M_2\otimes\M_2\\
\rho_4&:=\frac15\proj{00}
       +\frac45\proj{\Phi_4^\perp}\in\M_5\otimes\M_5,
\end{align*}
where
$$
\ket{\Phi_d}:=\frac1{\sqrt d}\sum_{i=1}^d\ket{ii}
       \in\mathbb C^d\otimes\mathbb C^d
$$
is the maximally entangled state. For $\rho_4$, the local basis of $\mathbb C^5$ is
$\ket0,\ldots,\ket4$, and we set
$$
\ket{\Phi_4^\perp}:=\frac12\sum_{i=1}^4\ket{ii}
       \in\mathbb C^5\otimes\mathbb C^5.
$$
Note that the two terms in $\rho_4$ have orthogonal support. The four matrices $\rho_1,\rho_2,\rho_3,\rho_4$ are manifestly density matrices.

Computation of partial traces, partial transposition, and realignment give, in the
ordered tuple $(x,y,q_3,a,b,c)$, the following values for the four density matrices:
\begin{align*}
\rho_1:\quad&
\left(
\frac{21}{32},
\frac{657}{2048},
\frac7{16},
\frac58,
\frac58,
\frac7{16}
\right),\\
\rho_2:\quad&
\left(
\frac13,
\frac1{54},
\frac5{81},
\frac13,
\frac13,
\frac19
\right),\\
\rho_3:\quad&
\left(
\frac{2501}{5000},
\frac{54804409}{400000000},
\frac{492797}{2000000},
\frac{12401}{20000},
\frac{12401}{20000},
\frac{737699}{2000000}
\right),\\
\rho_4:\quad&
\left(
\frac{17}{25},
\frac{17}{625},
\frac1{25},
\frac15,
\frac15,
\frac1{25}
\right).
\end{align*}

Substitution into \eqref{eq:four-defects-first}--\eqref{eq:four-defects-second}
gives the following exact table; a negative entry means detection.
\begin{equation}
\begin{array}{c|rrrr}
 &W_{\rm P}&W_{\rm E}&W_{\rm M}&W_{\rm S}\\
\hline
\rho_1&\frac{7}{1024}&-\frac1{32}&\frac{721}{4096}&-\frac{13}{128}\\[1mm]
\rho_2&-\frac4{81}&\frac29&\frac{55}{108}&\frac7{27}\\[1mm]
\rho_3&-\frac{190077}{50000000}&-\frac{2601}{1000000}
&\frac{254564409}{800000000}&\frac{4684401}{200000000}\\[1mm]
\rho_4&-\frac{264}{625}&0&-\frac4{625}&-\frac{96}{125}
\end{array}
\end{equation}

\end{proof}
The table above contains all the information we need, see \cref{fig:degree-four-comparison}:
\begin{itemize}
\item $\rho_2\in\mathcal V_{\rm P}\setminus\mathcal V_{\rm M}$
and $\rho_1\in\mathcal V_{\rm S}\setminus\mathcal V_{\rm M}$, proving
that both inclusions in \eqref{eq:only-degree-four-inclusions} are strict
\item the pair $\rho_1,\rho_2$ proves that $\mathcal V_{\rm P}$ and
$\mathcal V_{\rm S}$ are incomparable
\item the same pair $\rho_1,\rho_2$ proves that $\mathcal V_{\rm P}$ and
$\mathcal V_{\rm E}$ are incomparable
\item the pair $\rho_1,\rho_4$ proves that  $\mathcal V_{\rm M}$ and
$\mathcal V_{\rm E}$ are incomparable
\item finally, the pair $\rho_3,\rho_4$ proves that $\mathcal V_{\rm S}$ and
$\mathcal V_{\rm E}$ are incomparable.
\end{itemize}
\begin{figure}[htbp]
\centering
\begin{tikzcd}[
  row sep=5.8em,
  column sep=7.8em
]
& \mathrm P
  \arrow[dl,dotted,bend right=12,"2"']
  \arrow[dr,dotted,bend left=12,"{2,4}"]
  \arrow[dd,dotted,bend right=34,"2"']
& \\
\mathrm M
  \arrow[ru]
  \arrow[rd]
  \arrow[rr,dotted,bend left=12,"4"]
&
& \mathrm E
  \arrow[ll,dotted,bend left=12,"{1,3}"]
  \arrow[ul,dotted,bend left=12,"1"']
  \arrow[dl,dotted,bend right=12,"3"] \\
& \mathrm S
  \arrow[ul,dotted,bend left=12,"1"]
  \arrow[ur,dotted,bend right=12,"4"']
  \arrow[uu,dotted,bend right=34,"1"]
&
\end{tikzcd}
\caption{Complete ordered comparison of the four entanglement criteria.  The
only solid arrows are $\mathrm M\to\mathrm P$ and
$\mathrm M\to\mathrm S$; they encode inclusion of the detection sets. The ten dotted arrows give explicit witnesses for
all remaining ordered non-inclusions among the four states $\rho_{1,2,3,4}$.}
\label{fig:degree-four-comparison}
\end{figure}
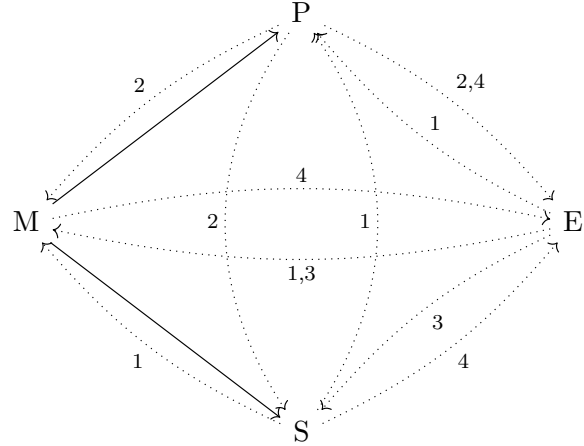

The fact that the second binomial minorant entanglement criterion is weak can also be seen to follow from the fact that the ray $\mathbb R_{\geq0}W_{\rm M}$ is not extremal in the
full cone $\ineqsep_4$.  More precisely,
$$
W_{\rm M}
=\frac14W_{\rm S}
+\frac14(t^2-r_1)(3t^2-r_1).
$$ 

\medskip

Finally, we plot the realignment moments of randomly generated qubit mixed states in \cref{fig:realignment-comparison}.  To randomly generate the $2$ qubit mixed states, we consider Haar distributed random states on $\mathbb{C}^2 \otimes \mathbb{C}^2 \otimes \mathbb{C}^k$ and trace out the last subsystem of dimension $k$.

\begin{figure}[htbp]
    \centering

    \begin{subfigure}[t]{0.48\textwidth}
        \centering
        \includegraphics[width=\linewidth]{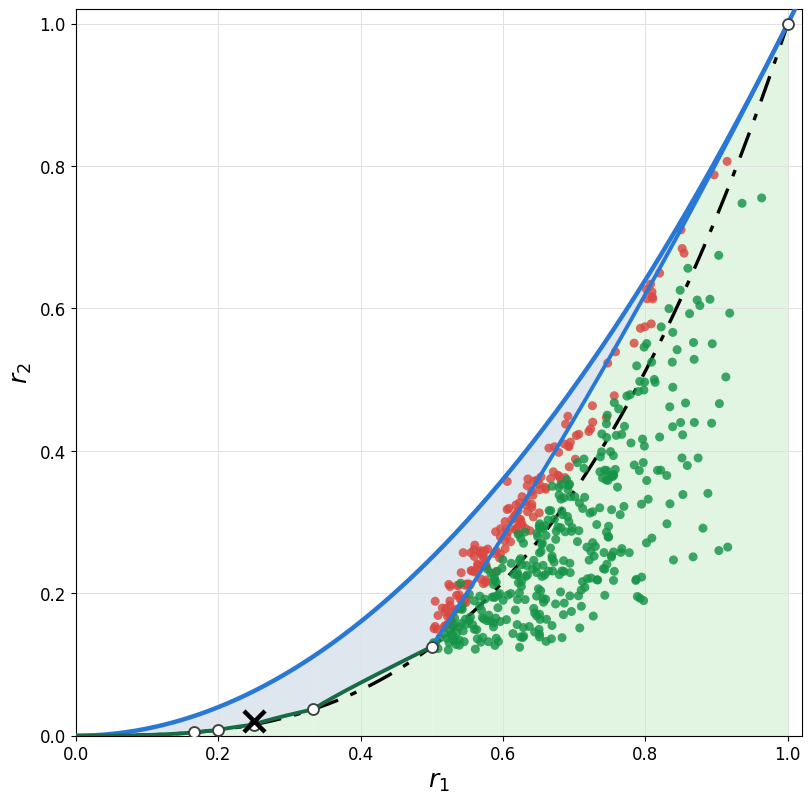}
        \label{fig:realignment-ppt}
    \end{subfigure}
    \hfill
    \begin{subfigure}[t]{0.48\textwidth}
        \centering
        \includegraphics[width=\linewidth]{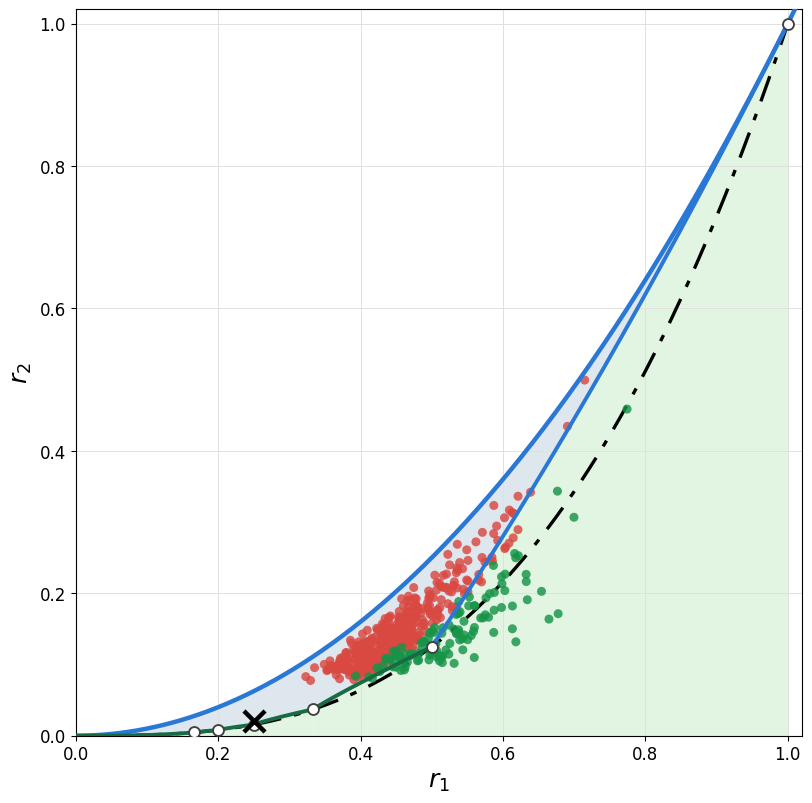}
        \label{fig:realignment-k4}
    \end{subfigure}

    \caption{The figures above show values $(r_1, r_2)$ for $1000$ randomly generated mixed quantum states of two qubits, and also for the Tiles $3 \times 3$ state (black cross). The left and the right figures have scattered $(r_1, r_2)$  for $k=2$ and $k=4$ respectively; higher $k$ implies less purity, $\operatorname{Tr}(\rho^2) =r_1$. The color of the dots denotes whether the state is detected entangled by the 3 moments PPT criterion; the green dots correspond to the detectable states while the red ones are not detected. The dots in the light green region are detected by utilizing the inequalities corresponding to the realignment moments; in particular the first arc corresponds to the inequality $W_S$. The dot-dashed black line denotes the interpolation inequality $r_2 \geq r_1^3$ which is weaker than the 3 moments PPT criterion (ref. \cref{thm:interpolation-implied-ppt3}). The Tiles state is not detectable by just using the first two moments of realignment or the 3 moments PPT criterion.}
    \label{fig:realignment-comparison}
\end{figure}

\subsection{PPT entangled state detection: the Tiles state}\label{sec:tiles}

We now turn to the case that motivates our construction: the detection of PPT-entangled states. As discussed in \cref{sec:sep}, the criteria arising from $\iota_+(\ineqpsd_n)$ and $\iota_\Gamma(\ineqpsd_n)$ depend only on the spectra of $\rho$ and $\rho^{\Gamma_B}$, respectively. Since both $\rho$ and $\rho^{\Gamma_B}$ are positive semidefinite for a PPT state, all such criteria are necessarily satisfied and hence cannot detect PPT entanglement. Therefore, to detect PPT entangled states through trace invariants, one needs genuinely bipartite elements of $\ineqsep_n$, such as those arising from moments of the realignment map.

We illustrate this mechanism using the Tiles state $\rho_{\rm N}$ \cite{bennett1999unextendible}, defined as the normalized projection onto the orthogonal complement of the five-element unextendible product basis in $\C^3\otimes\C^3$. The state $\rho_{\rm T}$ is PPT-entangled, and both the standard and centered realignment criteria detect its entanglement. To assess the detection efficiency of our proposed realignment-moment criteria, \Cref{tab:tiles} reports, for each criterion, its value at $p=1$ and the corresponding detection threshold $p^\ast$ for the noisy family
\[
\rho_{\rm N}(p)
=
p\,\rho_{\rm T}
+
(1-p)\,\frac{\text{I}_9}{9}.
\]
\begin{table}[htbp]
\centering
\small
\setlength{\tabcolsep}{6pt}
\renewcommand{\arraystretch}{1.2}
\begin{tabular}{@{}cl c l@{}}
\toprule
& \textbf{Criterion}
& \makecell[c]{\textbf{detects}\\$\rho_{\rm T}$}
& \textbf{detection range}\\
\midrule
\multirow{6}{*}{\rotatebox{90}{\textbf{standard}}}
& $\|\Rn(\rho)\|_1\leq1$ \ (realignment)                 & \yes & $p\in(0.88969,\,1]$\\[1.5pt]
\cmidrule(l){2-4}

& interpolation $(a,b)=(1,2)$, \ $r_1^3\leq r_2$         & \no  & \nodet\\
& binomial minorant, $k=52$ (minimal)                    & \yes & $p\in(0.99946,\,1]$\\
& $L^1_\beta$ minorant, $k=9$ \ ($\beta=1.35$)           & \yes & $p\in(0.98966,\,1]$\\
& $L^1_\beta$ minorant, $k=15$ \ ($\beta=1$)             & \yes & $p\in(0.92576,\,1]$\\
\midrule
\multirow{6}{*}{\rotatebox{90}{\textbf{centered}}}
& $\|\Rn(\tilde\rho)\|_1\leq\sqrt G$ \ (enhanced)        & \yes & $p\in(0.88218,\,1]$\\[1.5pt]
\cmidrule(l){2-4}

& interpolation $(a,b)=(1,2)$, \ $\tilde r_1^{\,3}\leq G\,\tilde r_2$     & \yes & $p\in(0.93296,\,1]$\\
& binomial minorant, $k=20$ (minimal)                    & \yes & $p\in(0.99507,\,1]$\\
& $L^1_\beta$ minorant, $k=5$ \ ($\beta=1.45$)           & \yes & $p\in(0.94239,\,1]$\\
& $L^1_\beta$ minorant, $k=15$ \ ($\beta=0.5$)           & \yes & $p\in(0.90813,\,1]$\\
\bottomrule
\end{tabular}
\caption{\justifying Detection of the noisy Tiles state
$\rho_{\rm N}(p)=p\,\rho_{\rm T}+(1-p) \frac{\text{I}_9}{9}$. The binomial minorants are the $f_k$
of \cref{eq:binomial-realignment-minorant} and the $L^1_\beta$ minorants the
optimizers of \cref{eq:l1beta}, $k$ being both the order and the number of
moments used.}
\label{tab:tiles}
\end{table}
The binomial hierarchy is nested: since
$f_{k+1}(x)-f_k(x)=\binom{2k}{k}4^{-k}\,x(1-x)^k\geq0$ on $[0,1]$, a state
detected at order $k$ is detected at every order $k'\geq k$. For the noisy
Tiles state, however, the first order at which it becomes effective is high,
$k=52$ with the standard moments and $k=20$ with the centered ones, and the
two-moment standard interpolation criterion does not detect $\rho_{\rm T}$ at
all. Whether such orders are intrinsic to the hierarchy or a feature of the
Tiles state is the question addressed in the next subsection.
\subsection{A new PPT entangled state of minimal Taylor order}\label{sec:ppt-state}

The high orders required for the Tiles state do not reflect an intrinsic
limitation of the realignment-moment hierarchies. To see this, we numerically
searched the PPT states of $\C^3\otimes\C^3$ for states that are detected at the
smallest Taylor order $k$ of the binomial hierarchy. Maximizing the realignment
norm over this set produces a state which can be identified in closed form, and
whose realignment spectrum consists of only two distinct nonzero singular
values:
\begin{equation}\label{eq:rho-star}
\rho_\star:=\frac14\sum_{k=1}^4\ket{\psi_k}\!\bra{\psi_k},
\qquad
\begin{aligned}
\ket{\psi_1}&=\tfrac{1}{\sqrt3}\bigl(\ket{00}+\sqrt2\,\ket{21}\bigr), &
\ket{\psi_2}&=\tfrac{1}{\sqrt3}\bigl(\ket{01}+\sqrt2\,\ket{12}\bigr),\\
\ket{\psi_3}&=\tfrac{1}{\sqrt3}\bigl(\sqrt2\,\ket{02}-\ket{10}\bigr), &
\ket{\psi_4}&=\tfrac{1}{\sqrt3}\bigl(\ket{11}-\sqrt2\,\ket{20}\bigr).
\end{aligned}
\end{equation}
In the product basis $\ket{00},\ket{01},\dots,\ket{22}$,
\[
\rho_\star=
\frac{1}{36}\begin{pmatrix}
3&0&0&0&0&0&0&3\sqrt2&0\\
0&3&0&0&0&3\sqrt2&0&0&0\\
0&0&6&-3\sqrt2&0&0&0&0&0\\
0&0&-3\sqrt2&3&0&0&0&0&0\\
0&0&0&0&3&0&-3\sqrt2&0&0\\
0&3\sqrt2&0&0&0&6&0&0&0\\
0&0&0&0&-3\sqrt2&0&6&0&0\\
3\sqrt2&0&0&0&0&0&0&6&0\\
0&0&0&0&0&0&0&0&0
\end{pmatrix}.
\]

\subsubsection*{Properties of the state}

\begin{proposition}\label{prop:rho-star}
The state $\rho_\star$ has the following properties.
\begin{enumerate}
\item $\rho_\star$ is one quarter of a rank-four orthogonal projection, and so
is its partial transpose:
\[
\rho_\star^{\Gamma_B}=\frac14\sum_{k=1}^4\ket{\varphi_k}\!\bra{\varphi_k},
\qquad
\begin{aligned}
\ket{\varphi_1}&=\tfrac{1}{\sqrt3}\bigl(\ket{01}+\sqrt2\,\ket{20}\bigr), &
\ket{\varphi_2}&=\tfrac{1}{\sqrt3}\bigl(\sqrt2\,\ket{02}+\ket{11}\bigr),\\
\ket{\varphi_3}&=\tfrac{1}{\sqrt3}\bigl(\ket{00}-\sqrt2\,\ket{12}\bigr), &
\ket{\varphi_4}&=\tfrac{1}{\sqrt3}\bigl(\ket{10}-\sqrt2\,\ket{21}\bigr).
\end{aligned}
\]
In particular $\rho_\star$ is PPT.
\item Both marginals are maximally mixed: $(\rho_\star)_A=(\rho_\star)_B=\text{I}_3/3$.
\item The nonzero singular values of $\Rn(\rho_\star)$ are $1/3$ (once) and
$1/6$ (five times), so $\|\Rn(\rho_\star)\|_1=7/6>1$ and $\rho_\star$ is entangled.
\item For the centered operator, $G(\rho_\star)=4/9$ and the nonzero singular
values of $\Rn(\tilde\rho_\star)/\sqrt{G(\rho_\star)}$ are $1/4$ (five times),
so $\|\Rn(\tilde\rho_\star)\|_1/\sqrt{G(\rho_\star)}=5/4$.
\end{enumerate}
\end{proposition}

\subsubsection*{Detection by few realignment moments}
Let $s_i$ denote the singular values of $\Rn(\rho_\star)$ and $\tilde s_i$ those
of $\Rn(\tilde\rho_\star)/\sqrt{G(\rho_\star)}$. For the binomial hierarchy one
finds
\[
\sum_i f_{29}\bigl(s_i^2\bigr)=0.9943<1,
\qquad
\sum_i f_{30}\bigl(s_i^2\bigr)=1.0011>1,
\]
so that the binomial hierarchy certifies entanglement already at Taylor order
$k=30$ in the standard case, and
\[
\sum_i f_{12}\bigl(\tilde s_i^2\bigr)=0.9774<1,
\qquad
\sum_i f_{13}\bigl(\tilde s_i^2\bigr)=1.0007>1
\]
in the centered case, where the criterion is thus violated from $k=13$ on.
With only two moments, the standard interpolation criterion is satisfied,
since $r_1(\rho_\star)^3=1/64<7/432=r_2(\rho_\star)$, whereas the centered one is
violated: the normalized centered moments are $5/16$ and $5/256$, and
$(5/16)^3/(5/256)=25/16>1$. The orders required for certification are thus
substantially lower than those obtained for the noisy Tiles state, and in the
centered case two moments already suffice.

\subsubsection*{Robustness to noise}
For $p\in[0,1]$ let $\rho_\star(p):=p\,\rho_\star+(1-p)\,\text{I}_9/9$. Its
partial transpose has eigenvalues $p/4+(1-p)/9$ and $(1-p)/9$, so
$\rho_\star(p)$ is PPT for every $p$. Its marginals remain $\text{I}_3/3$, the
nonzero singular values of $\Rn(\rho_\star(p))$ are $1/3$ and $p/6$ (five
times), and those of $\Rn(\tilde\rho_\star(p))/\sqrt G$ are $p/4$ (five times).
Hence both the realignment and the enhanced realignment criteria detect
$\rho_\star(p)$ exactly for $p>4/5$, and the centered interpolation criterion,
which reads $25p^2/16>1$, detects exactly the same range. The standard
interpolation criterion, on the contrary, is satisfied for every $p$: from the
singular values above,
\[
r_1\bigl(\rho_\star(p)\bigr)=\frac{4+5p^2}{36},
\qquad
r_2\bigl(\rho_\star(p)\bigr)=\frac{16+5p^4}{1296},
\]
and therefore
\[
r_2-r_1^{\,3}
=\frac{512-240p^2-120p^4-125p^6}{46656}>0
\qquad (0\leq p\leq1).
\]
\Cref{tab:rho-star} collects the detection ranges.

\begin{table}[htbp]
\centering
\small
\setlength{\tabcolsep}{6pt}
\renewcommand{\arraystretch}{1.2}
\begin{tabular}{@{}cl c l@{}}
\toprule
& \textbf{Criterion}
& \makecell[c]{\textbf{detects}\\$\rho_\star$}
& \textbf{detection range}\\
\midrule
\multirow{3}{*}{\rotatebox{90}{\footnotesize\textbf{standard}}}
& $\|\Rn(\rho)\|_1\leq1$ (realignment)                 & \yes & $p\in(4/5,\,1]$\\
\cmidrule(l){2-4}
& interpolation, $r_1^3\leq r_2$                       & \no  & \nodet\\
& binomial minorant, $k=30$ (minimal)                  & \yes & $p\in(0.99894,\,1]$\\
\midrule
\multirow{3}{*}{\rotatebox{90}{\footnotesize\textbf{centered}}}
& $\|\Rn(\tilde\rho)\|_1\leq\sqrt G$ (enhanced)        & \yes & $p\in(4/5,\,1]$\\
\cmidrule(l){2-4}
& interpolation, $\tilde r_1^{\,3}\leq G\,\tilde r_2$  & \yes & $p\in(4/5,\,1]$\\
& binomial minorant, $k=13$ (minimal)                  & \yes & $p\in(0.99958,\,1]$\\
\bottomrule
\end{tabular}
\caption{\justifying Detection of $\rho_\star(p)=p\,\rho_\star+(1-p)\,\text{I}_9/9$
(compare with \cref{tab:tiles}). The binomial minorants are the $f_k$ of
\cref{eq:binomial-realignment-minorant}, $k$ being both the order and the
number of moments used; ``minimal'' is the smallest order of the nested
binomial hierarchy that detects $\rho_\star$. The two interpolation criteria
use the first two (standard, resp.\ centered) realignment moments.}
\label{tab:rho-star}
\end{table}

In summary, the state $\rho_\star$ shows that the $k$-Taylor hierarchy can
certify PPT entanglement at much lower orders than suggested by the Tiles state
($k=30$ instead of $52$ in the standard case, $k=13$ instead of $20$ in the
centered case), and that the centered interpolation criterion
$r_1(\tilde\rho)^3\leq G\,r_2(\tilde\rho)$, built from only two centered
realignment moments, can reproduce the full detection range of the enhanced
realignment criterion.

\bigskip

\noindent\textbf{Acknowledgments.} This research was supported by the ANR project \href{https://www.ceremade.dauphine.fr/dokuwiki/anr-tagada:start}{TAGADA}, grant number ANR-25-CE40-5672. A.G.~received support from the University Research School EUR-MINT
(State support managed by the National Research Agency for Future Investments
program bearing the reference ANR-18-EURE-0023). B.M. acknowledges partial financial support from the NanoX grant (ANR-17-EURE-0009), in the framework of the "Programme des Investissements d’Avenir", for his research visit to CNRS, Toulouse, France.

\medskip

\noindent\textbf{AI use statement.} The main idea and results of this work were developed by the authors. Large language models suggested the $L^1$-weighted minorants and some ideas in the corresponding proofs, in particular the reference \cite{sakai2017sharp} in Appendix~\ref{sec:proof-1-norm}. Coding models were used to assist with the numerical simulations. The authors take full responsibility for the content of this publication.

\bibliographystyle{alpha}
\bibliography{references}

\appendix

\section{The $L^1_\beta$-weighted hierarchy}
\label{sec:l1beta}
Recall that $g_{m,\beta}$ is defined by \cref{eq:l1beta}, namely
\[
g_{m,\beta}\in\operatorname*{arg\,min}_{h\in\mathcal H_m}
E_\beta^{(1)}(h),
\qquad
E_\beta^{(1)}(h):=\int_0^1x^{-\beta}\bigl(\sqrt{x}-h(x)\bigr)\,\dd x.
\]

The feasible set of polynomials from \cref{eq:Hm} satisfy $h(x)\leq\sqrt{x}$ for all $0\leq x\leq1.$
By the Markov--Luk\'acs theorem \cite{blekherman2012semidefinite} or the one-dimensional Positivstellensatz, the feasible polynomials have a representation as a sum of squares
$$
t-h(t^2)=Q(t)+t(1-t)S(t),
$$
for polynomials $Q$ and $S$ which are sums of squares of degrees $2m$, resp~$2m-2$. The restriction $\beta<3/2$, is forced by
integrability of the objective: since $\sqrt{x}-h(x)\sim\sqrt{x}$ as $x\downarrow0$, the
weighted integrand behaves like $x^{1/2-\beta}$, which is integrable on
$(0,1]$ exactly when $\beta<3/2$.

\begin{proposition}
\label{prop:l1beta-sdp}
Fix $k\ge1$ and $0\le\beta<3/2$. On the feasible set (i.e.\ $h(x)=\sum_{j=1}^k a_j x^j$ with
$h(1)=1$ and $t-h(t^2)=Q(t)+t(1-t)S(t)$ for sums of squares $Q,S$),
$E_\beta^{(1)}$ is \emph{linear} in the coefficient vector
$a=(a_1,\dots,a_k)$, with the explicit closed form
\begin{equation}
  E_\beta^{(1)}(h) \;=\; \frac{2}{3-2\beta}\;-\;\sum_{j=1}^{k}
  \frac{1}{j+1-\beta}\,a_j .
  \label{eq:l1energy-closed}
\end{equation}
Hence the problem
\[
  \min_h\ \Bigl\{\, E_\beta^{(1)}(h)\ :\ h(x)=\textstyle\sum_{j=1}^k a_j x^j,\
  \textstyle\sum_j a_j=1,\ t-h(t^2)=Q(t)+t(1-t)S(t),\ Q,S \text{SOS} \Bigr\}
\]
is equivalent to the semidefinite program
\begin{equation}
  \max_{a,\,Q,\,S}\ \sum_{j=1}^{k}\frac{1}{j+1-\beta}\,a_j
  \quad\text{s.t.}\quad
  \textstyle\sum_j a_j = 1,\quad
  t-\textstyle\sum_j a_j t^{2j}=Q(t)+t(1-t)S(t),
  \label{eq:l1beta-sdp}
\end{equation}
with $Q,S$ sums of squares.
\end{proposition}

\begin{proof}
Feasibility gives $F(t)=t-h(t^2)\ge0$ on $[0,1]$, hence $\sqrt{x}-h(x)\ge0$ on $[0,1]$.
For the closed form, substitute $x=t^2$: with $F(t)=t-h(t^2)=t-\sum_j a_j
t^{2j}$,
\[
  E_\beta^{(1)}(h) = \int_0^1 x^{-\beta}\bigl(\sqrt x - h(x)\bigr)\,dx
  = \int_0^1 2t^{1-2\beta} F(t)\,dt
  = 2\int_0^1 t^{2-2\beta}\,dt
    - 2\sum_{j=1}^k a_j \int_0^1 t^{2j+1-2\beta}\,dt .
\]
Evaluating the two elementary integrals,
$\int_0^1 t^{2-2\beta}dt = 1/(3-2\beta)$ and
$\int_0^1 t^{2j+1-2\beta}dt = 1/(2j+2-2\beta)$, gives
\[
  E_\beta^{(1)}(h) = \frac{2}{3-2\beta} - \sum_{j=1}^k
  \frac{2}{2j+2-2\beta}\,a_j
  = \frac{2}{3-2\beta} - \sum_{j=1}^k \frac{a_j}{j+1-\beta},
\]
which is Eq.~\eqref{eq:l1energy-closed}. Since $h(x)=\sum_j a_jx^j$ enters
$E_\beta^{(1)}$ linearly and the constant term $2/(3-2\beta)$ does not
depend on $a$, minimizing $E_\beta^{(1)}$ over the feasible set is
equivalent to maximizing $\sum_j a_j/(j+1-\beta)$, which is
Eq.~\eqref{eq:l1beta-sdp}.
\end{proof}

The first two members of the two families coincide, $h_{1,\beta}=f_1$ and
$h_{2,\beta}=f_2$ for every $\beta$, where
$$f_1(x) = x \quad f_2(x) = -\frac{1}{2}x^2 +\frac{3}{2}x$$

are also the minorants obtained from the binomial/Taylor expansion. From $k=3$, the $L^1_\beta$ minorants
are markedly better near the origin, in contrast to the binomial ones.  For $\beta=1$ the first few are listed below.
\begin{gather}
\label{eq:l1beta-explicit}
h_{3,1}(x)=\Bigl(\tfrac14+\tfrac{3\sqrt3}{2}\Bigr)x
-\Bigl(\tfrac32+\tfrac{3\sqrt3}{2}\Bigr)x^2+\tfrac94x^3,
\qquad
h_{4,1}(x)=\tfrac{185}{54}x-\tfrac{137}{18}x^2+\tfrac{76}{9}x^3-\tfrac{88}{27}x^4,
\\
h_{5,1}(x)\approx 4.500062\,x-17.929221\,x^2+40.111576\,x^3
-41.286843\,x^4+15.604426\,x^5,\notag
\end{gather}

Note that $h_{3,1}$ touches $\sqrt x$ at $x=\tfrac13$, and $h_{4,1}$ touches it at
$x=\tfrac14$ and tangentially at $x=1$.  For $\beta=0$ one finds instead
$h_{3,0}(x)=(1+\sqrt2)x-(4-\sqrt2)x^2+(4-2\sqrt2)x^3$, with contact point
$x=\tfrac12$: decreasing $\beta$ pushes the contact points towards $x=1$ as shown in \cref{fig:minorants}.

\section{Proof of the 1 norm optimization theorem}\label{sec:proof-1-norm}

\begin{proof}[Proof of \cref{thm:r1-r2-exact}]
For $r \in \mathbb{R}^2_+ / \{0,0\}$, we need to solve the following problem:
$$\alpha(r_1,r_2) = \min_{n \in \mathbb{N}}\min_{s \in [0,1]^n} \left\{\sum^{n}_{i=1} s_i \mid \sum_{i} s^2_i = r_1, \sum_{i} s^4_i = r_2\right\}.$$
By substituting $s_i = \sqrt{t_i}$, we can reformulate the problem as
$$\alpha(r_1,r_2) = \min_{n \in \mathbb{N}}\min_{t \in [0,1]^n} \left\{\sum^{n}_{i=1} \sqrt{t_i} \mid \sum_{i} t_i = r_1, \sum_{i} t^2_i = r_2\right\}.$$
Let $\Delta$ denote the set of all discrete valued probability distributions, and $\|\cdot \|_p$ denote the $\ell^p$ norms. Then, by normalizing, we obtain the formulation of the problem in terms of $p$-norms
$${\alpha(r_1,r_2)} = \sqrt{r_1}\min_{p \in \Delta}\left\{\|p\|^{1/2}_{1/2} \mid \|p\|^2_2 =  r_2/r_1^2\right\}.$$

For $x \in (0,1)$, let us consider the following probability distribution:
$$w(x) := (\underbrace{x,x \ldots x}_{\lfloor 1/x \rfloor}, 1-\lfloor 1/x \rfloor x).$$

It follows from \cite[Theorem 1]{sakai2017sharp} that for all probability vectors such that $\|p\|^2_{2} = r_2/r_1^2$, the $\ell_{1/2}$ norm attains its minumum value at the vector $w(x)$ where $x$ is chosen such that $\|w(x)\|_2^2= \gamma$ for $\gamma := r_2/r_1^2$:
$$\|p\|_{1/2} \geq \|w(x)\|_{1/2} \qquad \forall p \in \Delta \, \text{ s.t. } \|p\|^2_{2} = r_2/r_1^2.$$

Moreover, the function $x \mapsto \|w(x)\|_{2}$ is strictly monotonic in $x$, and hence can be inverted to obtain $x$ \cite{sakai2017sharp}. Using the results in \cite{sakai2017sharp}, the inverse is exactly given by

$$x=\frac{m+\sqrt{\gamma\,m\,(1+m-\gamma^{-1})}}{m(1+m)},\qquad \text{ with } m=\lfloor 1/\gamma\rfloor.$$
Note that we have
$$\alpha(r_1,r_2) = \sqrt{r_1}\|w(x)\|^{1/2}_{1/2} = \sqrt{r_1}\left(\left\lfloor \frac{1}{x} \right\rfloor \sqrt x + \sqrt{1-\left\lfloor \frac{1}{x} \right\rfloor x}\right).$$

Finally, to show the formula for $\alpha(r)$ in the theorem, we need to understand these functions piecewise. Firstly, note that if $r_1^2/r_2 = \frac{1}{\gamma} = n \in \N$, $x = \frac{1}{n}$ and thus $w(x) = (1/n, \ldots, 1/n)$. We have in this case $\alpha(r_1,r_2) = \sqrt{r_1}\sqrt{n} = \frac{r^{3/2}_1}{r^{1/2}_2}$.

Conisder now $\gamma \in (\frac{1}{m+1}, \frac{1}{m})$. Then, we have $\lfloor 1/\gamma\rfloor = m$.  Also note that, due to the monotonicity property, $\lfloor\frac{1}{x}\rfloor = m$, while $\lceil r_1^2/r_2\rceil = m+1 = |\operatorname{supp}(w(x))|$.

Finally, setting as in the statement $n:=m+1$, we have
$$
\Delta=\sqrt{n r_2-r_1^2}
      =\sqrt{r_2}\sqrt{n-\gamma^{-1}}
      \qquad \text{ and } \qquad
x=\frac{m+\sqrt{m}\,\Delta/r_1}{m n},
$$
so that
$$
m x=\frac{m r_1+\sqrt{m}\,\Delta}{n r_1}
\qquad \text{ and } \qquad
1-m x=\frac{r_1-\sqrt{m}\,\Delta}{n r_1}.
$$
Hence
$$
\alpha(r_1,r_2)
=\sqrt{r_1}\bigl(m\sqrt{x}+\sqrt{1-m x}\bigr)
=\frac{\sqrt{m^2 r_1+m^{3/2}\Delta}
+\sqrt{r_1-\sqrt{m}\,\Delta}}{\sqrt n},
$$
which is the claimed result. For $r_1=r_2 = 0$, we can see easily that $\alpha(r)=0.$

\end{proof}

\end{document}